\documentclass{article}
\usepackage{arxiv}
\pdfoutput=1

\usepackage[utf8]{inputenc} 
\usepackage[T1]{fontenc}    
\usepackage{url}            
\usepackage{booktabs}       
\usepackage{amsfonts}       
\usepackage{nicefrac}       
\usepackage{microtype}      

\usepackage{amsthm,mathtools}
\usepackage{physics,amsbsy,amssymb,amscd,amsfonts,latexsym,amstext,delarray,amsmath,graphicx,color,caption,braket,appendix,bbm}
	\newcommand{\Z}{\mathbb{Z}}
	\newcommand{\R}{\mathbb{R}}

\usepackage{hyperref}
\hypersetup{
    colorlinks=true,
    linkcolor=blue,
    filecolor=magenta,
    urlcolor=cyan,
}
\newcommand{\floor}[1]{\lfloor #1 \rfloor}

\newtheorem{thm}{Theorem}

\newtheorem{prop}[thm]{Proposition}

\newcommand{\ihat}{\hat{\imath}}
\newcommand{\jhat}{\hat{\jmath}}
\usepackage{xcolor}

\title{Geometrically Interpreting Higher Cup Products, and Application to Combinatorial Pin Structures}

\author{
 Sri Tata \\
  Department of Physics\\
  University of Maryland, College Park\\
  College Park, MD 20740 \\
  \texttt{stata@umd.edu} \\
}

\begin{document}
\maketitle

\begin{abstract}
We provide a geometric interpretation of the formulas for Steenrod's $\cup_i$ products, giving an explicit construction for a conjecture of Thorngren. We construct from a simplex and a branching structure a special frame of vector fields inside each simplex that allow us to interpret cochain-level formulas for the $\cup_i$ as a generalized intersection product on the dual cellular decomposition. It can be thought of as measuring the intersection between a collection of dual cells and thickened, shifted version of another collection, where the vector field frame determines the thickening and shifting. Defining this vector field frame in a neighborhood of the dual 1-skeleton of a simplicial complex allows us to combinatorially define $Spin$ and $Pin^\pm$ structures on triangulated manifolds. We use them to geometrically interpret the `Grassmann Integral' of Gu-Wen/Gaiotto-Kapustin, without using Grassmann variables. In particular, we find that the `quadratic refinement' property of Gaiotto-Kapustin can be derived geometrically using our vector fields and interpretation of $\cup_i$, together with a certain trivalent resolution of the dual 1-skeleton. This lets us extend the scope of their function to arbitrary triangulations and explicitly see its connection to spin structures. Vandermonde matrices play a key role in all constructions.
\end{abstract}


\section{Introduction}
The Steenrod operations and higher cup products, $\cup_i$ are an important part of algebraic topology, and have recently been emerging as a critical tool in the theory of fermionic quantum field theories. They were invented by Steenrod \cite{SteenrodCupPaper} in the study of homotopy theory. More recently, they have made a surprising entrance in the theory of fermionic and spin TQFTs in the study of Symmetry-Protected Topological (SPT) phases of matter \cite{GuWen,GaiottoKapustin,BrumfielMorgan}. As such, it would be desirable to give them a geometric interpretation beyond their mysterious cochain formulas, in a similar way that the regular cup product, $\cup_0$, can be interpreted as an intersection product between cells and a shifted version of the other cells. 

In this note, we will show that in fact, there is such an interpretation as a generalized intersection product, which gives the intersection class of the cells dual to a cochain with a \textit{thickened} and shifted version of the other's cells. Similar interpretations for the Steenrod squares (the maps $\alpha \mapsto \alpha \cup_i \alpha$) as self-intersections from immersions have been shown \cite{EcclesGrant} and are related to classical formulas of Wu and Thom. However, a more general interpretation of the $\cup_i$ products has still not been demonstrated. This interpretation was conjectured by Thorngren in \cite{Thorngren2018Thesis} that describes the $\cup_i$ product as an intersection from an $i$-parameter thickening with respect to $i$ vector fields. Such vector fields will be referred to as `Morse Flows'.

We will verify the conjecture by giving an explicit construction of a set of such $n$ vector fields inside each $n$-simplex. Thickening the Poincaré dual cells with respect to the first $i$ fields and shifting with respect to the next field will show us that the cells that intersect each other with respect to these fields are the exact pairs that appear in Steenrod's formula for $\cup_i$. In the section \ref{sec:Prelim}, we review some convenient ways to describe and parameterize the Poincaré dual cells, which we will use extensively throughout the note. While this material is standard, it would be helpful to skim through it to review our notation. In Section \ref{CupZeroSection}, we will warm up by reviewing how the intersection properties of the $\cup_0$ product's formulas can be obtained from a vector field flow. In Section \ref{sec:HigherCupProductMain}, we will start by reviewing the definitions of the higher cup formulas and the Steenrod operations. Then we'll provide some more motivation, given the $\cup_0$ product's interpretation, as to why the thickening procedure should seem adequate to describe the higher cup products. Then, we'll describe the thickening procedure, state more precisely our main proposition about the higher cup formula in Section \ref{mainPropCup_m}, and prove it in Sections \ref{mainPropCup_m}-\ref{proofOfMainPropCup_m}. 
The main calculation is in Section \ref{proofOfMainPropCup_m}. In principle the main content is Sections \ref{definingThickenedCells}-\ref{proofOfMainPropCup_m}, and the rest of Sections \ref{CupZeroSection}-\ref{sec:HigherCupProductMain} are there to build intuition for the construction.
Throughout, we only work with $\Z_2$ coefficients.

After talking about the higher cup products, we will show how our interpretation can be applied to interpreting the `Grassmann integral' of Gu-Wen/Gaiotto-Kapustin \cite{GuWen,GaiottoKapustin}, which we'll call the ``GWGK Grassmann Integral" or simply the ``GWGK Integral". In Section \ref{backgroundPropertiesOfGuWenGrassmannIntegral}, we review some background material on $Spin$ structures and Stiefel-Whitney classes on a triangulated manifold, as well as the formal properties of the GWGK Integral we set out to reproduce. In Section \ref{geometricGuWenIn2D}, we review how the GWGK Integral can be defined geometrically in 2D with respect to the vector fields we constructed before and a loop decomposition of a $(d-1)$-cocycle on the dual 1-skeleton. And in Section \ref{geometricGuWenInHigherDimensions}, we extend this understanding to higher dimensions. The interpretation of the higher cup product makes its application in Section \ref{verifyingTheFormalPropertiesHigherDimensions} in demonstrating the `quadratic refinement' property of our construction.  

\section*{\LARGE{\underline{\textbf{Interpreting Higher Cup Products}}}}
\addcontentsline{toc}{section}{\LARGE{\underline{\textbf{Interpreting Higher Cup Products}}}}

\section{Preliminaries} \label{sec:Prelim}
It will be helpful to review how Poincaré duality looks on the standard $n$-simplex, $\Delta^n$. Recall that

\begin{equation}
\Delta^n = \{ (x_0,\dots,x_n) \in \R^{n+1} | x_0 + \dots + x_n = 1 , x_i \ge 0 \text{ for all i} \}
\end{equation}

In particular, we'll review and write out explicit formulas parameterizing the cells in the dual cellulation of $\Delta^n$ and how they are mapped to their cochain partners.

\subsection{Cochains}
Recall that we are working with $\Z_2$-valued chains and cochains. If we fix $\alpha$ to be a $p$-cochain, then $\alpha$ restricted to $\Delta^n$ will manifest itself as a function from the set of size-$(p+1)$ subsets of $\{0,\dots,n\}$ to $\Z_2$. In other words

\begin{equation}
\alpha(i_0,\dots,i_p) \in \{0,1\} = \Z_2, \text{ where } 0 \le i_0 < \dots < i_p \le n
\end{equation}

Note that there are $2^{{n+1}\choose{p+1}}$ distinct p-cochains on $\Delta^n$, since there are ${{n+1}\choose{p+1}}$ choices of $\{i_0<\dots<i_p\} \subset \{0,\dots,n\}$ and two choices of the value of each $\alpha(i_0,\dots,i_p)$.

The `coboundary' of a $p$-cochain $\alpha$ is a $(p+1)$-cochain $\delta \alpha$ defined by

\begin{equation}
\delta\alpha (i_0,\dots,i_{p+1}) = \sum_{j=0}^{p+1} \alpha(0,\dots,\ihat_j,\dots,i_{p+1}) 
\end{equation}

where $\ihat_j$ refers to skipping over $i_j$ in the list. We say $\alpha$ is `closed' if $\delta \alpha = 0$ everywhere, which means modulo 2 that at each simplex, $\alpha = 1$ on an even number of $p$-subsimplices. We say $\alpha$ is `exact' if $\alpha = \delta \lambda$ for some $\lambda$.

\subsection{The dual cellulation}
Now, let us review how to construct the dual cellulation of $\Delta^n$. For clarity, let's first look at the case $n=2$ before writing the formulas in general dimensions.

\subsubsection{Example: The 2-simplex}
The two simplex is $\Delta^2 = \{ (x_0,x_1,x_2) \in \R^{3} | x_0 + x_1 + x_2 = 1 , x_i \ge 0 \text{ for all i} \}$. The 'barycentric subdivision' is generated by the intersections of the planes $\{x_0 = x_1, x_0 = x_2, x_1 = x_2\}$ with $\Delta^2$, as shown in Figure(\ref{fig:2_Simplex1}). The Poincaré dual cells are made from a certain subset of the cells of the barycentric subdivision, indicated pictorially in Figures(\ref{fig:2_Simplex2}, \ref{fig:2_Simplex3}).

\begin{figure}[h!]
  \centering
  \begin{minipage}{0.44\textwidth}
    \centering
    \includegraphics[width=\linewidth]{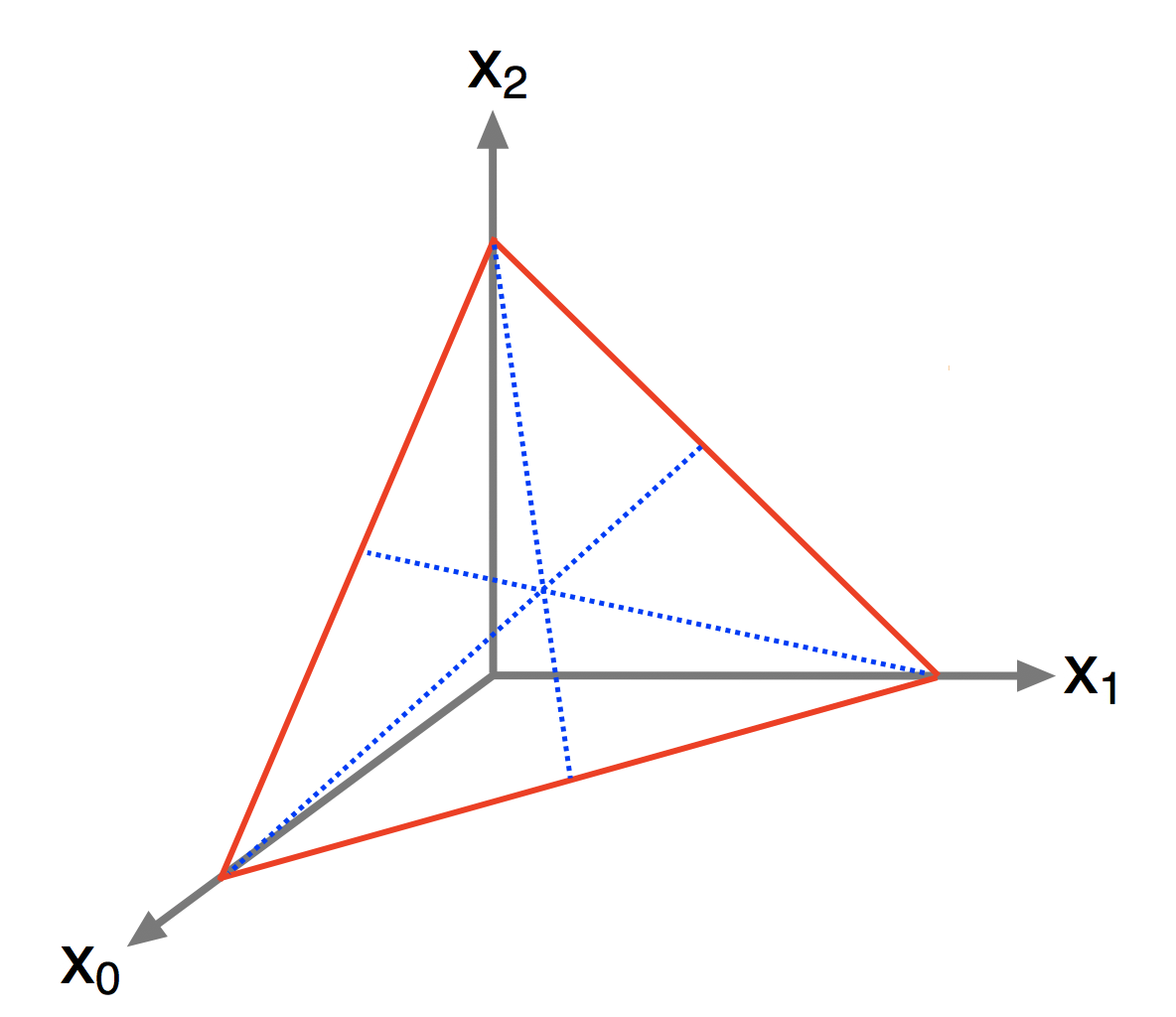}
    \caption{The standard 2-simplex $\Delta^2$. The blue, dashed lines are the barycentric subdivision of $\Delta^2$, obtained from the intersections of the planes $\{x_0 = x_1, x_0 = x_2, x_1 = x_2\}$ with $\Delta^2$}
    \label{fig:2_Simplex1}
  \end{minipage} \quad \quad \quad
  \begin{minipage}{0.44\textwidth}
    \centering
    \includegraphics[width=\linewidth]{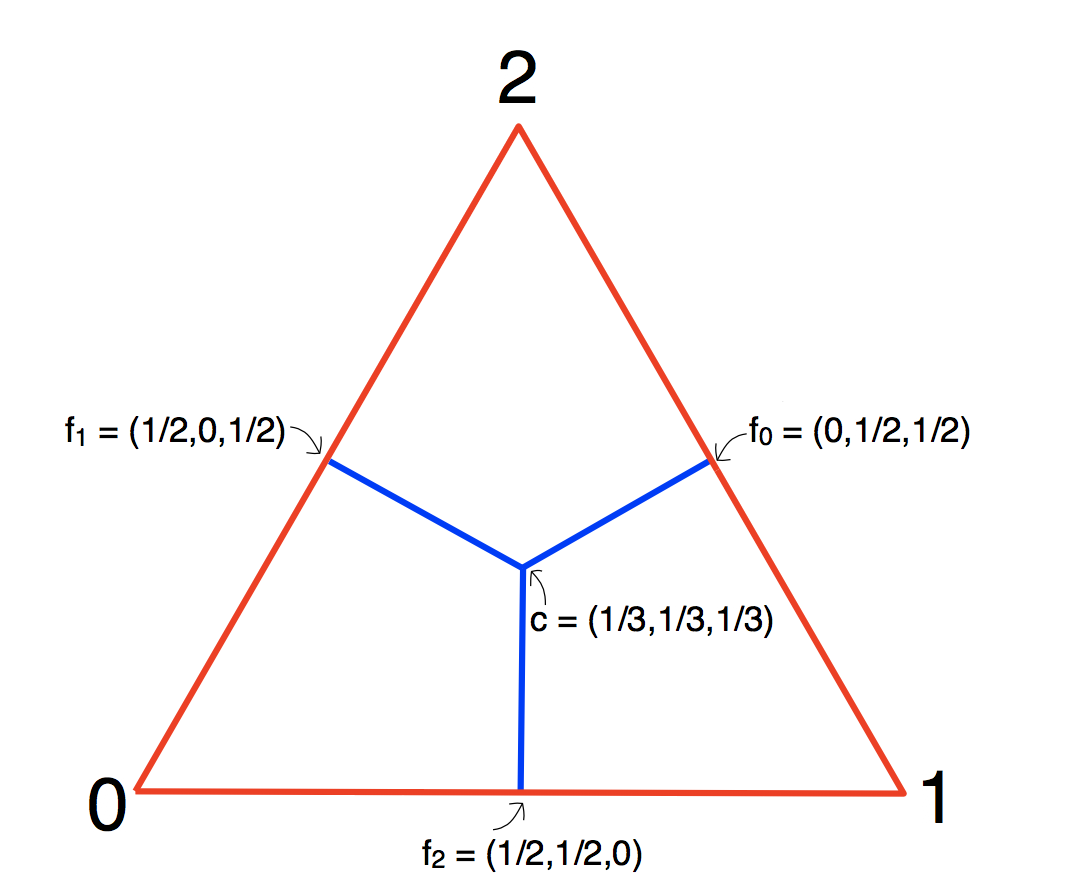}
    \caption{The Poincaré dual cellulation, whose 1-skeleton is in blue. We'll be able to express the dual cells in terms of the points $f_0,f_1,f_2,c$.}
    \label{fig:2_Simplex2}
  \end{minipage}
\end{figure}

\begin{figure}[h!]
  \centering
  \includegraphics[width=0.4\linewidth]{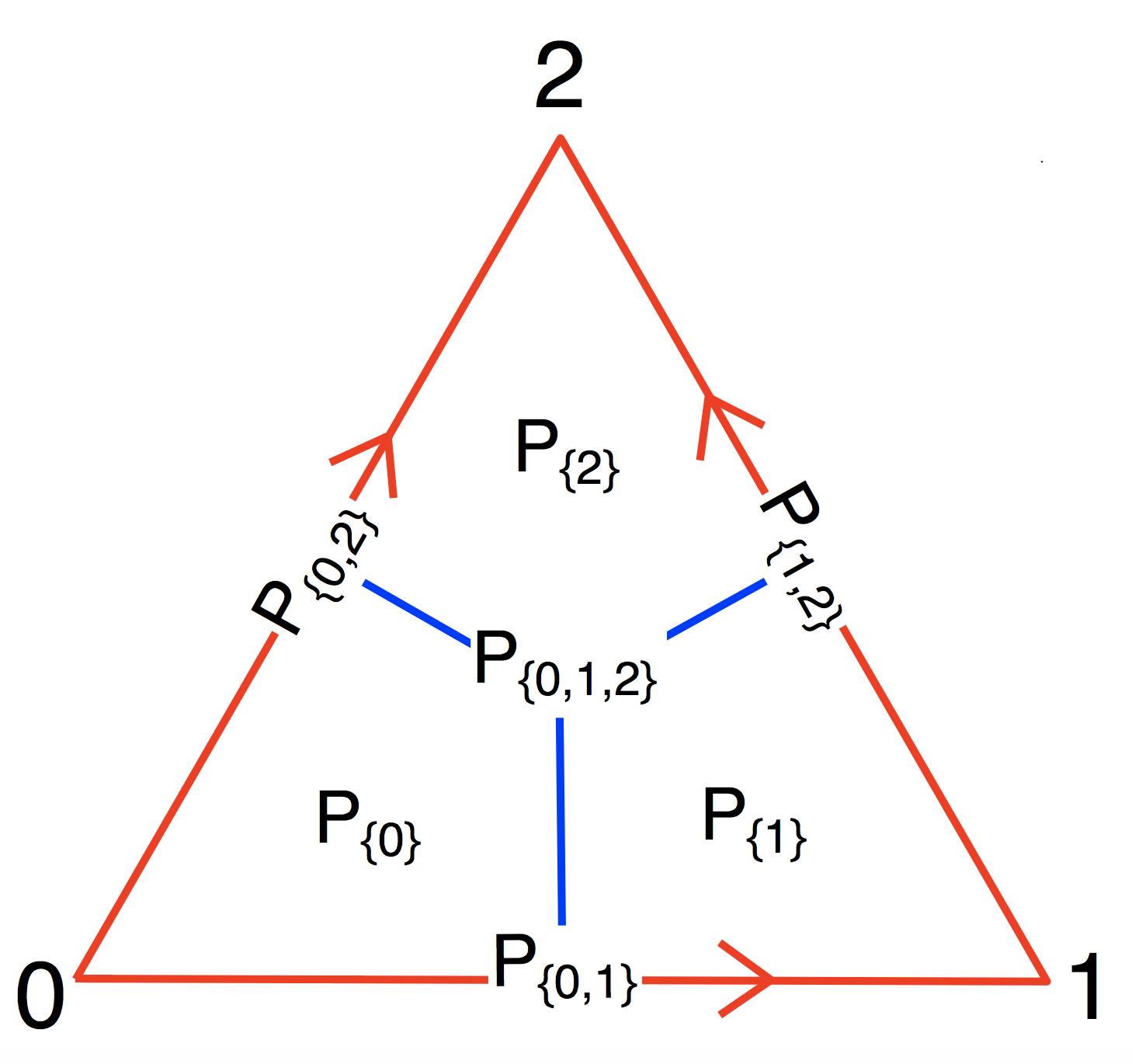}
  \caption{The cells in the Poincaré dual cellulation of $\Delta^2$}
  \label{fig:2_Simplex3}
\end{figure}

Let us now list all the cells in the Poincaré dual decomposition of $\Delta^2$. It is first helpful to define 4 points: $\{c, f_0, f_1, f_2\} \in \Delta^2$. Here, 
\begin{equation}
c = \{1/3,1/3,1/3\}
\end{equation}

is the coordinate of the barycenter of $\Delta^2$. And, 

\begin{equation}
\begin{split}
f_0 &= \{0,1/2,1/2\} \\
f_1 &= \{1/2,0,1/2\} \\ 
f_2 &= \{1/2,1/2,0\}
\end{split}
\end{equation}

are the barycenters of the boundary $1$-simplices of $\Delta^2$ respectively. We denote by $f_i$ the barycenter of the $1$-simplex opposite to the point $v_i$, where $v_i \in \{(1,0,0),(0,1,0),(0,0,1)\}$ is the point on the $x_i$-axis on $\Delta^2$.

There is one 0-cell $P_{\{0,1,2\}}$ which consists of only the center point $c$, 

\begin{equation}
P_{\{0,1,2\}} = \{c\}
\end{equation}

There are three 1-cells $P_{\{0,1\}},P_{\{0,2\}},P_{\{1,2\}}$, which consist of the intersection of $\Delta^2$ with the rays going from $c$ to $f_2, f_1, f_0$. In other words

\begin{equation}
\begin{split}
  P_{\{0,1\}}&=\Delta^2 \cap \{c + (f_2-c)t | t \ge 0\} \\
  P_{\{0,2\}}&=\Delta^2 \cap \{c + (f_1-c)t | t \ge 0\} \\
  P_{\{1,2\}}&=\Delta^2 \cap \{c + (f_0-c)t | t \ge 0\}
\end{split}
\end{equation}

And, there are three 2-cells, $P_{\{0\}},P_{\{1\}},P_{\{2\}}$, which consist of the points 

\begin{equation}
\begin{split}
  P_{\{0\}}=\Delta^2 \cap \{c + (f_1-c)t_1 + (f_2-c)t_2 | t_1, t_2 \ge 0\} \\
  P_{\{1\}}=\Delta^2 \cap \{c + (f_0-c)t_1 + (f_2-c)t_2 | t_1, t_2 \ge 0\} \\
  P_{\{2\}}=\Delta^2 \cap \{c + (f_0-c)t_1 + (f_1-c)t_2 | t_1, t_2 \ge 0\} 
\end{split}
\end{equation}

The reason we chose to name the cells this way was to make it clearer the relationship between the cochains and their dual chains. The statement is that the p-cochain $\alpha$ is dual to the union of the chains under which $\alpha$ doesn't vanish, i.e. $\alpha$ is dual to $\bigcup \{P_{\{i_0,\dots,i_p\}} | \alpha(i_0,\dots,i_p) = 1\}$.

Above, we have given an explicit parametrization of the cells $P_{I}, I \subset \{0,\dots,n\}$. But, it will also be helpful for us to express them in another way. One can easily check that the 1-cells can be written as:

\begin{equation}
\begin{split}
  P_{\{0,1\}}&=\Delta^2 \cap \{(x_0,x_1,x_2) | x_0 = x_1 \ge x_2 \} \\
  P_{\{0,2\}}&=\Delta^2 \cap \{(x_0,x_1,x_2) | x_0 = x_2 \ge x_1 \} \\
  P_{\{1,2\}}&=\Delta^2 \cap \{(x_0,x_1,x_2) | x_1 = x_2 \ge x_0 \}
\end{split}
\end{equation}

In words, $P_{\{i,j\}}$ is where the plane $x_i=x_j$ intersects $\Delta^2$, but restricted to those points where $x_i, x_j$ are greater than or equal to the other coordinates.

And, the 2-cells can be similarly written as:

\begin{equation}
\begin{split}
  P_{\{0\}}=\Delta^2 \cap \{(x_0,x_1,x_2) | x_0 \ge x_1, x_0 \ge x_2 \} \\
  P_{\{1\}}=\Delta^2 \cap \{(x_0,x_1,x_2) | x_1 \ge x_0, x_1 \ge x_2 \} \\
  P_{\{2\}}=\Delta^2 \cap \{(x_0,x_1,x_2) | x_2 \ge x_0, x_2 \ge x_1 \}
\end{split}
\end{equation}

\subsubsection{General dimensions}
We can see general patterns for the dual cell decompositions in $n$ dimensions. 

Just as before, we can define the points $c$, which is the barycenter of $\Delta^n$ and $f_0,\dots,f_n$ which are the barycenters of the $(n-1)-$simplices that are opposite to the points $v_i = (0,\dots,0,1,0,\dots,0)$ on the $x_i$ axis in $\Delta^n$. Explicitly, we'll have that the coordinates of these points are

\begin{equation}
c = \frac{1}{n+1}(1,...,1)
\end{equation}

which comes from setting $x_0 = x_1 = \dots = x_n$ and $\sum_j x_j = 1$. 

And, we'll have

\begin{equation}
\begin{split}
f_i &= ((f_i)_0,\dots,(f_i)_n), \text{ where} \\
(f_i)_j &= 
\begin{cases}
  1/n, & \text{if } i \neq j \\
  0, & \text{if } i = j \\
\end{cases} 
\end{split}
\end{equation}

which comes from setting $x_0 = \dots = \hat{x}_i = \dots = x_n$ \footnote{The notation $\hat{x}_i$ refers to skipping over it in the equality} and $x_i = 0$ and $\sum_j x_j = 1$.

From these points, an $(n-p)$-cell $P_{\{i_0,\dots,i_{p}\}}$ that would appear as a dual chain of a $p-$form with $\alpha(i_0,\dots,i_{p})=1$ can be written as:

\begin{equation}\label{cellParameterizations}
\begin{split}
&P_{\{i_0,\dots,i_{p}\}} = \Delta^n \cap \{c + \sum_{j=1}^{n-p} (f_{\ihat_j} - c)t_j | t_j \ge 0 \text{ for all } j\},\\
&\text{where } \{\ihat_1,\dots,\ihat_{n-p}\} = \{0,\dots,n\} \textbackslash \{i_0,\dots,i_{p}\}
\end{split}
\end{equation}

And in parallel, we can also write 

\begin{equation} \label{dualCellEqns1}
\begin{split}
P_{\{i_0,\dots,i_{p}\}} = \Delta^n \cap \big\{(&x_0,\dots,x_n) | x_{i_0} = \dots = x_{i_n} \text{ and } x_i \ge x_{\ihat},\\
&\text{  for all } i \in \{i_0,\dots,i_{p}\}, \ihat \notin  \{i_0,\dots,i_{p}\} \big\}\\
\end{split}
\end{equation}

which tells us that $P_{\{i_0,\dots,i_{p}\}}$ is where $\Delta^n$ intersects the plane of $x_{i_0} = \dots = x_{i_n}$, restricted to the points where $x_i \ge x_{\ihat}$ for $i \in \{i_0,\dots,i_p\}$ and $\ihat \in \{\ihat_1, \dots, \ihat_{n-p}\}$.

\subsection{More Notation}
Such $p$-cochains $\alpha$ will be denoted as living in the set $C^p(M,\Z_2)$. Closed $p$-cochains live in the set $Z^p(M,\Z_2) \subset C^p(M,\Z_2)$. So $C^p(M,\Z_2)$ with upper-index $p$ is the set of all functions from the $p$-simplices of $M$ to $\Z_2$. Here, $M$ implicitly refers to a manifold equipped with its triangulation. We will refer to the same manifold equipped with its dual cellulation as $M^\vee$. Poincaré duality says that the chains in $C_{n-p}(M^\vee,\Z_2)$ are in bijection with $C^p(M,\Z_2)$.

However, we could also use the words `cochains' and `chains' to describe a related set of objects. Namely, we could also consider $C^{p}(M^\vee,\Z_2)$, which are functions from $p$-cells of $M^\vee$ to $\Z_2$. There will be a completely analogous statement of Poincaré duality that $C^{p}(M^\vee,\Z_2)$ is in bijection with $C_{n-p}(M,\Z_2)$, so that chains living on $M$ are in bijection with cochains on $M^\vee$.

Throughout describing the higher cup products, we'll mostly be referring to `cochains' as being functions on a single $n$-simplex $\Delta^n$. Later on when discussing combinatorial $Spin$ structures, we'll see that representatives of Stiefel-Whitney classes naturally live in $C_{n-p}(M,\Z_2) = C^p(M^\vee,\Z_2)$.

\section{Warm up: The $\cup_0$ product as intersection from a `Morse Flow'} \label{CupZeroSection}
Now, as a warm up, let's review what the formula for $\cup_0$ had to do with vector field flow on the simplex $\Delta^n$. We'll use the standard notation that $\cup_0 = \cup$. Recall that for a $p$-cochain $\alpha \in C^p(X,\Z_2)$ and an $(n-p)$-cochain $\beta \in C^{n-p}(X,\Z_2)$, the value of $\alpha \cup \beta$ on an $n$-simplex $(0,\dots,n)$ is given as

\begin{equation}
(\alpha \cup \beta)(0,\dots,n) = \alpha(0,\dots,p) \beta(p,\dots,n)
\end{equation}

For a manifold $X$ with a simplicial decomposition and a branching structure, it is well known that the cup product on $H^*(X)$ is Poincaré dual to the intersection form on the associated chains, when viewed on $H_*(X)$. There is an elementary way to see directly on the cochain level why the intersection of the chains associated to $\alpha, \beta$ may take this form. This is discussed in \cite{Thorngren2018Thesis}, but it will be helpful to redo the discussion here before moving on to higher cup products.

As before, it will be helpful to explicitly visualize the case of $n=2$ before moving on to higher dimensions.

\subsection{Example: $\cup_0$ product in 2 dimensions}
The simplest example of a nontrivial cup product is the case $n=2$, between two 1-cochains. Suppose $\alpha$ and $\beta$ are both $\Z_2$ valued 1-cochains. Then, the value of $\alpha \cup \beta$ on the simplex $(0,1,2)$ is 

\begin{equation}
(\alpha \cup \beta)(0,1,2) = \alpha(0,1)\beta(1,2)
\end{equation}

Note that $\alpha$ and $\beta$ are both Poincaré dual to 1-chains, and the cell $P_{\{i,j\}}$ is included in the dual chain of $\alpha$ iff $\alpha(i,j)=1$. To see why the quantity $\alpha(0,1)\beta(1,2)$ plays a role in the intersection of $\alpha$ and $\beta$, we will introduce a `Morse Flow' of the chains within the simplex as follows. For some small real number $0 < \epsilon \ll 1$ and some fixed set of real numbers $b_0 < b_1 < b_2$, we will define new coordinates $\Tilde{x}_0,\Tilde{x}_1,\Tilde{x}_2$ on $\R^3$ as:

\begin{equation}
\Tilde{x}_i := x_i + \epsilon b_i, \text{ for } i=0,1,2
\end{equation}

Then, in parallel to our cells $P_{I}$ defined in Eq(\ref{dualCellEqns1}), we can define a set of `shifted' cells $\Tilde{P}_{I}$ defined by

\begin{figure}[h!]
  \centering
  \includegraphics[width=0.5\linewidth]{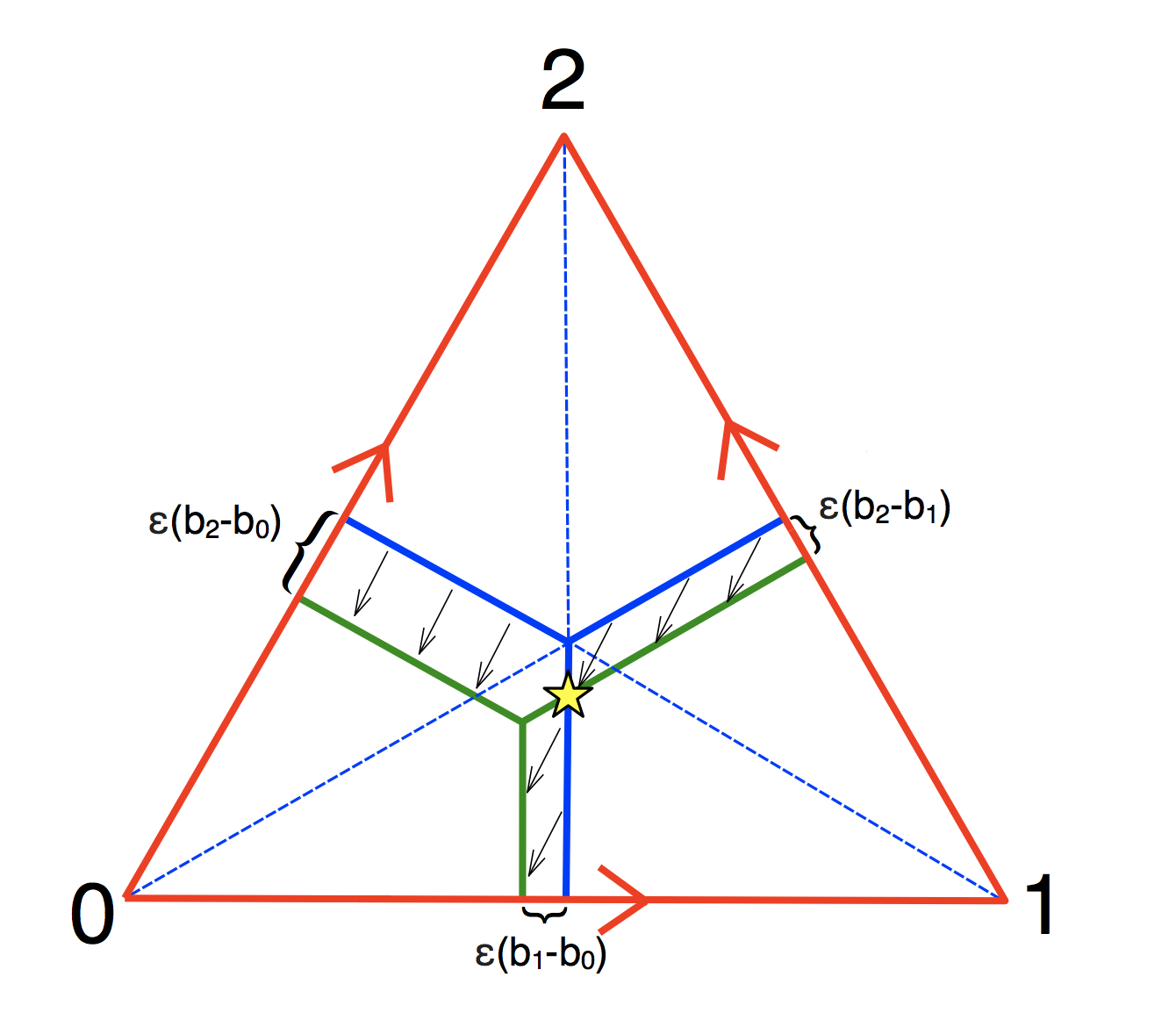}
  \caption{The green lines are the 1-skeleton of the \textit{flowed} 1-cells, i.e. $\Tilde{P}_{\{0,1\}},\Tilde{P}_{\{0,2\}},\Tilde{P}_{\{1,2\}}$. The solid blue lines are the 1-skeleton of the original cells, $P_{\{0,1\}},P_{\{0,2\}},P_{\{1,2\}}$, and the dashed blue lines complete the barycentric subdivision of $\Delta^2$. Notice that the flowed 1-cells of $\Tilde{P}$ only intersect once with the original cells $P$. More precisely, the only intersection point (the yellow star) is between $P_{\{0,1\}}$ and $\Tilde{P}_{\{1,2\}}$. This is the geometric interpretation for why the cup product $(\alpha \cup \beta)(0,1,2) = \alpha(0,1)\beta(1,2)$ represents an intersection.}
  \label{fig:2_Simplex_flowed}
\end{figure}

\begin{equation}\label{dualCellsShifted}
\begin{split}
\Tilde{P}_{\{i_0,\dots,i_{p}\}} = \Delta^n \cap \big\{(&x_0,\dots,x_n) | \Tilde{x}_{i_0} = \dots = \Tilde{x}_{i_n} \text{ and } \Tilde{x}_i \ge \Tilde{x}_{\ihat},\\
&\text{  for all } i \in \{i_0,\dots,i_{p}\}, \ihat \in \{0,\dots,n\} \big\}\\
\end{split}
\end{equation}

The Morse Flow will have this definition in every dimension. Pictorially, we can imagine the branching structures as playing the role of shifting the cells within $\Delta^2$, and creating an additional copy of them, as in Figure(\ref{fig:2_Simplex_flowed}). In that figure, we can see that each of the planes $\Tilde{x}_i=\Tilde{x}_j$ is `shifted' away from the plane $x_i=x_j$ by a transverse distance $b_j-b_i$. 

We can readily notice that the shifted cells $\Tilde{P}$ only intersect the original cells $P$ at exactly one point. Furthermore the only intersection point is between the cells $P_{\{0,1\}}$ and $\Tilde{P}_{\{1,2\}}$. This gives us a nice interpretation for the cup product. In other words, if we represent $\alpha$ by its representative chains on the original cells, $P$, and $\beta$ by its representatives on the $\Tilde{P}$, then we'll have that $\alpha \cup \beta (0,1,2)$ is 1 if those submanifolds intersect in $\Delta^2$ and 0 if they do not. 

Furthermore, for intersections of 0-cells with 2-cells, it's simple to see that the only pairs of such cells that intersect are $(P_{\{0\}},\Tilde{P}_{\{0,1,2\}})$ and $(P_{\{0,1,2\}},\Tilde{P}_{\{2\}})$. This matches up with the intuition that for $\alpha$ a 0-cochain (resp. 2-cochain) and $\beta$ a 2-cochain (resp. 0-cochain), then $\alpha \cup \beta (0,1,2) = \alpha(0)\beta(0,1,2)$  (resp. $\alpha \cup \beta (0,1,2) = \alpha(0,1,2)\beta(2)$).

Also, note that we can see a simple explanation the `non-commutative' property of the cup product, that on the cochain level $\alpha \cup \beta \neq \beta \cup \alpha$: it's simply because the Morse Flow breaks the symmetry of which cells in $\alpha$ intersect with which cells in $\beta$.

This intuition for the cup product will indeed hold for any chains in any dimension, a property which we'll state more precisely and verify in the next section.

\subsection{Cup product in general dimensions}
Let's state our first proposition about the $\cup_0$ product.

\begin{prop}
Fix $n \ge 2$. For $a$ sufficiently small and some subsets $I = \{i_0 < \dots < i_p\}$, $J =\{j_0 < \dots < j_q \}$ of $\{0,\dots,n\}$, the cells $P_{I,J}$ and $\Tilde{P}_{I,J}$ are defined as in Eq(\ref{dualCellsShifted}). Then, 
\begin{enumerate}
\item If $i_p > j_0$, then the intersection of the cells $P_I \cap \Tilde{P}_J$ is empty.
\item If $i_p = j_0$, then $\lim_{\epsilon \to 0} (P_I \cap \Tilde{P}_J) = P_{\{i_0,\dots,i_p=j_0,\dots,j_q\}}$
\item If $i_p < j_0$, then $\lim_{\epsilon \to 0} (P_I \cap \Tilde{P}_J) = P_{\{i_0,\dots,i_p,j_0,\dots,j_q\}}$
\end{enumerate}
where `limit' here means the Cauchy limit of the sets. 
\end{prop}

This is \textit{almost} the statement we want, modulo the subtlety which is Part 3 of the proposition. However, note that if $i_p < j_0$ then the cell $P_{\{i_0,\dots,i_p,j_0,\dots,j_q\}}$ is a dimension $n-p-q-1$ cell; this is a lower dimension than the case $i_p = j_0$ where $P_{\{i_0,\dots,i_p,j_0,\dots,j_q\}}$ is a dimension $n-p-q$ cell. Also note that for any finite $\epsilon >0$, the intersection of the cells will be an $(n-p-q)$-dimensional manifold. So in short, this proposition tells us that in the limit of $\epsilon \to 0$, the only intersections that retains the full dimension $n-p-q$ are between the $I = \{i_0,\dots,i_p\}$ and $J=\{j_0,\dots,j_q\}$ such that $i_p = j_0$.

Translated back to the cochain language, this proposition says for our $\Z_2$ cochains $\alpha, \beta$, the Poincaré dual of $\alpha \cup \beta$ is the union of all of the intersections as $\epsilon \to 0$ cells of $\big\{P_{\{i_0,\dots,i_p\}} | \alpha(i_0,\dots,i_p)=1 \big\}$ with $\big\{\Tilde{P}_{\{j_0,\dots,j_q\}} | \beta(j_0,\dots,j_q)=1 \big\}$ that survive as \textit{full, $(n-p-q)$-dimensional} cells in the limit of $\epsilon \to 0$, which satisfy $i_p = j_0$. This is a direct way to see how the cup product algebra interacts with the intersection algebra. Now, we can give the proof.

\begin{proof}
Recall that $P_I$ and $\Tilde{P}_J$ are defined, respectively, by the relations:
\begin{equation}
\begin{split}
P_{\{i_0,\dots,i_{p}\}} = \Delta^n \cap \big\{(&x_0,\dots,x_n) | x_{i_0} = \dots = x_{i_n} \text{ and } x_i \ge x_{\ihat},\\
&\text{  for all } i \in \{i_0,\dots,i_{p}\}, \ihat  \in \{0,\dots,n\} \big\}\\
\end{split}
\end{equation}

and

\begin{equation}
\begin{split}
\Tilde{P}_{\{j_0,\dots,j_q\}} = \Delta^n \cap \big\{(&\Tilde{x}_0,\dots,\Tilde{x}_n) | \Tilde{x}_{j_0} = \dots = \Tilde{x}_{j_q} \text{ and } \Tilde{x}_j \ge \Tilde{x}_{\jhat},\\
&\text{  for all } j \in \{j_0,\dots,j_q\}, \jhat \in \{0,\dots,n\} \big\}\\
\end{split}
\end{equation}

The definition of $\Tilde{P}_{\{j_0,\dots,j_q\}}$ can be rewritten as

\begin{equation}
\begin{split}
\Tilde{P}_{\{j_0,\dots,j_q\}} = \Delta^n \cap \big\{(&x_0,\dots,x_n) | x_{j_0} = x_{j_1} + \epsilon (b_{j_1} - b_{j_0}) = \dots = x_{j_q} + \epsilon (b_{j_q} - b_{j_0}) \text{ and } x_j \ge x_{\jhat} + \epsilon(b_{\jhat}-b_j),\\
&\text{  for all } j \in \{j_0,\dots,j_q\}, \jhat  \in \{0,\dots,n\} \big\}\\
\end{split}
\end{equation}

Now, we can see why Part 1 is true. Suppose $i_p > j_0$ and $\epsilon>0$. Any point $(x_0,\dots,x_n)$ in the intersection would need to satisfy $x_{j_0} \ge x_{i_p} + \epsilon(b_{i_p}-b_{j_0}) > x_{i_p} \ge x_{j_0}$, i.e. $x_{j_0} > x_{j_0}$ which is impossible. Here, we used that $b_{i_p}-b_{j_0} > 0$ for $i_p > j_0$. So, there are no points in $P_I \cap \Tilde{P}_J$.

The argument for Part 2 is similar. It's not hard to check that the intersection $P_I \cap \Tilde{P}_J$ is defined by the equations 

\begin{equation}
\begin{split}
P_I \cap \Tilde{P}_J = \Delta^n \cap \bigg\{(&x_0,\dots,x_n) | x_{i_0} = \dots = x_{i_p} = x_{j_0} = x_{j_1} + \epsilon (b_{j_1} - b_{j_0}) = \dots = x_{j_q} + \epsilon (b_{j_q} - b_{j_0}),\\
&\text{  and } x_k \ge x_{\hat{k}} + \epsilon \Tilde{D}_{k \hat{k}} \text{ for all } k \in \{i_0,\dots,i_p=j_0,\dots,j_q\}, \hat{k}  \in \{0,\dots,n\}, \\
&\text{  where } \bigg(\Tilde{D}_{k \hat{k}} := 
  \begin{cases}
    0               &\text{ if } k \notin J \\
    b_{\hat{k}}-b_k &\text{ if } k \in J \textbackslash \{i_p=j_0\} \\
    max\{0, b_{\hat{k}}-b_k\} &\text{ if } k = i_p = j_0 \\
  \end{cases} \bigg)
\bigg\}\\
\end{split}
\end{equation}

And, in the limit $\epsilon \to 0$, we'll have that this set becomes precisely $P_I \cap \Tilde{P}_J \to P_{i_0,\dots,i_p=j_0,\dots,j_q}$.

The argument for Part 3 is again similar to both of the previous parts. Similarly to Part 1, we have the constraint $x_{j_0} \ge x_{i_p} - \epsilon(b_{j_0}-b_{i_p}) \ge x_{j_0} - \epsilon(b_{j_0}-b_{i_p})$. But, since now $i_p < j_0$, we'll have that this constraint limits $x_{i_p}$ to lie in the range $[x_{j_0} - \epsilon(b_{j_0}-b_{i_p}),x_{j_0}]$. In the limit of $\epsilon \to 0$, this will enforce $x_{i_p} = x_{j_0}$. So, in the limit of $\epsilon \to 0$, we'll have that 

\begin{equation}
\begin{split}
P_I \cap \Tilde{P}_J &\xrightarrow{\epsilon \to 0} \Delta^n \cap \big\{(x_0,\dots,x_n) | x_{i_0} = \dots = x_{i_p} = x_{j_0} = \dots = x_{j_q}  \text{ and } x_k \ge x_{\hat{k}},\\
&\quad\quad\quad\quad\quad\quad \text{ for all } k \in \{i_0,\dots,i_p,j_0,\dots,j_q\}, \hat{k}  \in \{0,\dots,n\} \big\}\\
&= P_{\{i_0,\dots,i_p,j_0,\dots,j_q\}}
\end{split}
\end{equation}
\end{proof}

\subsection{Comparing the vector fields on different simplices} \label{flowedSimplices}
While our vector fields satisfy the desired intersection properties within each simplex, one minor issue that we should address is to think of how the vector fields compare on the boundaries of neighboring simplices. It will not be the case that the vector fields will match on neighboring two simplices. However, the branching structure will ensure that the vector fields can be smoothly glued together without causing any additional intersections between the chains (see Figure(\ref{neighboringTwoSimplices})). This is because the branching structure will make sure that the flowed simplices will be flowed on the same side of the original simplex. So those flows can be connected between different faces to avoid any further intersections than the ones inside the simplices themselves. So, the intersection numbers on the whole triangulated manifold will just be given by the intersections on the interiors. In the cases where the intersection classes are higher dimensional, the intersection classes themselves can be also be connected to meet on the boundaries.

We expect that these flows can be \textit{smoothly} connected to match on the boundaries. However, we will avoid explicitly smoothing the vector fields at the boundaries due to the technicalities that tend to be involved in such constructions. For example, in high dimensions a single piecewise-linear structures on a manifold generically corresponds to many smooth structures, so we would expect an explicit smoothing of these maps to depend on the particular smooth structure. However, in discussing the GWGK Integral, we'll be able to explicitly connect the vector fields in a neighborhood of the 1-skeleton. This is since we'll do everything in local coordinates which aren't as technical to deal with just near the 1-skeleton.

\begin{figure}[h!]
  \centering
  \includegraphics[width=0.5\linewidth]{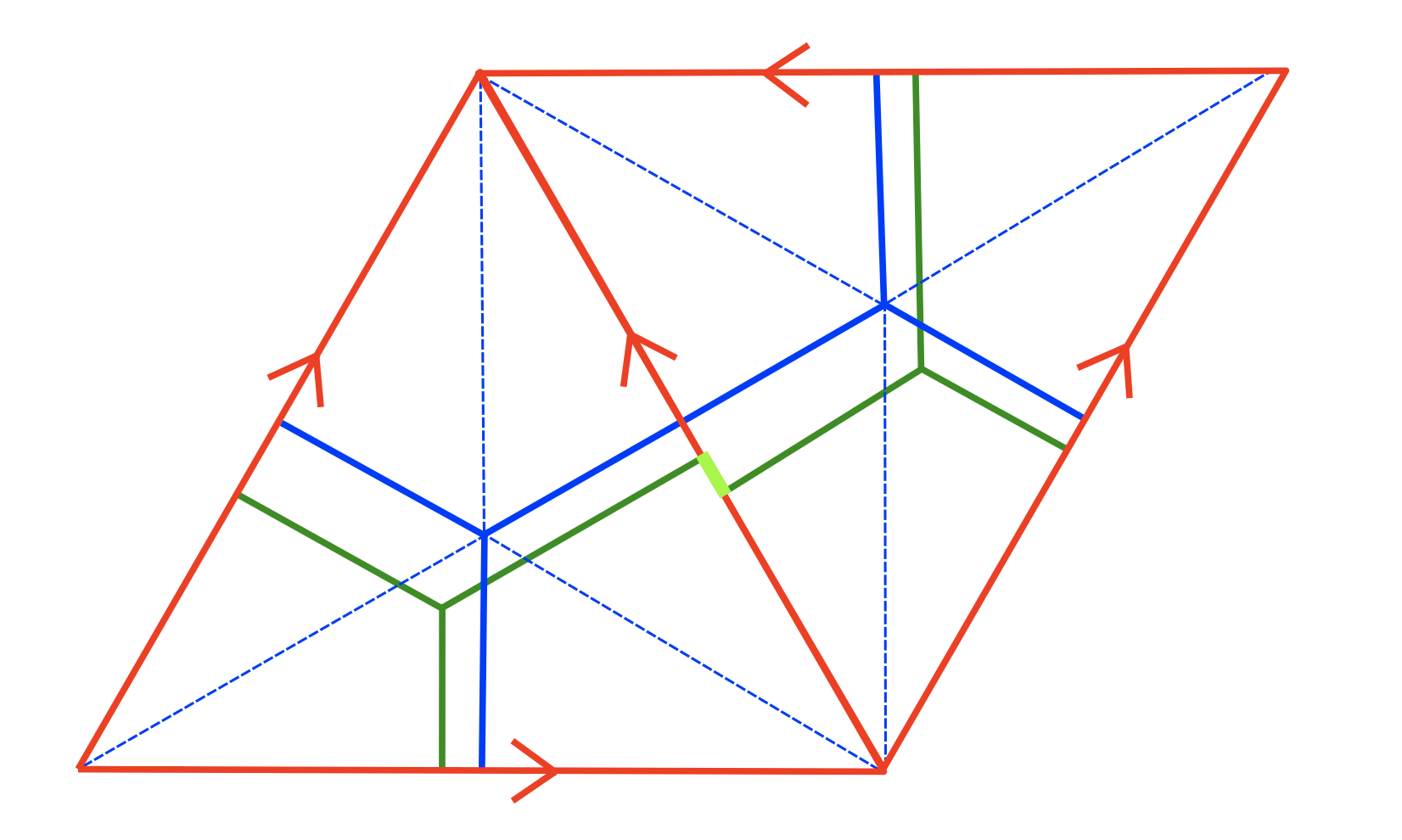}
  \caption{The flows on different simplices may not match on the boundaries. But, the flows can be connected (the bright green line) to avoid any additional intersections between the flowed (green) cells and the original (blue) cells.}
  \label{neighboringTwoSimplices}
\end{figure}

\section{$\cup_i$ products from $(i+1)$-parameter Morse Flows} \label{sec:HigherCupProductMain}
First, let us recall the definition and some properties of the higher cup products (see e.g. \cite{MosherTangora}). Given some $\alpha$ a $p$-cochain and $\beta$ a $(n-p+i)$-cochain, we'll have that $\alpha \cup_i \beta$ is an $n$-cochain, such that  when restricted to an $n$-simplex, 

\begin{equation} \label{higherCupFormula}
(\alpha \cup_i \beta)(0,\dots,n) = \sum_{0 \le j_0 < \dots < j_i \le n} \alpha(0 \to j_0, j_1 \to j_2, \dots) \beta(j_0 \to j_1, j_2 \to j_3,\dots)
\end{equation}

where we use the notation $j \to k$ to refer to $j,\dots,k$. There is a caveat in the above definition, that we just restrict to those $j_0 < \dots < j_i$ such that $\#\{0 \to j_0, j_1 \to j_2, \dots \} = p+1$ and $\#\{j_0 \to j_1, j_2 \to j_3, \dots \} = n-p+i+1$, so not all $\{j_0,\dots,j_i\}$ contribute to the sum. For example, if $\alpha$ and $\beta$ are both 2-cochains, then $\alpha \cup_1 \beta$ is a 3-cochain with $(\alpha \cup_1 \beta)(0,1,2,3) = \alpha(0,1,3)\beta(1,2,3) + \alpha(0,2,3)\beta(0,1,2)$, so only two of the ${4 \choose 2} = 6$ choices of $0 \le j_0 < j_1 \le 3$ contribute in this case.

It is well-known that the $\cup_i$ products are \textit{not} cohomology operations, i.e. that $\alpha \cup_i \delta \lambda$ may not be cohomologically trivial or even closed, even though $\delta \lambda$ is exact. Despite the fact that $\cup_i$ is not a \textit{binary} cohomology operation, it will in fact be a \textit{unary} cohomology operation - the maps called the `Steenrod Squares',

\begin{equation}
\alpha \mapsto \alpha \cup_i \alpha =: Sq^{p-i}(\alpha)
\end{equation}

will always be closed for $\alpha$ closed and (up to a boundary) only depend on the cohomology class of $\alpha$. The root algebraic property of the $\cup_i$ products is the formula:

\begin{equation}
\delta(\alpha \cup_k \beta) = \delta\alpha \cup_k \beta + \beta \cup_k \delta\alpha + \alpha \cup_{k-1} \beta + \beta \cup_{k-1} \alpha
\end{equation}

This implies that if $\delta \alpha = 0$, then $\delta(\alpha \cup_i \alpha) = \delta\alpha \cup_{i} \alpha + \alpha \cup_{i} \delta\alpha + 2 \alpha \cup_{i-1} \alpha = 0$, so $\alpha \cup_i \alpha$ is closed for closed $\alpha$. And, $(\alpha + \delta \beta) \cup_k (\alpha + \delta \beta) = \alpha \cup_k \alpha + \delta\big(\alpha \cup_{k+1} \delta\beta + \beta \cup_k \delta\beta + \beta \cup_{k-1} \beta \big)$, meaning $Sq^{p-k}(\alpha)$ and $Sq^{p-k}(\alpha + \delta \beta)$ are closed cocycles in the same cohomology class.

An important consequence of the above equality is that, up to a coboundary, we'll have that for $\alpha$ a cocycle and $\lambda$ any cochain, we'll have (up to a coboundary):

\begin{equation} \label{higherCupSymm}
\alpha \cup_i \delta\lambda \equiv \alpha \cup_{i-1} \lambda + \lambda \cup_{i-1} \alpha
\end{equation}

which relates the $\cup_i$ to how the $\cup_{i-1}$ differs under switching the order of $\alpha,\lambda$ (or equivalently under reversing the branching structure).

\subsection{Motivation for the $\cup_1$ product}
Now, let's give a key example to motivate why `thickening' the chains should seem useful to describe the higher cup products. We'll start with the simplest example of the $\cup_1$ product.

Let's consider the $\cup_1$ product between a closed cochain $\alpha$ and some boundary $\delta\lambda$. Let's consider $\alpha$ a $p$-cocycle and $\lambda$ an $(n-p)$-cochain, so that $\alpha \cup_1 \delta \lambda$ is an $n$-cochain. From Eq(\ref{higherCupSymm}), we'll have:

\begin{equation}
\alpha \cup_1 \delta\lambda \equiv \alpha \cup \lambda + \lambda \cup \alpha
\end{equation}

To think about what this term means, we should first think about what each of the $\alpha \cup \lambda, \lambda \cup \alpha$ mean. Based on our observations in the previous section, we can see that $\alpha \cup \lambda$ measures where $\lambda$ intersects with a version of $\alpha$ shifted in the direction of the \textit{positive} Morse flow. And $\lambda \cup \alpha$ measures where $\lambda$ intersects with a copy of $\alpha$ flowed in the \textit{negative} direction. So, we see that $\int \alpha \cup_1 \delta\lambda = \int \lambda \cup \alpha + \int \alpha \cup \lambda$ measures how the intersection numbers of the chains representing $\alpha$ and $\lambda$ change with respect to the positive and negative Morse Flows.\footnote{Note that $\alpha \cup \lambda$ may not equal $\lambda \cup \alpha$ plus a coboundary, since $\lambda$ may not be closed. So $\int \alpha \cup \lambda$ may not equal $\int \lambda \cup \alpha$}

\begin{figure}[h!]
  \centering
  \begin{minipage}{0.44\textwidth}
    \centering
    \includegraphics[width=\linewidth]{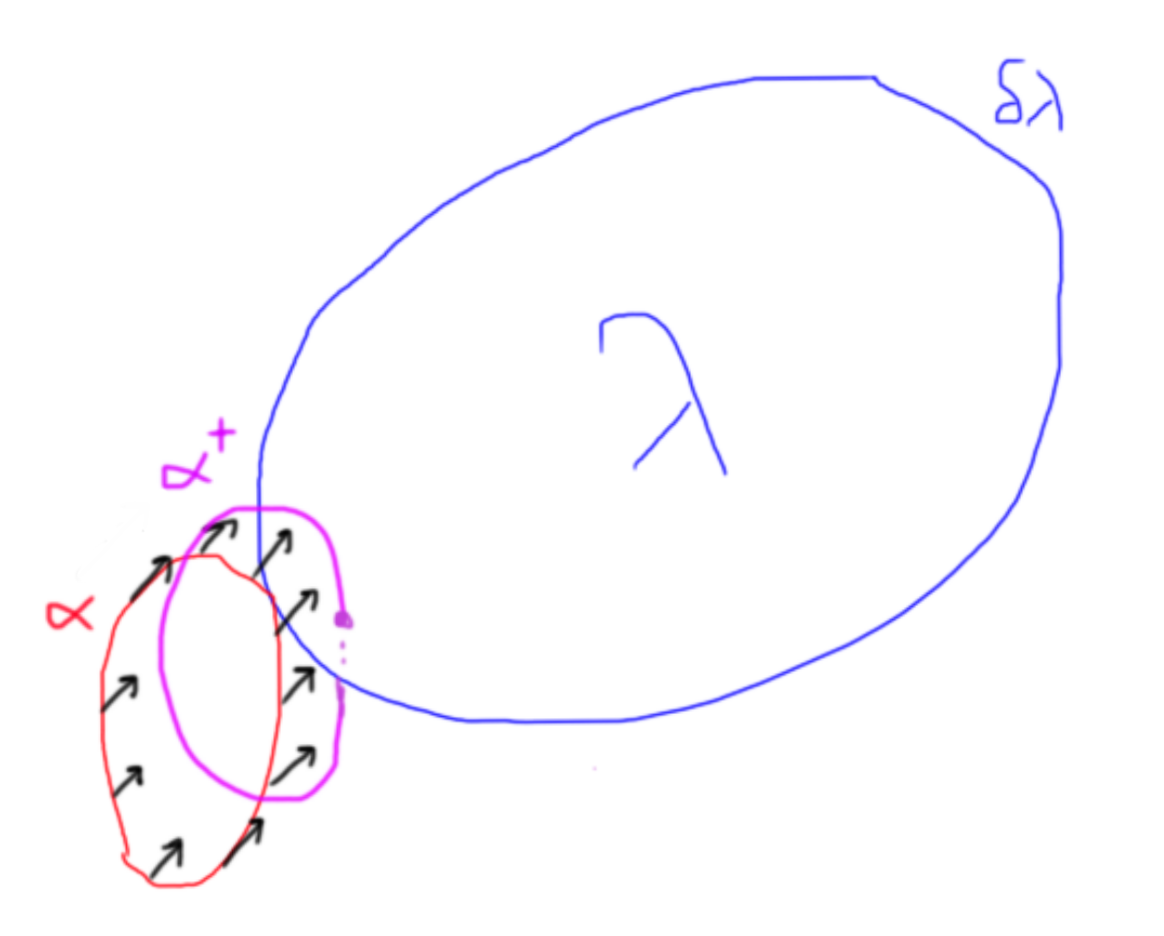}
  \end{minipage} \quad \quad \quad
  \begin{minipage}{0.44\textwidth}
    \centering
    \includegraphics[width=\linewidth]{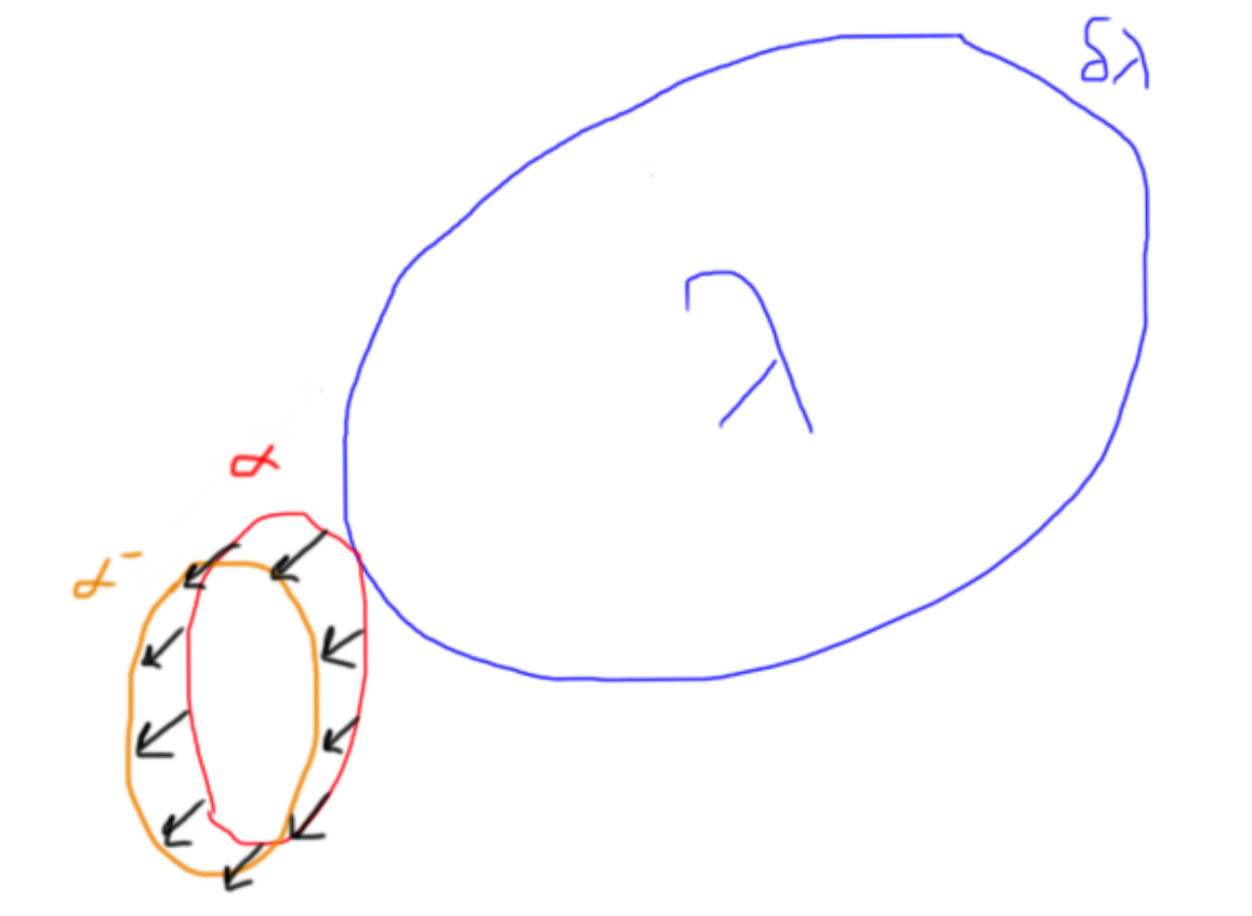}
  \end{minipage}
  
  \caption{The change in the linking number of $\delta\lambda$ and $\alpha$ under the Morse flow gives $\int \alpha \cup_1 \delta \lambda$. Note that by construction, $\alpha$ and $\delta \lambda$ can only intersect at the barycenters of the 3-simplices of the triangulation and the lines connecting these barycenters. Here, we chose to depict them only meeting at the barycenter of a single one simplex. This would happen if, e.g. $\alpha(0,1,2)=\alpha(0,1,3)=\delta\lambda(0,2,3)=\delta\lambda(1,2,3)=1$ with all other entries being zero at a 3-simplex. \\
  (Left) Shifting a 1-form $\alpha$ via the positive Morse flow, giving a shifted curve $\alpha^+$. $\alpha^+$ intersecting $\lambda$ once means that $\alpha^+$ and $\delta\lambda$ have a linking number of 1. (Right) Shifting a 1-form $\alpha$ via the negative Morse flow, giving a shifted curve $\alpha^-$. $\alpha^-$ doesn't intersect $\lambda$, so  $\alpha^-$ and $\delta\lambda$ have a linking number of 0. }
  \label{LinkingNumberCup1}
\end{figure}

In three dimensions, we can visualize this as follows. Suppose $\alpha$ and $\delta \lambda$ are 2-cocycles, so that $\lambda$ is a 1-cochain. This means that $\alpha$ and $\delta \lambda$ are dual to closed 1D curves on the dual lattice and $\delta \lambda$ is dual to a the boundary of the 2-surface that's dual to $\lambda$. Recall $\int \alpha \cup_1 \delta\lambda = \int \alpha \cup \lambda + \int \lambda \cup \alpha$ gives the difference between the intersection numbers of $\alpha$ with $\lambda$ in the positive and negative Morse flow directions. In the case where $\alpha$ is a trivial curve and the manifold is $S^3$, we can visualize this process as how the \textit{linking number} of $\alpha$ and $\delta \lambda$ changes under the Morse flow, due to the well known fact that the (mod 2) linking number between two curves $C_1, C_2$ is the number of times (mod 2) that $C_1$ intersects a surface that $C_2$ bounds. This is shown in Figure(\ref{LinkingNumberCup1}).

While this linking number picture is a nice way to visualize the integrals of certain $\cup_1$ products in three dimension, it is still somewhat unsatisfactory. First, the linking number is often subtle to define and may not make sense; e.g. in higher dimensions, or if the manifold or the curves themselves are topologically nontrivial, it's not always possible to define linking numbers. Next, a linking number is a global quantity that requires global data to compute it, whereas the higher cup products are local quantities defined on every simplex. And, most glaringly, this picture only gives us information about cochains of the form $\alpha \cup_1 \delta\lambda$, while it'd be nice to understand it for more general pairs of cochains.

\begin{figure}[h!]
  \centering
  \includegraphics[width=0.60\linewidth]{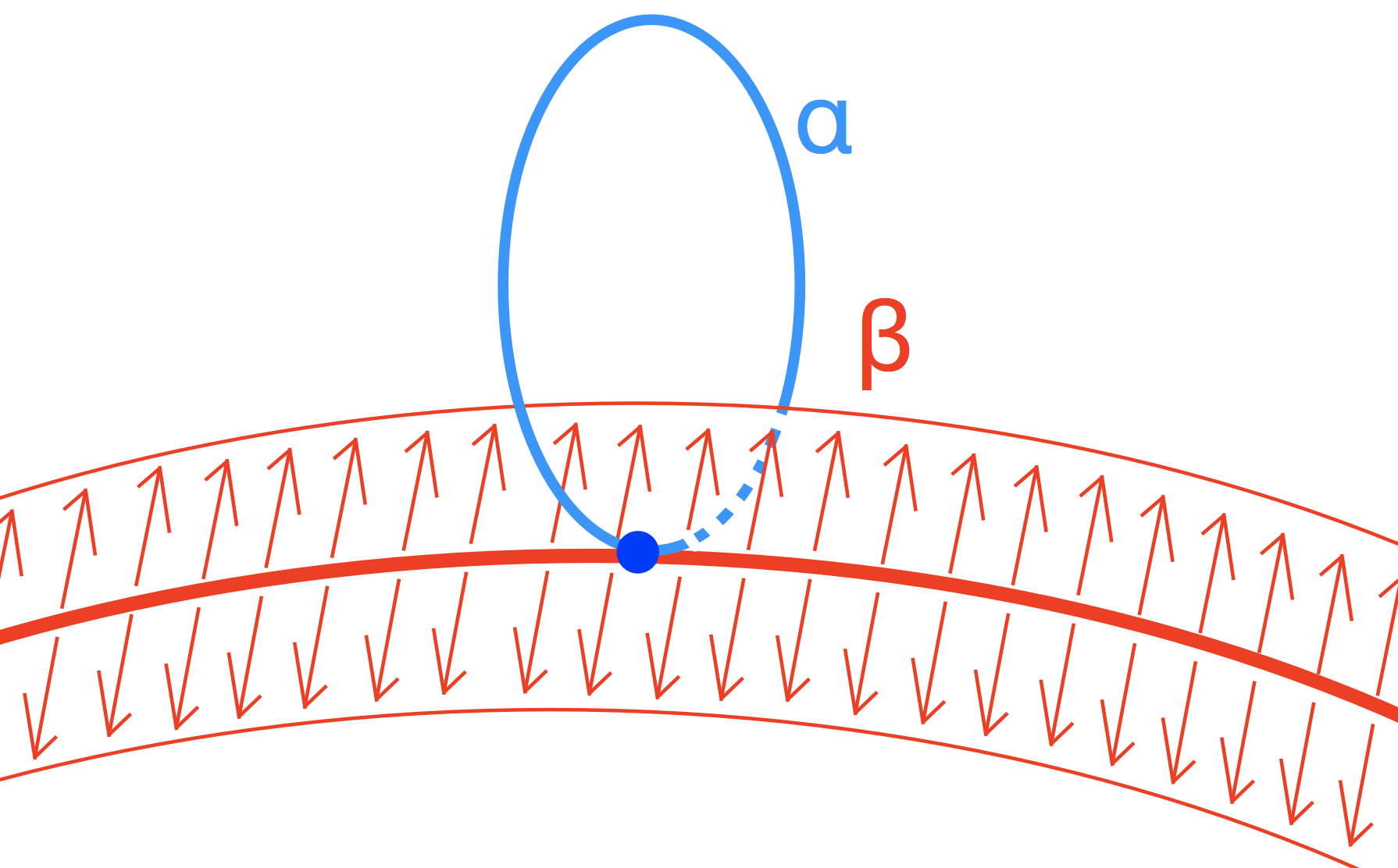}
  \caption{$\alpha$ and $\beta$ are represented by the blue curve and central red curves. The thickening of $\beta$ is given in both the positive and negative directions of the Morse flow (both directions pointing away from the central red curve). $\int \alpha \cup_1 \beta$ measures the intersection of $\alpha$ with this thickening of $\beta$. Note that if $\beta$ was a trivial curve, this integral gives how the linking numbers of $\alpha$ and $\beta$ change with respect to the Morse flow.}
  \label{fig:Cup1Thiccening}
\end{figure}

Following this intuition of trying to give a `local' geometric definition, we are lead to the idea of `thickening' the chains. In particular, we can note that this difference of linking numbers can be also attributed to `thickening' $\delta \lambda$ in \textit{both} directions of the Morse flow and then measuring the intersection number of $\alpha$ with this thickening of $\delta \lambda$. For example, see Figure(\ref{fig:Cup1Thiccening}). This could be anticipated from the linking number intuition, since the change in the linking number under the Morse flow only depends on the surface in a neighborhood of the second curve.

So, it seems like we've found a potential geometric prescription to assign to the $\cup_1$ product. While this is in line with our intuition, we quickly run into an issue when we try to implement this on the cochain level: the intersection between the original cells and their thickenings is degenerate (i.e. the curves intersect on their boundaries). We can see this by drawing the simplest example, see Figure(\ref{fig:CochainThicceningNotWorkVsWork}), of the intersection of a 1-chain with the thickening of a 1-chain. It's not hard to convince oneself that the only intersection point between a cell of $\alpha$ and the thickened version of a different cell of $\beta$ will be the barycenter, $c$, which is at the boundary of both the cell in $\alpha$ and the thickened cell of $\beta$. And, the intersection of a cell with its own thickening will simply be itself, not anything lower-dimensional.

This was basically the same issue we faced with the original cup product. The way we dealt with this degenerate intersection before was to shift $\alpha$ along the direction of the Morse Flow, which made the intersection nondegenerate. We could again try shifting $\alpha$ along the Morse flow, but we'll quickly realize that these shifted cells of $\alpha$ will only intersect at the thickened $\beta$'s edge: simply because the thickened $\beta$ was defined with respect to the Morse flow in the first place! To resolve this ambiguity, we will need to shift $\alpha$ by along a vector that's linearly independent from all the other vectors. This way, we can arrange for there to be a definite intersection point between the thickened cells of $\beta$ and the shifted cells of $\alpha$. 

There is one aspect in this that we should be careful about. Let's say we thickened $\beta$ along the original Morse flow vector $\Vec{v}$ by some thickness $\epsilon_1$. Then, we'll want to shift $\alpha$ along the second Morse flow vector $\Vec{w}$ by some distance $\epsilon_2 \ll \epsilon_1$. This is because once $\epsilon_2$ becomes too big compared to $\epsilon_1$, then the intersection locus might change its topology, which can be seen by examining Figure(\ref{fig:CochainThicceningNotWorkVsWork}).

\begin{figure}[h!]
  \centering
  \begin{minipage}{0.44\textwidth}
    \centering
    \includegraphics[width=\linewidth]{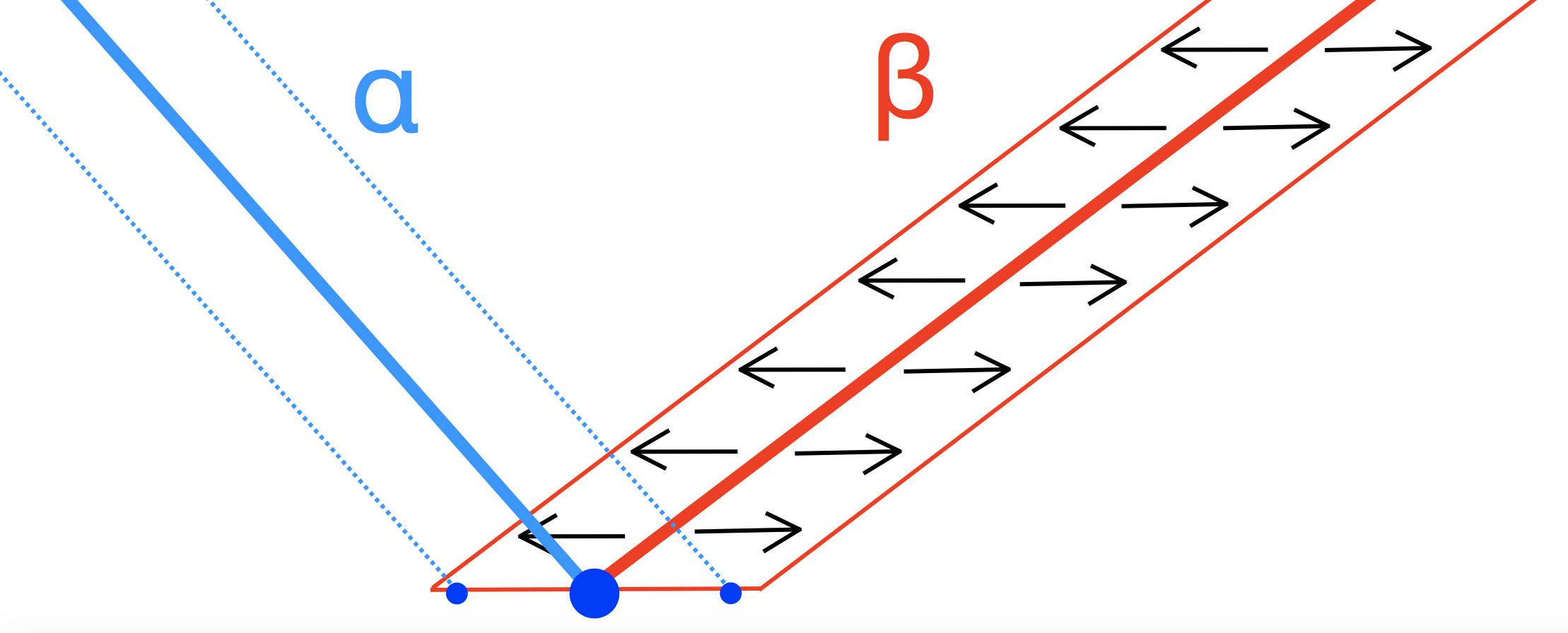}
  \end{minipage} \quad \quad \quad
  \begin{minipage}{0.44\textwidth}
    \centering
    \includegraphics[width=\linewidth]{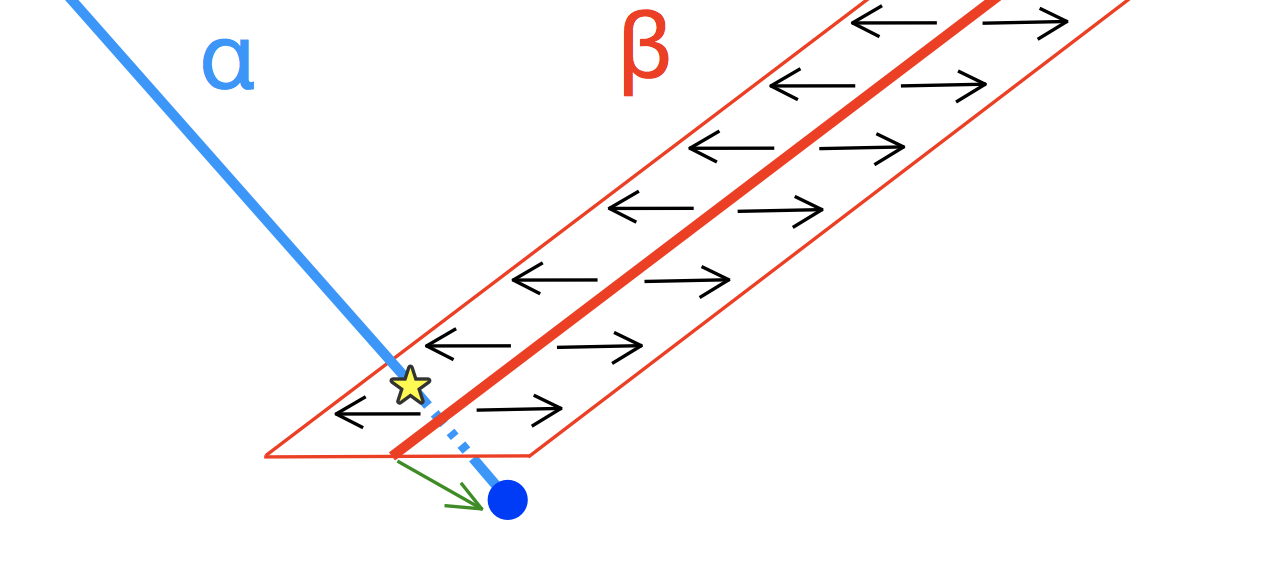}
  \end{minipage}
  
  \caption{(Left) $\alpha$ and $\beta$ are represented by the \textit{thick} blue and red curves.  The thickening of $\beta$ is given by the interior of the red rectangle. And some shifts of $\alpha$ along the Morse Flow are given by the dashed blue lines. No matter how you move $\alpha$ along the \textit{original} Morse flow, $\alpha$ and the thickening of $\beta$ will always have a degenerate intersection, at their boundaries. This means that the intersection of a chain with the thickened version of itself will be the chain itself, not a lower dimensional version. (Right) However, if we shift $\alpha$ by a \textit{new} vector (green arrow) that's linearly independent from the original Morse flow vector (black arrows), then we can say whether or not this shifted $\alpha$ and the thickened $\beta$ intersect, non-degenerately. Note that if we shift $\alpha$ by an amount comparable to how much we thickened $\beta$, then the intersection region (yellow star) will change in topology.}
  \label{fig:CochainThicceningNotWorkVsWork}
\end{figure}

\textbf{\underline{Intuition for Higher Cup Products}}

From here, we can see a general pattern of thickening and shifting that we can perform to try to compare to the higher cup products. For example, after accepting this for the $\cup_1$ product, we can apply the intuitive reasoning to

\begin{equation*}
\alpha \cup_2 \delta\lambda \equiv \alpha \cup_1 \lambda + \lambda \cup_1 \alpha
\end{equation*}

to see that $\cup_2$ could be thought of as measuring how $\alpha \cup_1 \lambda$ changes under a Morse Flow. We'll see that the $\cup_i$ product is obtained from thickening the cells of one cochain by an $i$-parameter Morse Flow and then shifting the other a small distance to make a well-defined intersection. 

It is often said that the higher cup products measure how much the lower cup products `fail to commute'. For example, the $\cup_1$ product gives an indication of how badly $\cup_0$ product doesn't commute on the cochain level. Geometrically, this is saying that the $\cup_1$ product `measures' the how the intersection of the cells differs under the Morse flow between $\epsilon$ positive or negative. Looking forward, we will see that the geometric way to see $\cup_1$ is to `thicken' the cells under \textit{both} the positive and negative direction of the Morse flow and measure an intersection of the original cell with the thickened cell. However, to measure such an intersection, we \textit{again} need to break the symmetry by introducing an additional vector flow, for a similar reason that we needed to break the symmetry to measure intersection in the first place. 

In general, the $\cup_i$ product will involve an $(i+1)$-parameter Morse Flow, where the first $i$ directions of the flow thicken the manifold and the last direction breaks the symmetry in order to be able to measure an intersection number. We note that much of this discussion was proposed by Thorngren in \cite{Thorngren2018Thesis}. The rest of this section will be devoted to setting up the algebra needed to realize this and showing that the higher cup product formulas are exactly reproduced by such a procedure.

\subsection{Defining the thickened cells} \label{definingThickenedCells}
Our first goal should be to write down parametric equations defining the points of the flowed cells, analogous to the ones in Eq(\ref{cellParameterizations}). To do this, we will need to define some variables $\Tilde{x}_i$  in an analogous way as we did for the 1-parameter Morse flow. Recall, that we defined $\Tilde{x}_j = x_j + \epsilon b_j$ where $b_0 < \dots < b_n$. Here, there was a single $\epsilon$ that played the role of the Morse Flow parameter and the vector $\Vec{b} = (b_j)$ was the Morse Flow vector. For the $\cup_m$ product, we will need an $(m+1)$-parameter Morse Flow, which means that we need $(m+1)$ linearly independent vectors. Let's call these vectors $\Vec{b}_i$ for $i \in \{1 \dots m+1\}$ and denote by $b_{ij}$ the matrix of these vectors,

\begin{equation}
\Vec{b}_i = (b_{i 0},\dots,b_{i n}) \text{  for } i = 1,\dots,m+1
\end{equation}

We'll define quantities $\epsilon_{j}$ which play the role of the Morse Flow parameters, and we'll similarly define our shifted coordinates $\Tilde{x}_i$ as:

\begin{equation}
\Tilde{x}_i = x_i + \epsilon_1 b_{1 i} + \dots + \epsilon_{m+1} b_{m+1, i}
\end{equation}

Now to parametrically define our thickened cells, we'll define the points $\Tilde{c}$ near the barycenter of $\Delta^n$ and the points $\Tilde{f}_j$ near the centers of the faces as follows. $\Tilde{c}$ will be defined by setting all the $\Tilde{x}_j$ coordinates equal, $\Tilde{x}_0 = \dots = \Tilde{x}_n$, analogously to how we defined the center $c$, earlier. Before writing the expression for $\Tilde{c}$, we will find it convenient to define new quantities $B_i$ as

\begin{equation}
B_i := \sum_{j=0}^n b_{ij}
\end{equation}

Then, solving the equations $x_0 + \dots + x_n = 1$ and $\Tilde{x}_0 = \dots = \Tilde{x}_n$, it's straightforward to see that:

\begin{equation} \label{shiftedCenterMorse}
\begin{split}
\Tilde{c}   &= (\Tilde{c}_0,\dots,\Tilde{c}_n), \text{  where} \\
\Tilde{c}_i &= \frac{1}{n+1}(1+B_1 \epsilon_1 + \dots + B_{m+1} \epsilon_{m+1}) - (b_{1 i}\epsilon_1 + \dots + b_{m+1,i}\epsilon_{m+1})
\end{split}
\end{equation}

And, we'll similarly define our points $\Tilde{f}_j$ by setting $\Tilde{x}_0 = \dots = \hat{\Tilde{x}}_j = \dots = \Tilde{x}_n$ and $x_j = 0$ and $x_0 + \dots + x_n = 1$. Solving these equations will give us:

\begin{equation} \label{shifedFacesMorse}
\begin{split}
\Tilde{f}_j     &= ((\Tilde{f}_j)_0,\dots,(\Tilde{f}_j)_n), \text{  where} \\
(\Tilde{f}_j)_i &= \bigg(\frac{1}{n}(1+(B_1 - b_{1 j}) \epsilon_1 + \dots + (B_{m+1}-b_{m+1,j}) \epsilon_{m+1}) - (b_{1 i}\epsilon_1 + \dots + b_{m+1, i}\epsilon_{m+1})\bigg) \delta_{i \neq j}, \text{  and} \\
\delta_{i \neq j} &= 
\begin{cases}
1 \text{  if  } i \neq j \\
0 \text{  if  } i   =  j \\
\end{cases}
\end{split}
\end{equation}

We should clarify that these $\Tilde{c}$ and $\Tilde{f}$ are really functions of the $\epsilon_j$.

Now, for some fixed set of $\epsilon_j$, we can define the shifted cells $\Tilde{P}_{\{i_0 \dots i_p\}}$ entirely analogously to the original cells in Eq(\ref{cellParameterizations}):

\begin{equation}
\begin{split}
\Tilde{P}_{\{i_0,\dots,i_{p}\}}(\epsilon_1,\dots,\epsilon_{m+1};\varepsilon) &= \Delta^n \cap \{\Tilde{c} + \sum_{j=1}^{n-p} (\Tilde{f}_{\ihat_j} - \Tilde{c})t_j | t_j \ge 0 \text{ for all } j\},\\
&\text{where } \{\ihat_1,\dots,\ihat_{n-p}\} = \{0,\dots,n\} \textbackslash \{i_0,\dots,i_{p}\}
\end{split}
\end{equation}

These $\Tilde{P}_{\{i_0 \dots i_p\}}$ are defined with respect to a fixed set of $\epsilon_j$: as of now they are not thickened. Note that they depend on multiple parameters $\epsilon_1,\dots,\epsilon_{m+1}$ and that this definition agrees with our previous expression Eq(\ref{dualCellsShifted}). We could have defined such a parameterization earlier when talking about the $\cup_0$ product, but it wasn't necessary at that point as it is now. Another equivalent way to express this is \footnote{This is because the cells can be thought of as shifting by the vectors $\{\epsilon_j (\Vec{b}_j - \frac{1}{n+1} B_j \Vec{b}_0)\}$, which are the projections of $\{\epsilon_j \Vec{b}_j\}$ onto $\Delta^n$}:

\begin{equation}
\begin{split}
\Tilde{P}_{\{i_0,\dots,i_{p}\}}(\epsilon_1,\dots,\epsilon_{m+1};\varepsilon) = \Delta^n \cap \{& \epsilon_1 \Vec{b}_1 + \cdots + \epsilon_{m+1} \Vec{b}_{m+1} - \frac{1}{n+1}(\epsilon_1 B_1 + \cdots \epsilon_{m+1}B_{m+1}) \cdot \Vec{b}_0 \\
                                                                                   & + \sum_{j=1}^{n-p} (f_{\ihat_j} - c)t_j | t_j \ge 0 \text{ for all } j\},\\
\text{where } & \{\ihat_1,\dots,\ihat_{n-p}\} = \{0,\dots,n\} \textbackslash \{i_0,\dots,i_{p}\}
\end{split}
\end{equation}

Here, $\Vec{b}_0$ refers to the vector $(1,\dots,1)$. To thicken them, we should also treat the $\epsilon_j$ as a parameter to be varied. So overall, our thickened cells $\Tilde{P}^\text{thick}_{\{i_0 \dots i_p\}}((\epsilon_{m+1};\varepsilon))$ will be written as:

\begin{equation} \label{thickenedCellDefinition}
\Tilde{P}^\text{thick}_{\{i_0 \dots i_p\}}(\epsilon_{m+1};\varepsilon) = \bigsqcup_{\epsilon_j \in (-\varepsilon, \varepsilon)  \text{ for } j \in \{1,\dots,m\}} \Tilde{P}_{\{i_0 \dots i_p\}}(\epsilon_1,\dots,\epsilon_{m+1};\varepsilon)
\end{equation}

Above, we only allowed the $\epsilon_1,\dots,\epsilon_m$ to vary since we're considering an $m$-parameter thickening. And, we only thickened the cells up to some small fixed number $0 < \varepsilon \ll 1$. Eventually $\epsilon_{m+1} \ll \varepsilon$ will induce a small shift to let us define a non-degenerate intersection.

What we really mean by these $\ll$ signs is that we're going to be considering the cells' intersections in the following order of limits:

\begin{equation} \label{thickenedCellIntersection}
\lim_{\varepsilon \to 0} \lim_{\epsilon_{m+1} \to 0^+} P_{\{i_0 \dots i_p\}} \cap \Tilde{P}^\text{thick}_{\{j_0 \dots j_q\}}(\epsilon_{m+1};\varepsilon)
\end{equation}

\subsection{Statement of the main proposition on $\cup_m$} \label{mainPropCup_m}
Now, we are in a position to begin to write down the main proposition of this note relating the $\cup_m$ formulas to the generalized intersections. Specifically given our $m$-parameter thickening, we want to find what the limit in Eq(\ref{thickenedCellIntersection}) equals. More specifically, if $\alpha$ is a $p$-cochain and $\beta$ is a $q$-cochain, then $\alpha \cup_m \beta$ will be an $(p+q-m)$-form, so we really care about which $(n+m-p-q)$-dimensional cells survive the limit. Recall for the regular cup product, we had that many pairs of $(n-p)$- and $(n-q)$-cells had limits of intersections that were `lower-dimensional' of dimension $(n-p-q-1)$ that survived the limit. Likewise, for these thickened cases, we'll have that there may be many pairs of $(n-p)-$cells whose intersection with the thickened $(n-q)$-cells limit to cells that have dimension less than $(n+m-p-q)$.

Now, before we state the main proposition, let's look more closely about what exactly the formula for the higher cup product is saying. For general indices, it reads that for a $p$-cochain $\alpha$ and a $q$-cochain $\beta$:

\begin{equation} \label{higherCupFormula2}
(\alpha \cup_m \beta)(i_0,\dots,i_{p+q-m}) = \sum_{\{j_0 < \dots < j_m\} \subset \{i_0,\dots,i_{p+q-m}\}} \alpha(i_0 \to j_0, j_1 \to j_2, \dots) \beta(j_0 \to j_1, j_2 \to j_3,\dots)
\end{equation}

For this subsection, we will refer to $j_\gamma \to j_{\gamma +1}$ as $\{j_\gamma,\dots j_{\gamma+1}\} \cap \{i_0,\dots,i_{p+q+m}\}$ where $\#\{i_0 \to j_0, j_1 \to j_2, \dots \} = p+1$ and $\#\{j_0 \to j_1, j_2 \to j_3, \dots \} = q+1$. Writing this statement in terms of the dual cells, this means that only pairs of cells of the form $P_{\{i_0 \to j_0, j_1 \to j_2, \dots\}}$ and $P_{\{j_0 \to j_1, j_2 \to j_3, \dots\}}$ will contribute to $(\alpha \cup_m \beta)$. Let's think about what kind of restrictions this would imply for general cells. Let's call the sets $J_1 := \{i_0 \to j_0, j_1 \to j_2, \dots\}$, $J_2 := \{j_0 \to j_1, j_2 \to j_3, \dots\}$. 

Note that the two sets $J_1$ and $J_2$ will always share exactly $m+1$ indices $\{j_0,\dots,j_m\}$. We'll have that  any $i \in  \{i_0,\dots,i_{p+q-m}\} \textbackslash \{j_0,\dots,j_m\}$ will be contained in some interval $j_k < i < j_{k+1}$ for some $k \in \{0 \dots m-1\}$. The forms of the sets $J_1, J_2$ tell us that $i \in J_1$ iff k is even, and $i \in J_2$ iff k is odd.

Now we are ready to state our main proposition.

\begin{prop} \label{mainProp}
Choose two cells $P_{K}$ and $P_{L}$ of $\Delta^n$, where $K =  \{k_0 < \dots < k_p\}$ has $(p+1)$ elements and $L = \{\ell_0 < \dots < \ell_q\}$ has $(q+1)$ elements. Let's say that $K \cup L = \{i_0 < \dots < i_r\}$ has $r+1$ elements. Then there exists a set of linearly independent vectors $\Vec{b}_i, i \in \{1,m+1\}$ such that, given $\Tilde{P}^\text{thick}_{L}(\epsilon_{m+1};\varepsilon)$ as defined in Eq(\ref{thickenedCellDefinition}), the following statements hold

\begin{enumerate}
    \item If $r \neq p+q-m$, then $\lim_{\varepsilon \to 0} \lim_{\epsilon_{m+1} \to 0^+} P_{\{i_0 \dots i_p\}} \cap \Tilde{P}^\text{thick}_{\{j_0 \dots j_q\}}$ will be empty or consist of cells whose dimensions are lower than $n+m-p-q$
    \item If $r = p+q-m$, then $K$ and $L$ share $(m+1)$ elements which we'll denote $j_0 < \dots < j_m$.
    \begin{enumerate}
        \item If $K = \{i_0 \to j_0, j_1 \to j_2, \dots\}$ and $L = \{j_0 \to j_1, j_2 \to j_3, \dots\}$, then $\lim_{\varepsilon \to 0} \lim_{\epsilon_{m+1} \to 0^+} P_{K} \cap \Tilde{P}^\text{thick}_{L}(\epsilon_{m+1};\varepsilon) = P_{K \cup L}$.
        \item Otherwise, $\lim_{\varepsilon \to 0} \lim_{\epsilon_{m+1} \to 0^+} P_{K} \cap \Tilde{P}^\text{thick}_{L}(\epsilon_{m+1};\varepsilon)$ will be empty.
    \end{enumerate}
\end{enumerate}
Furthermore, we can choose the $\Vec{b}_i$ so that any subset of $n$ vectors chosen from the set $\{\Vec{b}_1, \dots, \Vec{b}_{m+1}, (c-f_0),\dots,(c-f_n)\}$ are linearly independent.
\end{prop}

We can readily verify Part 1 of the proposition.

\begin{proof}[Proof of Part 1 of Proposition \ref{mainProp}] 
Let us first verify the case that $r > p+q-m$. Note that $\lim_{\varepsilon \to 0} \lim_{\epsilon_{m+1} \to 0^+} P_{K} \cap \Tilde{P}^\text{thick}_{L}$ should be a subset of $P_K \cap P_L = P_{K \cup L}$. This is immediate from the definition of the Cauchy limit of sets, since $\lim_{\varepsilon \to 0} \lim_{\epsilon_{m+1} \to 0^+} \Tilde{P}^\text{thick}_{L} = P_{L}$. So if $r > p+q-m$, then $P_{K \cup L}$ will be of dimension $n-r < n+m-p-q$.

Now, if $r < p+q-m$, then we'll show that each $P_{K} \cap \Tilde{P}^\text{thick}_{L}$ is empty for $\epsilon_{m+1} \neq 0$, so there can't be any intersection points at all. For this second case, we need the property that any subset of $n$ vectors chosen from the set $\{\Vec{b}_1, \dots, \Vec{b}_{m+1}, (c-f_0),\dots,(c-f_n)\}$ are linearly independent. Let us write $\Tilde{P}^\text{thick}_{L}(\epsilon_{m+1})$ to indicate that this hyperplane is a function of $\epsilon_{m+1}$. $P_{K}$ is a subset of an $(n-p)$-dimensional plane $Q_1$ in $\Delta^n$ and $\Tilde{P}^\text{thick}_{L}(\epsilon_{m+1})$ is a subset of an $(n+m-q)$-dimensional plane $Q_2(\epsilon_{m+1})$ with $Q_2(0)$ containing $P_L$. 

Since $r < p+q-m$, then $Q_1$ and $Q_2(0)$ will share an $n-r > n+m-p-q$ dimensional subspace, consisting of the points of the plane containing $P_{K \cup L}$. However, we claim that for any $\epsilon_{m+1} \neq 0$, $Q_1 \cap Q_2(\epsilon_{m+1})$ is empty, which would then imply that $P_{K} \cap \Tilde{P}^\text{thick}_{L}$ is empty. Note that $Q_1 = \{c + t_1(f_{\hat{k}_1} - c) + \dots t_{n-p} (f_{\hat{k}_{n-p}}-c) | t_i \in \R \}$ and $Q_2 = \{(c + \epsilon_{m+1} b_{m+1}) + s_1(f_{\hat{l}_1} - c) + \dots s_{n-q} (f_{\hat{l}_{n-q}} - c) + s_{n-q+1} b_1 + \dots + s_{n-q+m} b_m | s_i \in \R\}$, where $\{\hat{k}_1,\dots,\hat{k}_{n-p}\} = \{1,\dots,n\} \textbackslash K$ and $\{\hat{\ell}_1,\dots,\hat{\ell}_{n-q}\} = \{1,\dots,n\} \textbackslash L$. We'll have that of the $(n-p)+(n-q+m)$ vectors in $\{f_{\hat{k}_1} - c, \dots, f_{\hat{k}_{n-p}} - c\} \cup \{f_{\hat{\ell}_1} - c, \dots, f_{\hat{\ell}_{n-q}} - c\} \cup \{b_1,\dots,b_m\}$, $(n-r)$ of these are repeated. So, there are $(n+m+r-p-q)$ unique vectors, which we may call $\{w_1,\dots,w_{n+m+r-p-q}\}$. Finding where $Q_1$ and $Q_2(\epsilon_{m+1})$ intersect amounts to solving the equation

$$\epsilon_{m+1} b_{m+1} + s_1 w_1 + \dots + s_{n+m+r-p-q} w_{n+m+r-p-q} = 0$$

But, since $r < p+q-m$, we'll have $n+m+r-p-q < n$. And since $b_{m+1}$ is not contained in $\{w_i\}$, this would imply that $\{b_{m+1},w_1,\dots,w_{n+m+r-p-q}\}$ (of size $\le n$) are linearly dependent. But, since $\epsilon_{m+1} \neq 0$, solving these equations contradicts the fact that we chose the $b$ so that any subset of $n$ of the $\{\Vec{b}_1, \dots, \Vec{b}_{m+1}, (c-f_0),\dots,(c-f_n)\}$ are linearly independent. So $Q_1 \cap Q_2 (\epsilon_{m+1}) = 0$, meaning that $\lim_{\varepsilon \to 0} \lim_{\epsilon_{m+1} \to 0^+} P_{K} \cap \Tilde{P}^\text{thick}_{L}$ is empty.
\end{proof}

To verify Part 2 in the case where $r=p+q-m$ we need to do some more work and then actually construct vectors $\Vec{b}$ with the desired properties. But, let us observe that we'll only have to worry about the case when $r=p+q-m=n$, i.e. when $K \cup L = \{0,\dots,n\}$. This is because if $K \cup L = \{i_0,\dots,i_r\}$ with $\{\ihat_1,\dots,\ihat_{n-r}\} = \{0,\dots,n\} \textbackslash \{i_0,\dots,i_r\}$, then we can restict to the subsimplex $\Delta' = \Delta_{\{i_0,\dots,i_r\}} = \Delta^n \cap \{(x_0,\dots,x_n) | x_{\ihat_1} = \dots = x_{\ihat_{n-r}} = 0\}$ and consider the intersection question on that subsimplex. We can similarly define the dual cells associated to and the $\Delta'$ and consider how the Morse Flows and thickenings act on those cells. We can analyze this by defining the center, $c'$, and the centers of the faces, $f'_j$, of $\Delta'$ and explicitly writing the cell decompositions in terms of these variables. If we do this, it will be immediate that the restriction to $\Delta'$ of the cells' intersections limit to the center $c'$ iff they limit to $P_{K \cup L}$ thoughout $\Delta$. This is because the shifted cells are all parallel to the original cells, so if the intersection on that boundary cell is nonempty, then the intersection throughout $\Delta$ will be a either be a shifted version of $P_{K \cup L}$ that limits to $P_{K \cup L}$. If it is empty, then it won't contain $c'$ and will be a lower dimensional cell.

Also note that while we will explicitly construct the fields inside each simplex and the fields won't necessarily match on the simplices' boundaries. But, we expect that the observations of Section \ref{flowedSimplices} will also apply to these constructions allowing us to define the vector fields continuously on a simplicical complex, so we wouldn't have to worry about any additional intersections that come from these boundary mismatches.

\subsection{Proof of Part 2 of Main Propostion} \label{proofOfMainPropCup_m}
Now, let us set-up our main calculation for the case of $K \cup L = \{0,\dots,n\}$. For the rest of this section, we will use the variables $i,j$ to label the cells and $k$ to label coordinates, and they won't be related to our previous usages. 

Recall that we want to find the intersection points of the cell 

\begin{equation} \label{mainProofCellEqn1}
{P}_{\{i_0 \dots i_p\}} = \Delta^n \cap \{c + (f_{\ihat_1} - c)s_1 + \dots + (f_{\ihat_{n-p}} - c)s_{n-p} | s_1,\dots,s_{n-p} \ge 0 \}
\end{equation}

with the shifted cell 

\begin{equation} \label{mainProofCellEqn2}
\begin{split}
\Tilde{P}^\text{thick}_{\{j_0 \dots j_{n+m-p}\}}(\epsilon_{m+1};\varepsilon) = \Delta^n \cap \{& \epsilon_1 \Vec{b}_1 + \cdots + \epsilon_{m+1} \Vec{b}_{m+1} - \frac{1}{n+1}(\epsilon_1 B_1 + \cdots \epsilon_{m+1}B_{m+1}) \cdot \Vec{b}_0 \\
& + c + (f_{\jhat_1} - c)t_1 + \dots + (f_{\jhat_{n-p}} - c)t_{p-m} | t_1,\dots,t_{p-m} \ge 0 , -\varepsilon \le \epsilon_1, \dots, \epsilon_{m} \le \varepsilon\}
\end{split}
\end{equation}

where we define $\{\ihat_1,\dots,\ihat_{n-p}\} = \{0,\dots,n\} \textbackslash \{i_0,\dots,i_p\} $ and $\{\jhat_1,\dots,\jhat_{p-m}\} = \{0,\dots,n\} \textbackslash \{j_0,\dots,j_{n+m-p}\}$.

So, we want to solve the equations, 

\begin{equation}
\begin{split}
c + (f_{\ihat_1} - c)s_{\ihat_1} + \dots + (f_{\ihat_{n-p}} - c)s_{\ihat_{n-p}} = & \epsilon_1 \Vec{b}_1 + \cdots + \epsilon_{m+1} \Vec{b}_{m+1} - \frac{1}{n+1}(\epsilon_1 B_1 + \cdots \epsilon_{m+1}B_{m+1}) \cdot \Vec{b}_0 \\ 
& + c + (f_{\jhat_1} - c)t_{\jhat_1} + \dots + (f_{\jhat_{p-m}} - c)t_{\jhat_1}, \quad \text{or} \\
c(1-s_{\ihat_1} - \dots - s_{\ihat_{n-p}}) + f_{\ihat_1} s_{\ihat_1} + \dots + f_{\ihat_{n-p}} s_{\ihat_{n-p}} = & \epsilon_1 \Vec{b}_1 + \cdots + \epsilon_{m+1} \Vec{b}_{m+1} - \frac{1}{n+1}(\epsilon_1 B_1 + \cdots \epsilon_{m+1}B_{m+1}) \cdot \Vec{b}_0 \\ 
& + c(1-t_{\jhat_1} - \dots - t_{\jhat_{p-m}}) + f_{\jhat_1} t_{\jhat_1} + \dots +  f_{\jhat_{p-m}} t_{\jhat_{p-m}}
\end{split}
\end{equation}

We have $n$ equations for $x_1,\dots,x_n$ (the $x_0$ equation is redundant since $x_0 + \dots + x_n = 1$). And, we have $n$ variables $\{s_{\ihat_1},\dots,s_{\ihat_{n-p}},t_{\jhat_1},\dots,t_{\jhat_{p-m}},\epsilon_1,\dots,\epsilon_m\}$ to solve for. It'll be convenient to change variables, and instead solve for $\{S_{\ihat_1},\dots,S_{\ihat_{n-p}},T_{\jhat_1},\dots,T_{\jhat_{p-m}},A_1,\dots,A_m\}$, defined as 

\begin{equation}
\begin{split}
s_{i}        &= \epsilon_{m+1} S_i, \text{ for } i \in \{\ihat_1,\dots,\ihat_{n-p}\} \\
t_{j}        &= \epsilon_{m+1} T_j, \text{ for } j \in \{\jhat_1,\dots,\jhat_{p-m}\} \\
\epsilon_{i} &= \epsilon_{m+1} A_i, \text{ for } i \in \{1,\dots,m\}
\end{split}
\end{equation}

When we expand out each equation for $x_k$, $k=1,\dots,n$, we get that

\begin{equation}
\begin{split}
&\frac{1}{n}\big(S_{\ihat_1}\delta_{k \neq \ihat_1} + \dots + S_{\ihat_{n-p}}\delta_{k \neq \ihat_{n-p}} \big) - \frac{1}{n+1}\big(S_{\ihat_1} + \dots + S_{\ihat_{n-p}}\big)\\
 = &\frac{1}{n+1}(B_1 A_1 + \dots + B_m A_m + B_{m+1}) - b_{1 k}A_1 - \dots - b_{m k} A_m - b_{m+1, k}  \\
& +\frac{1}{n}\big(T_{\jhat_1}\delta_{k \neq \jhat_1} + \dots + T_{\jhat_{p-m}}\delta_{k \neq \jhat_{p-m}} \big) - \frac{1}{n+1}\big(T_{\jhat_0} + \dots + T_{\jhat_{p-m}} \big)
\end{split}
\end{equation}

We can then multiply by $n(n+1)$ and do some rearranging to give the equations

\begin{equation}
\begin{split}
&\big(S_{\ihat_1} + \dots + S_{\ihat_{n-p}}\big) - (n+1)S_k \delta_{k \in \{\ihat\}} \\
&= \big(T_{\jhat_1} + \dots + T_{\jhat_{p-m}}\big) - (n+1)T_k \delta_{k \in \{\jhat\}} \\
&\quad - n(n+1)\big(b_{1 k} A_1 + \dots + b_{m k} A_m + b_{m+1,k} \big) + n(B_1 A_1 + \dots + B_m A_m + B_{m+1})
\end{split}
\end{equation}

where $\delta_{k \in \{\ihat\}}$ is $1$ if $k \in \{\ihat_1,\dots,\ihat_{n-p}\}$ and $0$ otherwise, and similarly for $\delta_{k \in \{\jhat\}}$. We are also abusing notation above, since we only defined $S_k$ for $k \in \{\ihat_1,\dots,\ihat_{n-p}\}$ in the first place. But, this is inconsequential since the term would vanish anyways for $k \notin \{\ihat_1,\dots,\ihat_{n-p}\}$. 

We will find it convenient to cast these equations in a more symmetric form by a change of variables. First, let us define the sets:

\begin{equation}
\begin{split}
\{\lambda_0,\dots,\lambda_m\} &= \{i_0,\dots,i_p\} \cap \{j_0,\dots,j_{n+m-p}\}, \text{ and } \\
\{\hat{\lambda}_1,\dots,\hat{\lambda}_{n-m}\} &= \{0,\dots,n\} \textbackslash \{\lambda_0,\dots,\lambda_{m}\}
\end{split}
\end{equation}

Note that $\{\hat{\lambda}_1,\dots,\hat{\lambda}_{n-m}\} = \{\ihat_1,\dots,\ihat_{n-p}\} \sqcup \{\jhat_1,\dots,\jhat_{p-m}\}$. And, let us redefine the variables for $k \in \{\hat{\lambda}_1,\dots,\hat{\lambda}_{n-m}\}$,

\begin{equation}
Z_k = 
\begin{cases}
S_k  &\text{ if } k \in \{\ihat_1,\dots,\ihat_{n-p}\} \\
-T_k &\text{ if } k \in \{\jhat_1,\dots,\jhat_{p-m}\}
\end{cases}
\end{equation}

Then, we can rewrite our equations in their final form as:

\begin{equation} \label{mainProofFinalForm}
\begin{split}
&(Z_{\hat{\lambda}_1} + \dots + Z_{\hat{\lambda}_{n-m}}) - n(B_1 A_1 + \dots + B_m A_m + B_{m+1}) \\
&= (n+1) Z_k \delta_{k \in \{\hat{\lambda}\}} - n(n+1)\big(b_{1 k} A_1 + \dots + b_{m k} A_m + b_{m+1,k} \big)
\end{split}
\end{equation}

where $\delta_{k \in \{\hat{\lambda}\}}$ is 1 if $k \in \{\hat{\lambda}\}$ and 0 otherwise. While these may again seem tricky to solve, some computer algebra experimentation shows that they have an elegant solution in terms of the $b_i$. Namely, the solutions are:

\begin{equation} \label{mainSolZ}
\begin{split}
Z_{\hat{\lambda}} &= n \frac{\det
\begin{pmatrix}
1      & b_{1 \lambda_0}     & \cdots & b_{m+1,\lambda_0}     \\
\vdots &  \vdots             &        & \vdots                \\
1      & b_{1 \lambda_m}     & \cdots & b_{m+1,\lambda_m}     \\
1      & b_{1 \hat{\lambda}} & \cdots & b_{m+1,\hat{\lambda}} \\
\end{pmatrix}
}{\det
\begin{pmatrix}
1      & b_{1 \lambda_0} & \cdots & b_{m \lambda_0} \\
\vdots & \vdots          &        & \vdots          \\
1      & b_{1 \lambda_m} & \cdots & b_{m \lambda_m} \\
\end{pmatrix}
} \\ 
&\text{ where } \hat{\lambda} \in \{\hat{\lambda}_1, \dots, \hat{\lambda}_{n-m} \}
\end{split}
\end{equation}

\begin{equation}\label{mainSolA}
\begin{split}
A_{\ell} &= (-1)^{m-\ell+1}\frac{\det
\begin{pmatrix}
1      & b_{\Tilde{\imath}_1 \lambda_0} & b_{\Tilde{\imath}_2 \lambda_0} & \cdots & b_{\Tilde{\imath}_m \lambda_0} \\
1      & b_{\Tilde{\imath}_1 \lambda_1} & b_{\Tilde{\imath}_2 \lambda_1} & \cdots & b_{\Tilde{\imath}_m \lambda_1} \\
\vdots &            \vdots              &        \vdots                  &        &        \vdots                  \\
1      & b_{\Tilde{\imath}_1 \lambda_m} & b_{\Tilde{\imath}_2 \lambda_1} & \cdots & b_{\Tilde{\imath}_m \lambda_m} \\
\end{pmatrix}
}{\det
\begin{pmatrix}
1      & b_{1 \lambda_0} & b_{2 \lambda_0} & \cdots & b_{m \lambda_0} \\
1      & b_{1 \lambda_1} & b_{2 \lambda_1} & \cdots & b_{m \lambda_1} \\
\vdots &    \vdots       &   \vdots          &        &   \vdots          \\
1      & b_{1 \lambda_m} & b_{2 \lambda_1} & \cdots & b_{m \lambda_m} \\
\end{pmatrix}
}\\ 
&\text{ where } \ell \in \{1,\dots,m\} \text{ and } \{\Tilde{\imath}_1,\dots,\Tilde{\imath}_m\} = \{1,\dots,m+1\} \textbackslash \{ \ell \}
\end{split}
\end{equation}

We prove that these formulas solve Eq(\ref{mainProofFinalForm}) in Appendix \ref{proofOfVandermonde}. But for now, let's explore their consequences. We only care about the solutions to the $Z_{\hat{\lambda}}$. Recall that since we were choosing $\epsilon_{m+1} \to 0^+$ we'll have that $\epsilon_{m+1} > 0$. In terms of our original variables, we want each of the $s_{\ihat} > 0$ and each $t_{\ihat} > 0$. So, since $s_{\ihat} = \epsilon_{m+1} S_{\ihat}$ and $t_{\ihat} = \epsilon_{m+1} T_{\ihat}$, we will want to impose that each $S_{\ihat} > 0$ and each $T_{\jhat} > 0$. This translates to saying that we want to find solutions when $Z_{\hat{\lambda}} > 0$ if $\hat{\lambda} \in \{ \ihat\}$ and $Z_{\hat{\lambda}} < 0$ if $\hat{\lambda} \in \{ \jhat \}$. 

Now is when we pick our matrix $b_{uv}$. A nice choice to connect to the higher cup formula will be:

\begin{equation} \label{higherCupSolnMatrix}
b_{uv} = (\frac{1}{1+v})^{u}
\end{equation}

Given this choice, we'll have that $Z_{\hat{\lambda}}$ can be written in terms of the ratio of Vandermonde determinants:

\begin{equation}
\begin{split}
Z_{\hat{\lambda}} &= n \frac{\det
\begin{pmatrix}
1      & \frac{1}{1+\lambda_0}     & \cdots & (\frac{1}{1+\lambda_0})^{m+1}    \\
\vdots &              \vdots       &        &  \vdots                           \\
1      & \frac{1}{1+\lambda_m}     & \cdots & (\frac{1}{1+\lambda_m})^{m+1}    \\
1      & \frac{1}{1+\hat{\lambda}} & \cdots & (\frac{1}{1+\hat{\lambda}})^{m+1}\\
\end{pmatrix}
}{\det
\begin{pmatrix}
1      & \frac{1}{1+\lambda_0} & \cdots & (\frac{1}{1+\lambda_0})^m \\
\vdots &         \vdots        &        &  \vdots       \\
1      & \frac{1}{1+\lambda_m} & \cdots & (\frac{1}{1+\lambda_m})^m \\
\end{pmatrix}
}  \\
&=  n (\frac{1}{1+\hat{\lambda}}-\frac{1}{1+\lambda_0}) \cdots (\frac{1}{1+\hat{\lambda}}-\frac{1}{1+\lambda_m})
\end{split}
\end{equation}

Note that each of the factors $(\frac{1}{1+\hat{\lambda}}-\frac{1}{1+\lambda_k})$ is positive iff $\lambda_k > \hat{\lambda}$. So, this implies that if $\hat{\lambda} < \lambda_0$, then $Z_{\hat{\lambda}} > 0$. And in general, if $\lambda_k < \hat{\lambda} < \lambda_{k+1}$, then $Z_{\hat{\lambda}} < 0$ iff $k$ is even and $Z_{\hat{\lambda}} > 0$ iff $k$ is odd. But this is exactly the condition that we wanted to show to relate this to the higher cup product formula! 

More specifically, for valid solutions to the intersection equations where $S_{\ihat} > 0$ and $T_{\jhat} > 0$, we'll need that $\{i_0,\dots,i_p\} = \{0 \to \lambda_0,\lambda_1 \to \lambda_2, \dots \}$ and $\{j_0,\dots,j_p\} = \{\lambda_0 \to \lambda_1,\lambda_2 \to \lambda_3,\dots\}$ so that each $Z_{\ihat} > 0$ and each $Z_{\jhat} < 0$. This shows that the only cells with solutions to the intersection equations are exactly the pairs that appear in the higher cup product formulas.

And, it is straightforward to check that any $n$ of the $\{\Vec{b}_1,\dots,\Vec{b}_{m+1}, c-f_0,\dots,c-f_{n}\}$ are linearly independent.

We also note that there are many related choices of the $b_{uv}$ that reproduce the higher cup products. Really, choosing 

$$b_{uv} = g(v)^u$$

for any positive function $g(v) > 0$ satisfiying $g(v) > g(w)$ if $v < w$ works, and the same Vandermonde argument applies. We also note that the solutions for the $A$ are certain Schur polynomials in the $g(v)$. The signs of the $A$ can thus be determined using known formulas for Schur polynomials: so we can determine which sides of the thickened cell the intersection happens.

\section*{\LARGE{\underline{Interpreting the GWGK Grassmann Integral}}} 
\addcontentsline{toc}{section}{\LARGE{\underline{Interpreting the GWGK Grassmann Integral}}}

Now, let's discuss how this geometric viewpoint of the higher cup product can be used to give in general dimensions a geometric interpretation of the GWGK integral as formulated in \cite{GaiottoKapustin} for triangulations of $Spin$ manifolds, and extended by Kobayashi \cite{KobayashiPin} to non-orientable $Pin^-$ manifolds. Apart from the conceptual interpretation will give us two practical consequences would be helpful in doing computations. First, we will be able to give equivalent expressions of the Grassmann integrals of \cite{GaiottoKapustin, KobayashiPin} without actually using Grassmann variables. Next, we will be able to formulate the Grassmann integral on any branched triangulation of a manifold, whereas in, \cite{GaiottoKapustin, KobayashiPin}, only the cases of a barycentric subdivision were considered, which will have many more cells than a typical triangulation.

In two dimensions, the geometric meaning of the Grassmann integral was explained in Appendix A of \cite{GaiottoKapustin}. We also note that entirely analogous ideas of considering combinatorial $Spin$ structures in two dimensions were developed in \cite{CimasoniReshetikhin,CimasoniNonorientable} in the context of solving the dimer model, a statistical mechanics problem whose solution can be phrased in terms of Grassmann integration.

\section{Background and Properties of the GWGK Grassmann integral} \label{backgroundPropertiesOfGuWenGrassmannIntegral}
Let's discuss some background material and some formal properties of the Grassmann integral that we will want to reproduce. First, we'll need to start out with some background material on how geometric notions like $Spin$ and $Pin$ structures may be encoded on a triangulation. After this, we'll recall the formal properties of the Grassmann integral that we'll want to reproduce with geometric notions. We won't give its detailed definition on a barycentrically subdivided triangulation here and refer the reader to \cite{GaiottoKapustin, KobayashiPin}. But we won't need its definition to proceed with out discussion.

\subsection{Spin/Pin structures and $w_1, w_2$, $w_1^2$ on a triangulated manifold}

Let $E \to M$ by a vector bundle over $M$ of rank $m$. We'll denote by $w_i(E) \in H^i(M,\Z_2)$ the $i^{th}$ Stiefel-Whitney class of $E$ over $M$. When $E$ is the tangent bundle $TM$, we'll often refer to $w_i(TM)$ as $w_i$ and as the Stiefel-Whitney classes of $M$. The way we'll choose to think about the Stiefel-Whitney classes is via their obstruction-theoretic definitions as follows. Choose a frame of $(m-i+1)$ `generic' sections of $E$. The condition of being generic means that they are linearly independent almost everywhere, and the locus of points on $M$ where they are linearly dependent will form a closed codimension-$i$ submanifold of $M$. This locus of points will be Poincaré dual to some cohomology class $w_i(E) \in H^i(M,\Z_2)$. This obstruction theoretic definition will be useful for us because it will help us make use of the vector fields we defined earlier in this note. 

Now, we'll want to figure out how to represent $w_1,w_2$ on a simplicial complex. We'll note first that the Poincaré duals of $w_p$ will be more naturally defined as chains living on the simplices themselves, i.e. as elments of $C_{n-p}(M,\Z_2) = C^p(M^\vee,\Z_2)$. This is in contrast to the simplicial cochains we considered previously, whose duals were naturally defined by chains living on the dual cellulation, $C_{n-p}(M^\vee,\Z_2) = C^p(M^,\Z_2)$. We can see this by noting the simplest case, of $w_1$. 

A canonical definition of $w_1$ is that it's represented by the set of all $(n-1)$-simplices for which the branching structure gives adjacent $n$-simplices opposite local orientations. In particular, this is encoded for us by noting that if a vector field frame reverses orientations between adjacent $n$-simplices, then the orientation must reverse upon passing their shared $(n-1)$-simplex. These $(n-1)$-simplices taken together will be Poincaré dual to the cohomology class $w_1$. A manifold is orientable iff the sum of such $(n-1)$-simplices are the boundary of some collection of $n$-simplices. If not, we may typically choose a simpler representative than this canonical one to do calculations. So while this is a canonical way to define $w_1$, it may practically be helpful to choose a representation with fewer simplices. 

In general, similar constructions for formulas for chain representatives of \textit{any} Stiefel-Whitney class on a branched triangulation are have been known since the 1970's, like in \cite{GoldsteinTurner}. For a barycentrically subdivided triangulation, the answer is particularly simple \cite{HalperinToledo}: that \textit{every} $(n-i)$-simplex is part of $w_i$. This is one reason that the GWGK Integral in \cite{GaiottoKapustin, KobayashiPin} was more readily formulated on a barycentrically subdivided triangulation. We expect that the vector fields constructed above are closely related to these older constructions, but we have not explicitly found the relationship and the formulas of \cite{GoldsteinTurner} may not apply directly to us. 

Soon, we will see a way to use our vector fields to give a canonical definition of $w_2$ on a branched triangulation. But for now, it'll be helpful to talk about $Spin$/$Pin$ structures on a triangulated manifold. A quick review of $Spin$ and $Pin^\pm$ groups are given in Appendix \ref{spinPinAppendix}.

One can generally define a $Spin$ structure on a vector bundle $E \to M$ for which $w_1(E)$ and $w_2(E)$ vanish in cohomology. A $Spin$ structure of $E \to M$ is a cochain $\eta \in C^1(M^\vee,\Z_2)$ with $\delta \eta = w_2(E)$. We say $\eta$ is a $Spin$ structure of $M$ if it's $Spin$ structure of $TM$. Note that a $Spin$ structure of $M$ needs that $M$ is orientable and is only defined if its $w_2$ vanishes.

Let $\det(TM)$ denote the determinant line bundle on $TM$. We can use $\det(TM)$ to characterize $Pin^{\pm}$ structures on $M$ (c.f. \cite{KirbyTaylor}). A $Pin^+$ structure may be defined on orientable or nonorientable manifolds. It has the same obstruction condition as a $Spin$ structure, and can be repesented by a cochain $\eta$ s.t. $\delta \eta = w_2$. Equivalently, it can be thought of as a $Spin$ structure on $TM \oplus 3 \det(TM)$ And, a $Pin^-$ structure may also be defined on orientable or nonorientable manifolds. The obstruction condition is different from the other ones, and is defined by a cochain $\eta$ s.t. $\delta \eta = w_2 + w_1^2$. It can also be thought of as a $Spin$ structure on $TM \oplus \det(TM)$. In general, such structures on a manifold are considered to be equivalent iff they differ by a coboundary. So, $Spin$/$Pin$ strucutres on $M$ are in bijection with $H^1(M,\Z_2)$. 

Another way to think a $Spin$ structure is to consider how we restrict $E \to M$ to the 1-skeleton and the 2-skeleton of $M$. $Spin$ structure on $E$ can be thought of as a frame of $(m-1)$ linearly independent sections of $E$ over the 1-skeleton, which can be extended (generically) over the 2-skeleton of $M$, possibly becoming linearly dependent at some even number of points within each 2-cell \footnote{Note that if $E$ is orientable, then we can put a nonvanishing positive definite metric on $E$, which given our $(m-1)$ linearly independent vectors, define a trivialization of $E$ over the 1-skeleton. An example of a manifold where the $(m-1)$ sections must be linearly dependent at some points inside a 2-cell is the tangent bundle of $S^2$, since any generic vector field will vanish at two points on the sphere.}. 

This is consistent with our obstruction theoretic definition, since it would be impossible to arrange this if every generic set of $(m-1)$ sections vanishes a total of an odd number of times on the 2-skeleton. And, this gives us a hint of how to construct a canonical representative of $w_2(E)$. In particular, suppose we chose some trivialization of $TM$ over the 1-skeleton where a generic extension becomes linearly dependent at some number of points $k$ on some 2-cell, $P$. Then, we'll have that $w_2(E)(P)=0$ if $k$ is even and $w_2(P)=1$ if $k$ is odd. So this gives a chain representative of $w_2$. So if we can always construct some trivialization of $E$ over the 1-skeleton and we know how to compute how many points vanish on each 2-cell, we'll have gotten our representative of $E$. 

We'll see later on that we can construct such a framing this canonically for $E=TM$, which will be a `canonical' chain representative of $w_2(TM)$. Although $w_1$ and $w_2$ have canonical chain representatives which can be expressed solely in terms of the branching structure, $w_1^2$ (as far as we know) does not have such an intrinsic chain-level definition. $w_1^2$ is a `self-intersection' of the orientation reversing wall, which can be defined by perturbing $w_1$ by a generic vector field and seeing the locus where it intersects its perturbed version. So, to define $w_1^2$, we need to specify a vector field to perturb along. The reference \cite{KobayashiPin} encodes this self-intersection in their definition of the Grassmann integral. Similarly in our geometric construction of a $Pin^-$ structure, such a choice will be encoded in the user's choice of a trivialization of $TM \oplus \det(TM)$, which we'll see equivalently encodes this perturbing vector. So given a branching structure and this additional user choice, we can represent $w_1^2$. 

We'll also see how given these framings, we can encode $Spin/Pin^-$ structures as adding `twists' in the background framing, which change the background framing into extending to even-index singularities on each $2$-cell $P$. We'll see that this can only be arranged if $w_2 + w_1^2$ is trivial. In Appendix \ref{combDefinePinPlus}, we give the construction for $Pin^+$ structures.

\subsection{Formal Properties of the GWGK Grassmann integral} \label{formalPropertiesGuWen}
Now, let's recall the formal properties of the GWGK Integral that we'll need to reproduce. In this section, when we denote by $M$ some manifold, we'll implicitly think of $M$ as encoding a triangulated manifold equipped with some branching structure. Formally, the GWKG Integral $\sigma(M,\alpha)$ depends on a branched triangulation, $M$, of some $n$-manifold and some closed cochain $\alpha \in Z^{n-1}(M,\Z_2)$. Note that elements $\alpha \in Z^{n-1}(M,\Z_2)$ are dual to some sum of closed loops on the dual graph. These loops are physically meant to represent worldlines of the fermions in this Euclidean setting.

On an orientable manifold, $\sigma(M,\alpha)$ takes values in $\Z_2 = \{\pm 1\}$. On a nonorientable manifold, we'll have $\sigma(M,\alpha)$ takes values in $\Z_4 = \{\pm 1, \pm i\}$, and $\sigma(M,\alpha) = \pm i$ iff $\int \alpha \cup w_1 = 1$. The definition of $\sigma(M,\alpha)$ depends on the (canonical) chain representative of $w_2$ and the (user-defined) chain representative of $w_1^2$. Given this, the main properties of $\sigma(M,\alpha)$ are:

\begin{enumerate}
    \item Suppose $\lambda \in C^2(M,\Z_2)$ is Poincaré dual to an elementary 2-cell of the dual complex, so that $\delta \lambda$ is dual to the boundary of an elementary cell. Then $\sigma(M,\delta \lambda)$ = $(-1)^{\int (w_2+w_1^2) (\lambda)} = (-1)^{\int_\lambda w_2+w_1^2}$, which is 1 if $(w_2+w_1^2)$ is zero on $\lambda$ and $-1$ if $(w_2+w_1^2)$ is nonzero on $\lambda$.
    \item (quadratic refinement) $\sigma(a)\sigma(b) = (-1)^{\int a \cup_{n-2} b} \sigma(a+b)$ 
\end{enumerate}

These two properties uniquely define $\sigma$ on homologically trivial loops. For homologically nontrivial loops, it is not determined by the above properties. So to compute the Grassmann integral for nontrivial loops, if we have the value of $\sigma$ for some loops that form a representative basis of $H_1(M,\Z_2)$, then we can use the quadratic refinement property to define it for any sum of closed curves on $M$.

Now we should consider how $\sigma$ changes under a re-triangulation or a bordism. Suppose $M_1 \sqcup \bar{M_2} = \partial N$, so that $N$ is some triangulated bordism between $M_1$ and $M_2$. And, suppose that $\alpha \in Z^{n-1}(N,\Z_2)$ is a gauge field that restricts to $\alpha_{1,2}$ on $M_{1,2}$. Then, arguments of \cite{GaiottoKapustin, KobayashiPin} show that 

\begin{equation} \label{guWenBordism}
\sigma(M_1,\alpha_1) = \sigma(M_2,\alpha_2) (-1)^{\int_N Sq^2(\alpha) + (w_2 + w_1^2)(\alpha)}
\end{equation}

A special case is that if the manifold is admits a $Pin^-$ structure and is $Pin^-$ null-bordant, we have the following formula (c.f. \cite{GaiottoKapustin}):

\begin{equation}
\sigma(\delta\lambda,M) = (-1)^{\int_M \lambda \cup_{d-3} \delta \lambda + \lambda \cup_{d-4}\lambda + (w_2+w_1^2) (\lambda)}
\end{equation}

Note that this formula only works if $(w_2 +w_1^2)$ is trivial on $N$, since otherwise shifting $\lambda \to \lambda + \mu$ for some $\mu$ with $\int (w_2 + w_1^2) \cup \mu = 1$ will change the integral by a factor $-1$.

Now, let's comment on why the Grassmann integral is important in the context of spin-TQFT's. It is due to the fact that under a cobordism the Grassmann integral changes by a factor of $(-1)^{\int_N Sq^2(\alpha) + (w_2 + w_1^2)(\alpha)}$, which can be thought of as a `retriangulation anomaly'. We can consider coupling the theory to a $Spin$ or $Pin^-$ structure, depending on whether $w_1 = 0$. This would entail finding some cochain $\eta$ with $\delta \eta = w_2$ or $\delta \eta = w_2 + w_1^2$. Then, the combination $z_{\Pi}(M,\eta,\alpha) := \sigma(M,\alpha)(-1)^{\int \eta \cup \alpha}$ will change by a factor of $(-1)^{\int_N Sq^2(\alpha)}$ under a cobordism. So, coupling to a $Spin$ structure cancels part of this retriangulation anomaly. Note that for $M$ a 2-manifold, $Sq^2$ kills all 1-forms, so the factor $(-1)^{\int_N Sq^2(\alpha)}$ is trivial. This means that for 2-manifolds, $z_{\Pi}(M,\eta,\alpha)$ is invariant under bordisms. 

In fact, in \cite{GaiottoKapustin, KobayashiPin} they show that the sum of $z_{\Pi}(M,\eta,\alpha)$ over all possible loop configurations $\alpha$ can be identified precisely with the Arf invariant, or Arf-Brown-Kervaire invariant of a $Spin$ or $Pin^-$ manifold, which exactly classifies the bordism class of a 2-manifold equipped with a $Spin$/$Pin^-$ structure.

\section{Warm up: Geometric interpretation of $\sigma(M,\alpha)$ in 2D} \label{geometricGuWenIn2D}
Now, let's review the geometric interpretation of the Grassmann integral in two dimensions. This was reviewed in an Appendix of \cite{GaiottoKapustin} for the case of orientable surfaces, but was also known earlier in a slightly different context, in \cite{CimasoniReshetikhin,CimasoniNonorientable}. In particular, the observations and pictures drawn in \cite{CimasoniNonorientable} will be helpful in extending this understanding to the case of nonorientable manifolds, both in two and higher dimensions. 

\subsection{Orientable surfaces and the winding of a vector field}

We will start by focusing on the story for orientable surfaces. First, we will describe a pair of linearly independent vectors along the 1-skeleton. Then, we'll give our definition of $\sigma(M,\alpha)$, which is related to how many times the vector field winds with respect to the tangent vector of the loop. Then, we'll show that our definition of $\sigma(M,\alpha)$ satisfies both formal properties that we care about. Then after this subsection, we'll explain how to modify the picture for the case of nonorientable surfaces.

\subsubsection{Framing of $TM$ along the dual 1-skeleton and its winding along a loop}
Let's describe the frame of vectors we'll use along the dual 1-skeleton. For this purposes in this section, it will suffice to describe them pictorially. First, we can note that on an orientable surface with a branched triangulation, it is possible to consistently label the 2-simplices as either $+$ or $-$. A consistent labeling means that if two of the simplices are adjacent, then their labelings of $\pm$ will be the same iff the local orientations defined by the branching structures match. So, we choose some consistent labeling of the simplices. 

Given such a consistent labeling, the framing along the 1-skeleton can be described as in the Figure(\ref{vectorFieldOn2Simplices}). Away from the center of the 2-simplex, we'll have one vector that runs along the 1-skeleton and the vector `perpendicular' to it will be in an opposite direction of the arrow defining the branching structure. This vector field is related to the flow that we constructed earlier, in Fig(\ref{fig:2_Simplex_flowed}) for the 2-simplex. This is because when we deform the vector fields in the manner depicted close to the center, and one of the vectors will be pointing in the same direction as that flow. Note if there's a globally defined orientation, these vector fields will be consistent with each other when glued together on the boundaries of the adjacent simplices. However, in the nonorientable case, there will be some inconsistencies that occur when the representative of $w_1$ doesn't vanish one the simplices' shared boundary.

\begin{figure}[h!]
  \centering
  \includegraphics[width=\linewidth]{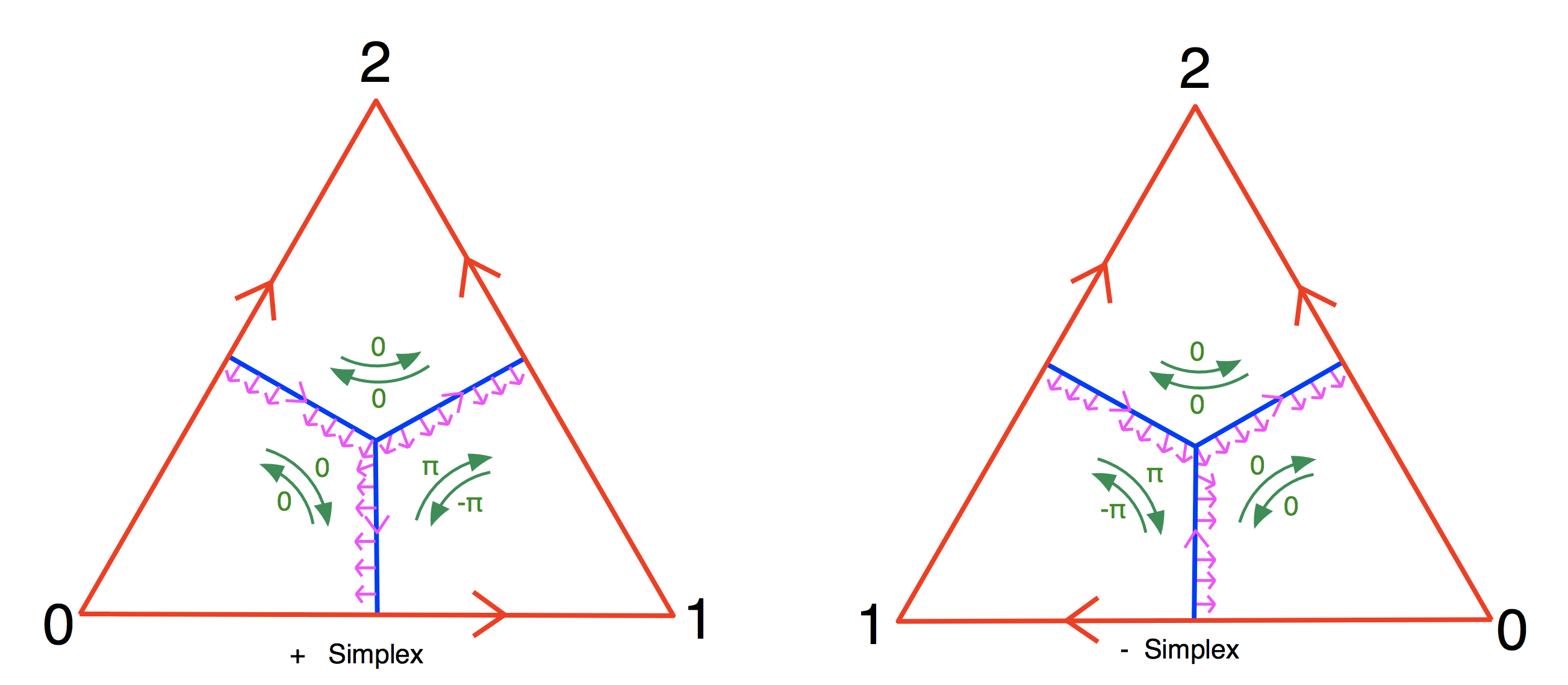}
  \caption{A pair of vector fields (in pink) along the 1-skeleton of a triangulation of a surface, for both positively-oriented (left) and negatively-oriented (right) 2-simplices. In this picture they're drawn `perpendicular' to each other. I.e. one vector field is parallel to the 1-skeleton away from the center and the other vector field is `perpendicular' to the 1-skeleton away from the middle. Note that in the center of the simplex, one of the vector fields is parallel to the flow vector depicted in Fig(\ref{fig:2_Simplex_flowed}). In green, we show the counterclockwise `winding angle' between these vector fields and the tangent vector of a curve that's restricted to the 1-skeleton.}
  \label{vectorFieldOn2Simplices}
\end{figure}

Also, if there's a global orientation means that we can talk about how many times this vector field frame `winds' in a counterclockwise direction with respect to the tangent vector of the loop. In Fig(\ref{vectorFieldOn2Simplices}), we show what these winding angles would look like for orientable manifolds. This winding will be crucial in constructing $\sigma(M,\alpha)$. Note that for the nonorientable case, we will have to be more careful in defining this winding, since `clockwise' and `counterclockwise' won't make sense. It will turn out that the analog of the `winding' can be expressed by a matrix, and these matrices won't necessarily commute.

\subsubsection{Definition of $\sigma(M,\alpha)$ in 2D and its formal properties}
Now, let us define $\sigma(M,\alpha)$ in two dimensions and show that it satisfies the formal properties we listed in Section \ref{formalPropertiesGuWen}. Given some closed cocycle $\alpha \in Z^1(M,\Z_2)$, we can represent it by some collection of curves on the dual 1-skeleton, which we'll denote $C_1,\dots,C_k$. Since the dual 1-skeleton is a trivalent graph, this decomposition into loops is unambiguous. For each curve $C_i$, define the quantity $wind(C_i)$ as the number of times the above vector field winds with respect to the tangent vector. Then the weight $\sigma(M,\alpha)$ will be defined:

\begin{equation}
\sigma(M,\alpha) = \prod_{i=1}^k (-1)^{1+wind(C_i)} = (-1)^{\text{\# of loops}}\prod_{i=1}^k (-1)^{wind(C_i)} 
\end{equation}

It's clear that this is well-defined, since $wind(C)$ will be the same mod 2 if we consider the curve going forwards as opposed to going backwards. Now let's see why this quantity satisfy the formal properties we cared about. First, we should show that a loop $C$ surrounding an elementary dual 2-plaquette, $P$, has a sign of $-1$ if $\int_P w_2 = 1$ and a sign of $-1$ if $\int_P w_2 = 0$. So, we should show

$$(-1)^{\int_P w_2} = \sigma(\alpha_C)$$

where $\alpha_C$ is the cochain representing the elementary plaquette loop $C$. The winding number definition will actually naturally (perhaps tautologically) satisfy this due to the obstruction theoretic definition of $w_2$. Suppose a vector field has winding number $wind(C)$ with respect to the tangent of a simple closed curve $C$. Then (depending on sign conventions) a generic extension of the vector field to the interior, $P$ of $C$ will vanish at $(\pm 1 \pm wind(C))$ points. So, our obstruction theoretic definition tells us that:

$$\int_{P} w_2 = 1 + wind(C) \quad \quad \text{(mod 2)}$$

which matches up with $\sigma(\alpha_C) = (-1)^{1+wind(C)} = (-1)^{\int_P w_2}$ for such elementary plaquette loops.

Next, we should show the quadratic refinement property, i.e. we should show for cochains $\beta$ and $\beta'$ that: 

$$\sigma(\beta)\sigma(\beta') = (-1)^{\int_M \beta \cup \beta'} \sigma(\beta+\beta')$$

So, the Grassmann integral of the sum of two cocycles will be the product of the Grassmann integrals of each summand, times this extra $(-1)^{\text{mod 2 intersection number of } \beta, \beta'}$. The argument for this is due to Johnson \cite{Johnson} who was studying the closely related notion of quadratic forms associated to 2D $Spin$ structures. 

\begin{figure}[h!]
    \centering
    \includegraphics[width=\linewidth]{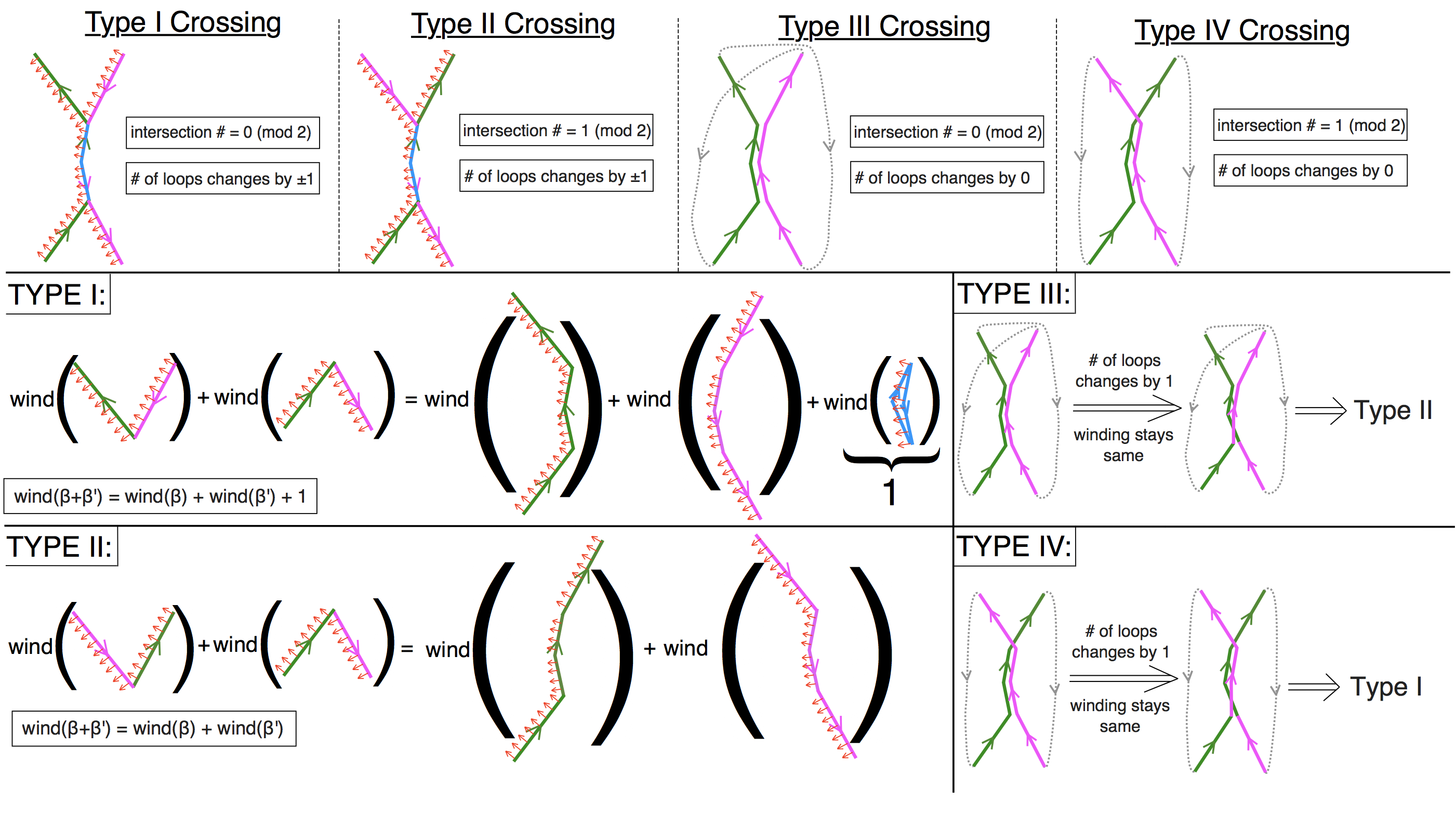}
    \caption{Different cases of shared segments of loops intersecting on a trivalent graph. Sums here are implicitly done modulo 2. The loop from $\beta$ is in green and the loop from $\beta'$ is in pink, the reference vector field is in red, and the intersection region that's shared by both $\beta$ and $\beta'$ is in blue. For the cases of Type I or Type II crossings, combining the loops after discarding the intersection region makes the number of loops change by 1. Type I intersections cause have intersection number zero and cause the winding number to change by 1 (mod 2) after resolving the intersection. Type II crossings have intersection number of one and cause the winding to stay the same. In general when $\beta$ and $\beta'$ share several different segments, we must choose the first resolved intersection to be either a Type I or II crossing. After the first intersection is resolved creating a combined loop, Types III and IV crossings can be resolved. Resolving a Type III or IV crossing doesn't change the number of loops, and is a two-step process where we reconnect the combined loop into two loops, which then allow us to reduce the resolution to either a Type II or I crossing, respectively.}
    \label{intersectionsAndWindings}
\end{figure}

Note that when the loops representing the cocycles never intersect, the formula is immediate, so we only need to consider  what happens when loops from the different cycles intersect each other. In particular, we'll want to visualize what happens to the windings when we combine the loops and discard the pieces that they both share. For these loops living on these kinds of trivalent graphs, loops intersecting will necessarily share some finite segment of edges. And in general, we'll have that the loops may intersect at a collection of more than one different segments of edges. The strategy is to resolve each intersecting segment of edges one at a time. So, we need to show that quadratic refinement holds as we resolve each intersection. We'll summarize the logic here, but refer to Fig(\ref{intersectionsAndWindings}) for a more detailed view.

For the first segment of intersections that is resolved, we have the freedom to change the directions of the curves so that they are directed oppositely to each other on their intersections of segments. The cases we'll need to distinguish are if the loops exit their shared line segments on the \textit{same} side of the shared segments, or on \textit{opposite} sides of the shared segments, which are labeled as Type I and Type II crossings In Fig(\ref{intersectionsAndWindings}). One can check that both cases change the number of loops by $\pm 1$. Type I crossings will contribute $0$ to the mod 2 intersection number, and Type II crossings will contribute $1$ to the mod two intersection number. And, Type I crossings change the total winding number by 1 whereas Type II crossings don't change the total winding number at all. This means that for Type I crossings, $(-1)^{\text{\# of loops}}(-1)^{\text{winding}}$ for the sum $\beta + \beta'$ is \textit{locally} the same as the $\cup$ products for $\beta$ and $\beta'$. Whereas for Type II crossings, $(-1)^{\text{\# of loops}}(-1)^{\text{winding}}$ for $\beta + \beta'$ differs \textit{locally} by a factor of $-1$ from the $\cup$ products for $\beta$ and $\beta'$. So, summing over all intersections, the quantity $\sigma(\beta + \beta') \{\sigma(\beta)\sigma(\beta')\}^{-1}$ will be the number of Type II crossings between $\beta$ and $\beta'$, which is just the mod 2 intersection number of $\beta$ and $\beta'$. This is precisely the statement of quadratic refinement.

If this segment was the only intersection region, then we're done. But now, we want to resolve the rest of the segments of intersections. Resolving the first intersection segment functioned as combining the two curves into one, and this combined curve may intersect itself in many different places. Some of these intersection regions look exactly like Type I or II crossings, for which the same logic applies as the previous paragraph. But there's also the possibility that the combined curve's shared regions are pointing in the same direction as each other, which are the Type III and IV crossings in Fig(\ref{intersectionsAndWindings}). We can resolve these intersections in a two-step process. First, reconnect the edges which turns the combined loop into two loops as in Fig(\ref{intersectionsAndWindings}). Then, for a Type III or IV crossing, after reversing one of these two reconnected loops we'll respectively get Type II and I crossings, which can then be resolved as such. Note that resolving these kinds of intersections ends up \textit{not} changing the number of loops, but the quadratic refinement property does hold after each such resolution.

\subsection{Nonorientable surfaces and `non-commuting' windings on $Pin^-$ surfaces}

Now, we will describe how to define $\sigma(M,\alpha)$ on nonorientable surfaces and see how we can connect it to the geometry of $Pin^-$ structures. This presentation is motivated by the entirely analogous ideas of \cite{CimasoniNonorientable}, who found a way to combinatorially encode the construction of \cite{KirbyTaylor} of $\Z_4$-valued quadratic forms on $Pin^-$ surfaces. Recall that a $Pin^-$ structure on $TM$ can be thought of as a $Spin$ structure on $TM \oplus \det(TM)$. So, we'll have that $\sigma(M,\alpha)$ will be related to some winding with respect to a trivialization of $TM \oplus \det(M)$ over $M$'s 1-skeleton.

First, we'll describe possible framings of $TM \oplus \det(TM)$ along the 1-skeleton and see how different choices of the framing can be related to different choices of the chains representing $w_1^2$. Then we'll define $\sigma(M,\alpha)$ and show how its formal properties match the ones we want.

Recall that the main issue in dealing with nonorientable surfaces is that it's not possible to consistently label 2-simplices as $+$ and $-$ with neighboring simplices having the same labeling iff their orientations locally agree. To deal with this, we'll just choose some labeling of $+$ and $-$ 2-simplices, and there will be some set of 1-simplices representing $w_1$ for which the local orientations don't match with their labeling. 

\subsubsection{Framing of $TM \oplus \det(TM)$ along the dual 1-skeleton} \label{sec:TMdetTMFraming}
Since we are adding an extra direct summand of $\det(TM)$ to the tangent bundle, it will be natural for us to visualize $TM \oplus \det(TM)$ at a point via a 2D plane parallel to the surface and a `third dimension' sticking out perpendicular to the plane. 

We've depicted such a framing for a positively oriented simplex in Fig(\ref{framingOfTM+DetTM}). Inside a 2-simplex, the framing of $TM \oplus \det(TM)$ along the 1-skeleton look similar to the framings in Fig(\ref{vectorFieldOn2Simplices}), except there will be an extra vector pointing in a direction `normal' to the surface, representing the framing of $\det(TM)$, in addition to the two vectors we had before, pointing in the directions along the surface. We'll refer to this vector along $\det(TM)$ as the `orientation vector'. We can give this framing an order by saying that the first (`$x$') vector is the orientation vector, the third (`$z$') vector is the one in $TM$ pointing along the 1-skeleton, and the second (`$y$') vector is the other vector along $TM$ pointing along the 1-skeleton, but transverse to the 1-skeleton.

Similarly, we can define the same kind of framing on a negatively oriented simplex, and as long as two neighboring simplices are not separated by a representative of $w_1$, this framing can be extended in the same way as the orientable case. The fact that $TM \oplus \det(TM)$ is always orientable ensures that this `normal' direction, or `orientation vector' along the dual 1-skeleton is well-defined on the interior of a 2-simplex.

\begin{figure}[h!]
    \centering
    \includegraphics[width=\linewidth]{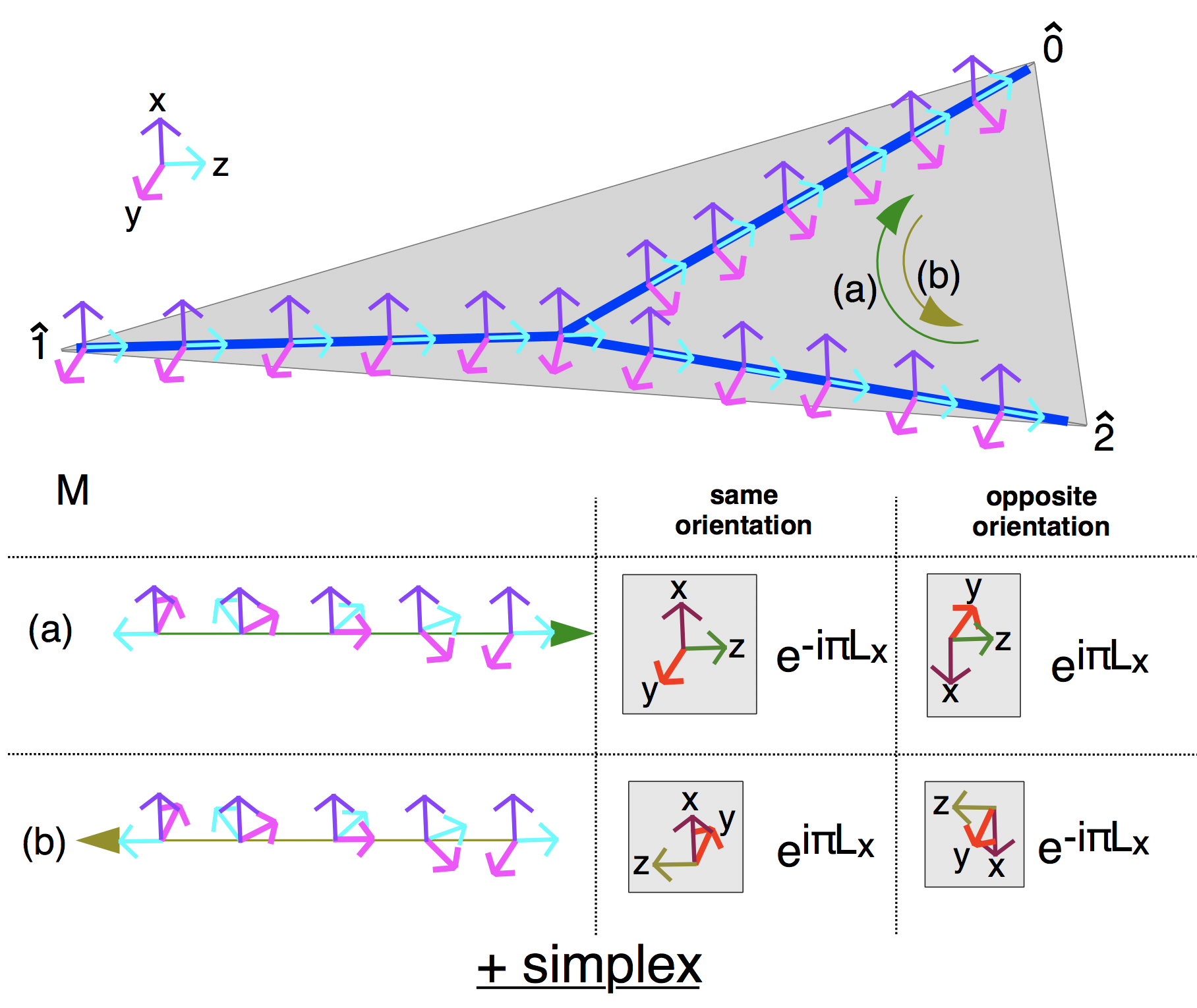}
    \caption{The framing of $TM \oplus \det(TM)$ along the dual 1-skeleton of a positively oriented simplex. The directions $y,z$ represent the coordinates for $TM$ going along the surface and the direction $x$ represents the coordinate for $\det(TM)$ which is depicted as `normal' to the surface. A curve that goes through this simplex can have $TM \oplus \det(TM)$ framed relative to these curves in two different ways as described in the main text, with the curve's `orientation vector' starting in either the same or opposite direction as that of the 2-simplex. ``Same Orientation'' refers to if the two orientation vectors start in the same orientation and ``opposite orientation'' means they start at different orientations. As one traverses along the directed curve, the curve's framing may change with respect to the framing along the 1-skeleton. The only directions in which the framing changes are listed as (a) and (b), as well as the explicit matrix under which the framing changes. Here $e^{\pm i \pi L_x}$ refers to a path in $SO(3)$ parameterized as $e^{\pm i t L_x}$ for $t:0 \to \pi$, whose endpoints for different choices of $\pm$ will be the same in $SO(3)$ but lift to different elements of $SU(2)$.}
    \label{framingOfTM+DetTM}
\end{figure}

Although these assignments can unambiguously determine framings inside each 2-simplex, there's an issue of what happens when the local orientation on $TM$ reverses, i.e. when we cross a 1-simplex where $w_1 \neq 0$. Since an orientation of $TM$ can't be defined everywhere, a trivialization of $TM \oplus \det(TM)$ requires that we rotate $TM$ and $\det(TM)$ into each other near $w_1$, where the orientation of $TM$ reverses. For each potential choice of mismatching framings, we can consider two different ways of extending them to match across $w_1$, as depicted in Fig(\ref{framingAcrossW1}). In particular, we'll only consider the possibilities of rotating into each other the `orientation vector' and the vector going along the dual 1-simplex. When doing this, our two choices to consider amount to our choice of which direction along the dual 1-simplex the orientation vector points as it traverses $w_1$. 

This choice can give us a choice of vector field transverse to the $w_1$ surface along the 1-skeleton as follows (see Fig(\ref{framingAcrossW1})). The orientation vector as it crosses $w_1$ will point to one side of $w_1$. On that side of $w_1$, we consider the background frame's `$z$' vector that usually points along the 1-skeleton. The direction this `$z$' vector points will determine the transverse direction. And, as depicted in Fig(\ref{w1Squared}), this choice of extensions of the framings across $w_1$ will define a representative of $w_1^2$. This is because if two adjacent dual edges have this vector pointing in opposite directions, then in between those two edges, there will be an odd number of self intersections of $w_1$, as defined by some extension of these vectors.

\begin{figure}[h!]
    \centering
    \includegraphics[width=\linewidth]{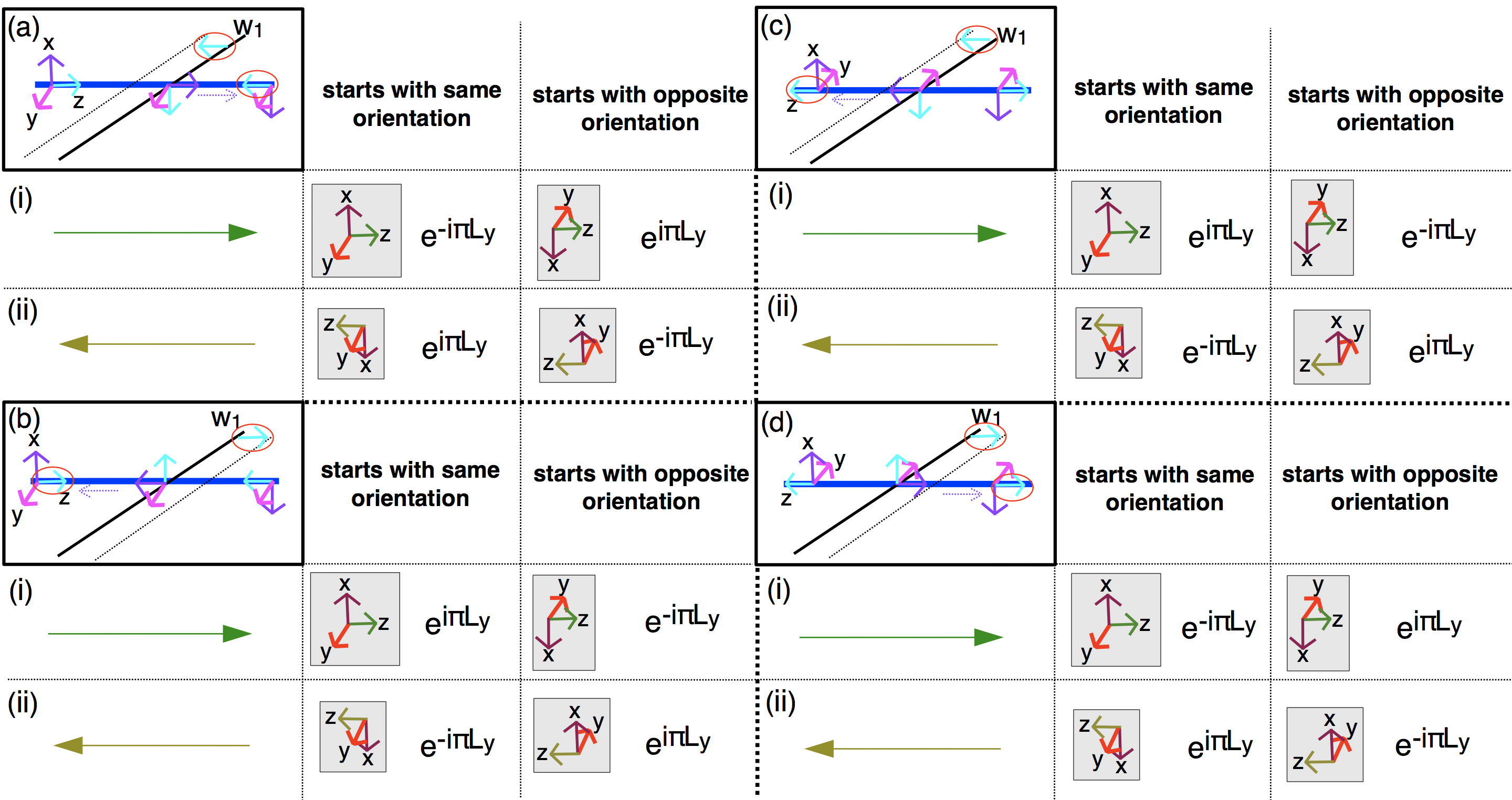}
    \caption{The framing of $TM \oplus \det(TM)$ when crossing a representative of $w_1$. Since an orientation of $TM$ can't be defined everywhere, we need to rotate $\det(TM)$ and $TM$ into each other across $w_1$. The pairs (a,b) and (c,d) are the different choices of extension for framings that match away from $w_1$. These different choices will end up corresponding to different windings with respect to the loop, whose different possibilities are listed in the table. In addition, they correspond to different choices of vectors transverse to $w_1$ at the 1-skeleton, which correspond shifting $w_1$ along a perturbing vector. This is depicted as the shift of the solid black line, representing $w_1$, to the dashed black line, representing its shift along the perturbing vector. The perturbing vector is given by the `$z$' vector along the 1-skeleton, on the side of the $w_1$ surface that the orientation vector points to as it traverses $w_1$, circled in red and pointed to by the purple dashed arrow, which points in the same direction as the orientation vector along $w_1$. }
    \label{framingAcrossW1}
\end{figure}

\begin{figure}[h!]
    \centering
    \includegraphics[width=0.5\linewidth]{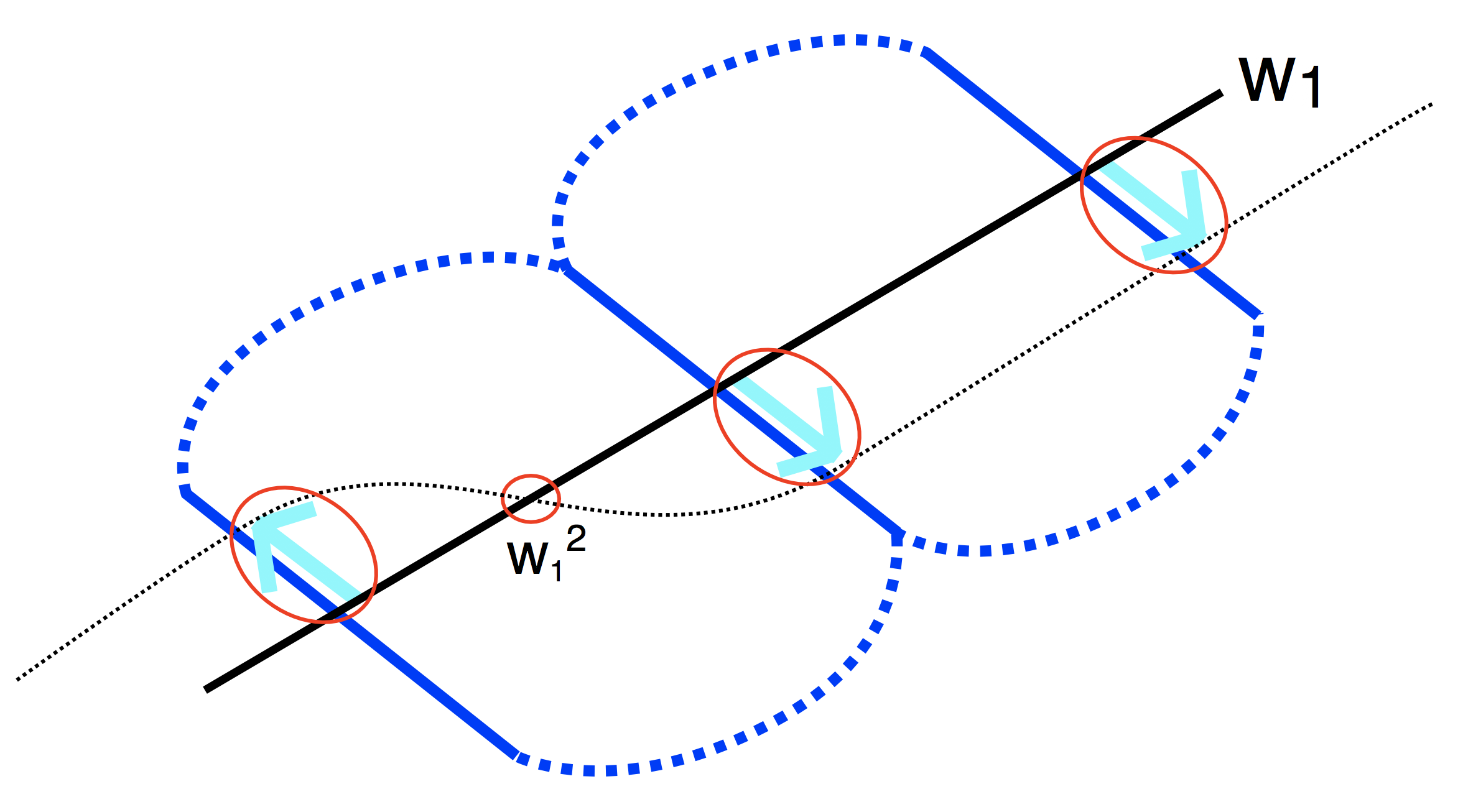}
    \caption{The extension of the framing $TM \oplus \det(TM)$ across $w_1$ defines vectors along the 1-skeleton transverse to $w_1$ as described in Fig(\ref{framingAcrossW1}). Some generic extension of this vector field will allow us to define a vector field along which we can define $w_1^2$, the self-intersection of $w_1$. In particular, $w_1^2$ is nonzero on some dual 2-plaquette iff the vector field points in opposite directions along the two links on the plaquette that intersect $w_1$.}
    \label{w1Squared}
\end{figure}

\subsubsection{Winding matrices around a loop}\label{2DWindingMatrices}

Now that we've defined a background framing of $TM \oplus \det(TM)$ along the 1-skeleton, we should think about how to define the winding matrices, the analog of the orientable case's winding matrices. 

\underline{\textbf{The relative framing around a loop}}

Recall in the orientable case, we compared the winding of the background vector field with respect to the tangent of the curve. For the nonorientable case, we'll be comparing this background framing of $TM \oplus \det(TM)$ to a certain `tangent framing' of $TM \oplus \det(TM)$ along a loop. This tangent framing along the loop is defined as follows. Pick a starting point along the curve on the interior of a 2-simplex. Then, first vector (`$x$') will be identified with the `orientation vector' of the framing, along the $\det(TM)$ part. The tangent vector of the curve will define the third (`$z$') vector of the tangent frame. And, the orientation of $TM \oplus \det(TM)$ will determine the second (`$y$') vector \footnote{Strictly speaking, we'd need to define some positive-definite metric on $TM \oplus \det(TM)$ to do this unambigously.}. 

Given the background framing of $TM \oplus \det(TM)$ and this tangent framing along some given loop $C$, we will want to compare these framings as we go along a loop. Assuming that, with respect to some positive definite metric, these frames are orthonormal, each point along the loop defines some element of $SO(3)$. So, going around the loop will mean that these relative framings define some path in $SO(3)$, i.e. a function $f_C: [0,1] \to SO(3)$. 

The possible changes in relative framing (i.e. changes in $f_C$) for a loop traversing inside a 2-simplex are given in (a,b) of Fig(\ref{framingOfTM+DetTM}) \footnote{We denote be $i L_x=\begin{pmatrix} 0 & 0 & 0 \\ 0 & 0 & -1 \\ 0 & 1 & 0 \end{pmatrix}, i L_y=\begin{pmatrix} 0 & 0 & 1 \\ 0 & 0 & 0 \\ -1 & 0 & 0 \end{pmatrix}, i L_z=\begin{pmatrix} 0 & -1 & 0 \\ 1 & 0 & 0 \\ 0 & 0 & 0 \end{pmatrix}$ the basis of the $\mathfrak{so}(3)$ Lie Algebra in its defining representation, satisfying $[i L_x,i L_y] = i L_z, [i L_y,i L_z] = i L_x, [i L_z,i L_x] = i L_y$ } for a positively oriented simplex, and can similarly be found for a negatively oriented simplex. And the path of winding going across $w_1$ are given in Fig(\ref{framingAcrossW1}). Note that the path of winding depends on the direction we traverse and also on the whether the orientation vector of the background framing agrees with orientation vector of the tangent framing. We will find it convenient to normalize $f_C$ so that $f_C(0) = \mathbbm{1}$, so that we only measure the change in framing from the start. 

Denote by $\alpha_C \in H^1(M,\Z_2)$ the cohomology class associated to $C$. We note that if $\int w_1 \cup \alpha_C = 0$, i.e. if $C$ crosses $w_1$ an even number of times, then $f_C(1)$ will be the identity, $f_C(1) = \mathbbm{1} \in SO(3)$. To see this, first note that the vectors along the 1-skeleton will be in the same relative orientation at the beginning and end of the loop. Next, note that crossing the $w_1$ surface an even number of times means that the tangent framing's orientation vector at the end of the loop will agree with the background framing's orientation vector, just like in the beginning of the loop. This means that the relative framings are the same at the beginning and end of the loop, i.e. that $f_C(1) = \mathbbm{1}$. 

By similar reasoning, we can note that if $\int w_1 \cup \alpha_C = 1$, i.e. $C$ crosses $w_1$ an odd number of times, then (relative to the coordinates `x,y,z')

$$f_C(1) = 
\begin{pmatrix}
-1 & 0 & 0 \\
0 & -1 & 0 \\
0 &  0 & 1 \\
\end{pmatrix} \in SO(3)$$

Again, the relative orientations of the vectors along the 1-skeleton isn't different at the beginning versus at the end. But, the orientation vectors at the end of the loop will have a relative sign change since there's such a sign change every time we cross $w_1$. This means that relatively between the initial and final frames, the `$x$' and `$y$' coordinates will flip sign.

\underline{\textbf{The quadratic form on a $Pin^-$ surface}}

Now, we should mention what exactly these relative framings have to do with a $\Z_4$-valued quadratic form. For motivation, let us recall the definition of the $\Z_4$-valued quadratic form of  Kirby/Taylor \cite{KirbyTaylor}. There, given some closed, non-self-intersecting loop $C$ in the $Pin^-$ surface $M$, we can consider the restriction $TM \oplus \det(TM)$ to $C$, which we call $\tau$. First, a $Pin^-$ structure on $M$ gives a trivialization of $\tau$. Now, denote $E$ as the total space of the bundle $\det(TM) \to M$. Another way to decompose $\tau$ is as:

$$\tau = TM \oplus \det(TM)|_{C} = TC \oplus N_{C \subset M} \oplus N_{M \subset E}$$

where $TC$ is the tangent bundle of the curve $C$ and $N_{A \subset B}$ denotes the normal bundle of a submanifold $A \subset B$. Note that we can trivialize $TC$ by considering tangent vectors in the direction we traverse. The definition of the quadratic form of \cite{KirbyTaylor} involves a comparing the framing induced by the $Pin^-$ structure's trivialization of $\tau$ with the framing induced by this bundle decomposition, in the same way that the orientable winding number was gotten by comparing two frames. To do this, first, we pick some framing of $\tau$ that's homotopic to the one induced by the $Pin^-$ structure and for which the third `$z$' vector lies on the curve's tangent. Then, we can choose the orientations of both framings to match each other at the starting point of the curve. Then the quadratic form $z(C)$ associated to $C$ is

$$z(C) = -i^{\text{number of right half-twists (mod 4)}}$$

where `number of right half-twists (mod 4)' is the number of right handed half-twists (mod 4) that $N_{C \subset M}$ makes traversing the loop compared to this background framing homotopic to the $Pin^-$ framing. \footnote{The (mod 4) factor comes in because different choices of framings homotopic to the $Pin^-$ one will differ by 4 in the number of right half-twists. This is because framings of the rank-3 bundle $\tau$ form a $\pi_1(SO(3))=\Z_2$ torsor, while framings of the rank-2 bundle $N_{C \subset M} \oplus N_{M \subset E}$ form a $\pi_1(SO(2))=\Z$ torsor. So, two framings of $N_{C \subset M} \oplus N_{M \subset E}$ that differ in $\Z$ by 2 will be homotopically the same framing of $\tau$ in $\Z_2$, and they will differ by 4 right half-twists going around.} The $-1$ in front corresponds to the $(-1)^\text{number of loops}$ factor that we had in the orientable case. 

It's shown in \cite{KirbyTaylor} that given some set of disjoint loops $C_1,\dots,C_k$ on $M$ representing $\alpha \in H^1(M,\Z_2)$ that the function

\begin{equation}
z(\alpha) = \prod_{i=1}^k z(C_i)
\end{equation}

doesn't depend on the representative curves of $\alpha$ so is a function on cohomology classes. And they also show the quadratic refinement property holds, that $z(\beta)z(\beta') = z(\beta + \beta') (-1)^{\int \beta \cup \beta'}$ for $\beta, \beta' \in H^1(M,\Z_2)$

\underline{\textbf{Our analog of the quadratic form on the 1-skeleton and how to compute it}}

Now, let's think about how this `number of right half-twists' is encoded in our function $f_C$. Since our function $f_C$ is a function $[0,1] \to SO(3)$ with $f_C(0) = \mathbbm{1}$, we'll be able to lift it to a unique function $\Tilde{f}_C: [0,1] \to SU(2)$ with $\Tilde{f}_C(0) = \mathbbm{1} \in SU(2)$. The homotopy class of the path $f_C$, and consequently the number of right half-twists it makes, is determined by the endpoint of its lift, i.e. by $\Tilde{f}_C(1)$. Recall we found earlier that 

$$f_C(1) = \mathbbm{1} \in SO(3) \quad \text{iff} \quad \int w_1 \cup \alpha_C = 0 $$

and that 

$$f_C(1) = 
\begin{pmatrix}
-1 & 0 & 0 \\
0 & -1 & 0 \\
0 &  0 & 1 \\
\end{pmatrix} \in SO(3)  \quad \text{iff} \quad \int w_1 \cup \alpha_C = 1$$

This implies that the endpoint of the lift $\Tilde{f}_C(1)$ can take the possible values

$$\Tilde{f}_C(1) = \pm \mathbbm{1} \in SU(2) \quad \text{iff} \quad \int w_1 \cup \alpha_C = 0$$

and \footnote{We denote by $X=\begin{pmatrix} 0 & 1 \\ 1 & 0 \end{pmatrix}, Y=\begin{pmatrix} 0 & -i \\ i & 0 \end{pmatrix}, Z=\begin{pmatrix} 1 & 0 \\ 0 & -1 \end{pmatrix}$ the lifts of $L_x,L_y,L_z$ to matrices of the $\mathfrak{su}(2)$ Lie Algebra in the fundamental representation.}

$$\Tilde{f}_C(1) = \pm i Z \in SU(2)  \quad \text{iff} \quad \int w_1 \cup \alpha_C = 1$$

The cases of $\Tilde{f}_C(1) = \{\mathbbm{1}, iZ, -\mathbbm{1},-iZ\}$ corresponds to the number of right half-twists being $\{0,1,2,3\}$ (mod 4). So, we can see that

\begin{equation} \label{defNumberRightHalfTwists}
\begin{pmatrix} 1 & 0 \end{pmatrix}
\Tilde{f}_C(1)
\begin{pmatrix} 1 \\ 0 \end{pmatrix}  = i^{\text{number of right half-twists}}
\end{equation}

Now, one may be concerned that this definition depends on things like the starting point of the curve and the direction we traverse the curve. We can show that this is not the case as follows. To show this, we'll first translate the windings of Figs(\ref{framingOfTM+DetTM},\ref{framingAcrossW1}), which denoted changes in $f_C$, into how they lift as corresponding changes of $\Tilde{f}_C$. 

We'll have to address that the winding on a part of the loop depends on relative direction of the orientation vector. To do this, it's convenient to introduce a 2-component tuple of orientations:

\begin{equation}
\mathcal{O} = 
\begin{pmatrix}
\mathcal{O}_\text{same} \\
\mathcal{O}_\text{opposite}
\end{pmatrix}
\end{equation}

Each of $\mathcal{O}_\text{same},\mathcal{O}_\text{opposite}$ will be in $SU(2)$ and only one component at a time will be nonzero. $\mathcal{O}_\text{same}$ being nonzero means that the orientation vectors agree between the background and tangent framings. And $\mathcal{O}_\text{opposite}$ being nonzero means that they disagree. 

As we said before, at the beginning of the loop we choose the orientation vectors to agree. So the intial tuple will be:

$$\mathcal{O}^\text{initial} = 
 \begin{pmatrix} \mathbbm{1} \\ 0 \end{pmatrix}
$$

And, at the end of the loop we'll have some tuple $\mathcal{O}^\text{final}$, from which we can extract $\Tilde{f}(1)$ as: 

\begin{equation}
\Tilde{f}(1) = \mathcal{O}_\text{same}^\text{final} + \mathcal{O}_\text{opposite}^\text{final} = \begin{pmatrix} \mathbbm{1} & \mathbbm{1} \end{pmatrix} \mathcal{O}^\text{final}
\end{equation}

To get from $\mathcal{O}^\text{initial}$ to $\mathcal{O}^\text{final}$, there will be some sequence of matrices $\{W_1,\dots,W_k\}$ so that $\mathcal{O}^\text{final} = W_k \cdots W_1 \mathcal{O}^\text{initial}$. Each $W_j$ will be 2$\times$2 blocks where each block is in $SU(2)$. In the cases where the part $j$ of the loop keeps the orientation vector relatively the same (i.e. parts within a 2-simplex), $W_j$ will be block-diagonal. And the parts where the orientation vector relatively switches (i.e. going across $w_1$), $W_j$ will be block-off-diagonal.

For a `$+$' simplex, we'll have:

$$W_j = 
\begin{pmatrix} 
-i X & 0 \\ 0 & i X
\end{pmatrix} \quad \text{if} \quad \hat{2} \to \hat{0}$$

$$W_j = 
\begin{pmatrix} 
i X & 0 \\ 0 & -i X
\end{pmatrix} \quad \text{if} \quad \hat{0} \to \hat{2}$$

$$W_j = 
\begin{pmatrix} 
\mathbbm{1}& 0 \\ 0 & \mathbbm{1}
\end{pmatrix} \quad \text{otherwise}$$

For a `$-$' simplex, we'll have: 

$$W_j = 
\begin{pmatrix} 
i X & 0 \\ 0 & -i X
\end{pmatrix} \quad \text{if} \quad \hat{2} \to \hat{0}$$

$$W_j = 
\begin{pmatrix} 
-i X & 0 \\ 0 & i X
\end{pmatrix} \quad \text{if} \quad \hat{0} \to \hat{2}$$

$$W_j = 
\begin{pmatrix} 
\mathbbm{1}& 0 \\ 0 & \mathbbm{1}
\end{pmatrix} \quad \text{otherwise}$$

And, we'll have a couple of cases to consider when crossing $w_1$ as we depicted in Fig(\ref{framingAcrossW1}).

$$W_j =
\begin{pmatrix} 
0 & iY \\ -iY & 0
\end{pmatrix} \quad \text{cases (a.i,b.ii,c.i,d.ii) of crossing  } w_1$$

$$W_j =
\begin{pmatrix} 
0 & -iY \\ iY & 0
\end{pmatrix} \quad \text{cases (a.ii,b.i,c.ii,d.i) of crossing  } w_1$$

Now, we can show that the number of right half-twists as defined in Eq(\ref{defNumberRightHalfTwists}) is well-defined: that it doesn't depend on the starting point of the curve and doesn't depend on the direction we go. One thing to note is that although the individual matrices $iX,iY,iZ$ don't commute with each other, all of the matrices $W_j$ will commute due to the block diagonal structure. This ensures that the matrix $\mathcal{O}^\text{final} = W_k \cdots W_1 \mathcal{O}^\text{initial}$ is independent of the starting point of the path. Another thing to note is that for a segment $j$ of the path, $W_j$ is the negative of the matrix gotten by traversing that part in the opposite direction. And, note that the total number of matrices $k$ will be even, since every part that contributes a nontrivial $W_j$ reverses the relative direction of the vector along the 1-skeleton, and this relative direction stays the same between the beginning and the end. So, reversing the path will change $\mathcal{O}^\text{final}$ by an even number of minus signs, i.e. it keeps $\mathcal{O}^\text{final}$ the same. 

We can also note that if we started out with the orientation vector of the tangent frame in the \textit{opposite} direction as the background frame, as opposed to the same direction, then this would also leave the number of right half-twists the same. This would amount to defining $\mathcal{O}^\text{initial} = \begin{pmatrix} 0 \\ \mathbbm{1} \end{pmatrix}$. This would give the same $\Tilde{f}_C(\mathbbm{1})$. We can see this because starting with this $\mathcal{O}^\text{initial}$ is equivalent to starting with $\begin{pmatrix} \mathbbm{1} \\ 0 \end{pmatrix}$ and conjugating $W_k \cdots W_1$ by $\begin{pmatrix} 0 & \mathbbm{1} \\ \mathbbm{1} & 0 \end{pmatrix}$. Conjugating each $W_k$ by this matrix introduces a minus sign, and there will be an even number of minus signs that all cancel.

\subsubsection{The definition of $\sigma(M,\alpha)$ and its formal properties} \label{formalPropsPinMinusSigma}
Now, suppose $\alpha$ is represented by some set of curves $C_1,\dots,C_k$ on the dual 1-skeleton. Then, we'll define the function $\sigma(M,\alpha)$ in a similar way as before:

\begin{equation}
\sigma(M,\alpha) = (-1)^{\text{\# of loops}} \prod_{i=1}^k i^{\text{number of right half-twists for } C_i}
\end{equation}

By the previous discussion, this quantity is well-defined. And, we can see that $\sigma(M,\alpha) = \pm 1$ iff $\int w_1 \cup \alpha = 0$ and $\sigma(M,\alpha) = \pm i$ iff $\int w_1 \cup \alpha = 1$, which we wanted.

\underline{\textbf{$\sigma$ for elementary plaquette loops}}

The next thing we should show is that if an elementary plaquette loop $C$ surrounds the plaquette $P$, then $\sigma(M,\alpha_C) = (-1)^{\int_P w_2 + w_1^2}$. Note that our definition of $w_2$ requires a vector field on $TM$ that is nonvanishing over the entire 1-skeleton. Of the vectors $x,y,z$ of the background framing, the only one that remains in $TM$ over the entire 1-skeleton is the $y$ field, which doesn't pay attention to $w_1$. So, $w_2$ is defined via the $y$ field.

Note that away from the $w_1$ surface, this computation is exactly the same as it was in the orientable case. But for $C$ that lie on $w_1$, the computation is subtle. The reason for this is that simplices that are neighboring each other across $w_1$ have their labels `$+,-$' that are inconsistent with their relative local orientations. This means that the winding number definition needs to be looked at with care to define $\int_P w_2$ since the labeling of the $+$ and $-$ simplices is `wrong' as far as measuring this winding is concerned. To treat this, let's consider the sequence of matrices $W_1,\dots,W_k$ that go into constructing $\mathcal{O}^\text{final} = W_k \cdots W_1 \mathcal{O}^\text{initial}$. And, let's say that the matrices at $i_0 < j_0$ correspond to the segments where the orientation reverses. (Here, we will treat the case where $w_1$ intersects the loop twice. The case of a higher even number of intersections is similar). 

Note that the matrices $W_{i_0}$ and $W_{j_0}$ are of the form $\pm \begin{pmatrix} 0 & iY \\ -iY & 0 \end{pmatrix}$ and the $W_{i_0+1},\dots,W_{j_0-1}$ are all of the form $\pm \begin{pmatrix} -i X & 0 \\ 0 & i X\end{pmatrix}$. An issue we need to deal with is that the $W_{i_0+1},\dots,W_{j_0-1}$ are all the \textit{negative} of what the local orientation would think, relative to the start of the curve. In other words, the winding part $(-1)^{\int_P w_2}$ of $\sigma(M,\alpha_C)$ would be given by the opposite of the sign $\pm \mathbbm{1}_{4 \times 4}$ of 

\begin{align*}
&W_1 \cdots W_{i_0-1} (-W_{i_0 + 1}) \cdots (-W_{j_0 - 1}) W_{j_0 + 1} \cdots W_k \\
&= (-1)^{j_0 - i_0 - 1} W_1 \cdots W_{i_0-1} W_{i_0 + 1} \cdots W_{j_0 - 1} W_{j_0 + 1} \cdots W_k
\end{align*}

From here, if we can verify that the sign of $(-1)^{j_0 - i_0 - 1} W_{i_0} W_{j_0}$ is equal to $(-1)^{\int_P w_1^2}$, then we've shown what we've want, that $\sigma(M,\alpha_C) = (-1)^{\int_P w_2 + w_1^2}$. This can be seen as follows, with the pictures in Fig(\ref{framingAcrossW1}) in mind. Note that $W_{i_0}$ or $W_{j_0}$ is $\begin{pmatrix} 0 & iY \\ -iY & 0 \end{pmatrix}$ if the direction the loop traverses is the same as the direction of the background orientation vector across $w_1$, and it's $\begin{pmatrix} 0 & -iY \\ iY & 0 \end{pmatrix}$ if the loop's direction is opposite that of the background orientation vector across $w_1$. This means that $W_{i_0} W_{j_0} = -\mathbbm{1}_{4 \times 4}$ if the background orientation vectors at the $i_0, j_0$ junctions near $w_1$ point to the same side of $w_1$ as each other, and $W_{i_0} W_{j_0} = \mathbbm{1}_{4 \times 4}$ if they point to opposite sides of $w_1$. Similarly, $(j_0 - i_0 - 1)$ gives the number (mod 2) of half-turns that the $y,z$ vectors make with respect to the curve on one side of $w_1$. So, $(-1)^{j_0-i_0-1}$ is related to the relative directions of the `$z$' vectors of the background framings near the $i_0, j_0$ junctions, on the same side of $w_1$. In particular, $(-1)^{j_0-i_0-1}$ is 1 if these directions are opposite, and it's $-1$ if the directions are the same. Combining these observations with the definition of the perturbing vectors and the defintion of $w_1^2$ shows that $(-1)^{j_0 - i_0 - 1} W_{i_0} W_{j_0} = (-1)^{\int_P w_1^2} \mathbbm{1}_{4 \times 4}$.

So, we've shown that $\sigma(M,\alpha_C) = (-1)^{\int_P w_2 + w_1^2}$.

\underline{\textbf{Quadratic refinement for $\sigma$}}

Fortunately, the quadratic refinement property for $\sigma$ follows from a similar analysis as with the orientable case, for which argument was depicted in Fig(\ref{intersectionsAndWindings}). As before the problem reduces to the case of when $\beta,\beta' \in H^1(M,\Z_2)$ are each represented by a single loop, $C,C'$ resp., on the dual 1-skeleton. The main point of extending that logic to the nonorientable case is that we should compare what happens to the total signs of the winding matrices for each curve before and after combining them on each intersection. 

In particular, let $W := W_k \cdots W_1$ be the total winding matrix for $C$ and $W' := W'_{k'} \cdots W'_{1}$ be the total winding matrix for $C'$. Since all the $W_j,W'_j$ commute with each other, we should consider the total matrix $W_\text{combined}$ is after combining them by resolving a single intersection. First, note that resolving each intersection changes the number of loops by $\pm 1$. Then, we should note that if $W W' = W_\text{combined}$, then the total number of half-twists stays the same (mod 4), and if $W W' = - W_\text{combined}$, then the total number of right half-twists will change by 2 (mod 4), which can be easily seen by examining the definition of the number of right half-twists in terms of the $W,W'$. So, the problem boils down to showing that $W W' = -W_\text{combined}$ for a Type I crossing and that $W W' = W_\text{combined}$ for a Type II crossing. 

If the shared part of the curve doesn't intersect $w_1$, this can also be seen in the same way we saw it in Fig(\ref{intersectionsAndWindings}). But if a shared part \textit{does} interesect $w_1$, then we should be more careful. Let's suppose first that the shared part intersects $w_1$ exactly once. Then, the local orientations at the ends of the shared part will be opposite to each other, so the analogous argument would tell us that there's a relative minus sign between our expected answer. But we should also consider the products of the winding matrices along the shared part. If the curve doesn't intersect $w_1$, then the winding going in one direction of the shared part exactly cancels the winding in the other direction. But since the curve intersects $w_1$ once, the winding matrices going in one direction times the winding in the other direction will actually be $-1$, which can be traced to the fact that 

$$-\begin{pmatrix} 0 & -iY \\ iY & 0 \end{pmatrix} \begin{pmatrix} 0 & -iY \\ iY & 0 \end{pmatrix} = -\mathbbm{1}$$

So this $-1$ from the local orientations will cancel the $-1$ from the winding matrices going in the opposite directions, which means that quadratic refinement still holds. 

\section{$\sigma(M,\alpha)$ in higher dimensions} \label{geometricGuWenInHigherDimensions}
Now, our goal will be to use the lessons from the 2D case and see how to extend this understanding to higher dimensions, for some triangulation $M$ of a $d$-dimensional manifold and some cocycle $\alpha \in Z^{d-1}(M,\Z_2)$. In summary, the basic idea for $\sigma(M,\alpha)$ will remain the same, that schematically:

$$\sigma(M,\alpha) = ``(-1)^{\text{\# of loops}} \prod_{\text{loops}} i^{\text{winding}}"$$

So, our goal will be to formulate what exactly we mean by these quantities and then show that they satisfy the formal properties we care about. One thing is that we need to decide what exactly this `winding' factor means. In 2D, there were two tangent vectors, so we could unambiguously decide what the winding angle or the winding matrices between the `background' and `tangent' framings were going around a loop. But in higher dimensions, it's not as clear how to do this. 

The other issue is we want to ensure is that there is a clear definition of the `loops' on the dual 1-skeleton. It's not immediately obvious that a loop decomposition makes sense, because the dual 1-skeleton in higher dimensions is $(d+1)$-valent. For a trivalent graph it's possible to decompose any cochain into loops, but for higher-valent graphs, there's an issue that if there are four or more edges at at a vertex then there are multiple ways of splitting these edges up into pairs to define the loops.

It will turn out that we will have a clear way to define this winding via a certain \textit{shared framing} along the 1-skeleton. In other words, across the entire 1-skeleton there will be some fixed set of $(d-2)$ vectors on the 1-skeleton that define a `shared framing' and are shared by both the background and tangent framings. Then, there will be two remaining vectors that will differ between the background and tangent framings, from which we can then unambiguously give a notion of winding. This need for additional framing vectors can be anticipated from another way to think about $Spin$ structures in higher dimensions. In particular, in higher dimensions a $Spin$ structure can be thought of assigning either a bounding or non-bounding $Spin$ structure to every \textit{framed} loop so that the $Spin$ structures changes when the framing is twisted by a unit. So, this shared framing will tell us that for any loop that passes through two edges of the dual 1-skeleton at a $d$-simplex, we can assign a winding to that segment of the loop. 

And, there's a way to deal with the problem of a $(d+1)$-valent dual 1-skeleton. To deal with this, one option to unambigously resolve a loop configuration is to resolve the $(d+1)$-valent vertex into $(d-2)$ trivalent vertices. In particular, we want to make sure that our trivalent resolution will allow $\sigma$ to satisfy the quadratic refinement property $\sigma(\beta) \sigma(\beta') = (-1)^{\int_M \beta \cup_{d-2} \beta'} \sigma(\beta + \beta')$. It turns out that for our purposes, this option works. We will see that there is a certain trivalent resolution of the dual 1-skeleton that will yield this property, even though generically not all trivalent resolutions will work.

It will turn out that the interpretation of the higher cup product as a thickening under vector field flows will be crucial in allowing these definitions to work. We will see that the quadratic refinement property can be readily deduced if we thicken the loops under the `shared' framing. And, we'll see that the vector used to shift the loops to resolve intersections will correspond to the `$y$' vector of the background framing. A nice feature of this will be that the constructions and arguments for the 2D case carry over directly to higher dimensions, and we can use the same geometric reasoning to deduce quadratic refinement in higher dimensions. So, we already did a large part of the work in spelling out the 2D case (apart from the trivalent resolution part).

Throughout this section, we will will want to label the edges of the dual 1-skeleton and the $(d-1)$-simplices that comprise of the boundary. As in the last section, for $i \in \{0,\dots,d\}$, we'll use $\ihat$ to refer interchangeably with the $(d-1)$-simplex that comprises of $(0 \dots \ihat \dots d)$ or its dual edge on the 1-skeleton. If the context isn't sufficient to distinguish the $(d-1)$-simplex with its dual edge, we'll refer to the $(d-1)$-simplex as $\ihat_\Delta$ and its dual edge as $\ihat_{ed}$

\subsection{The different framings} \label{theDifferentFramings}
We will first need to illustrate explicitly what all the different framings are that we'll be considering. In particular, we'll want to know how to describe along the 1-skeleton the $(d-2)$ vectors that go into shared framing, and the other two vectors each that go into the background and tangent framings. These will be closely related to the vector fields we constructed in describing the higher cup product. Then, we'll be in a position to see explicitly how to compute the winding of these frames with respect to each other as we enter and exit the $d$-simplex along two edges of the dual 1-skeleton. 

While we were able to do this pictorially in two dimensions, in higher dimensions we'll need to think more carefully to show the analogous statements in higher dimensions. Throughout, we'll see that some nice properties of the Vandermonde matrix will allow us to think about the windings and do the relevant computations.

Since we're dealing with $\Delta^d \subset \R^{d+1}$, the vectors we deal with in our vector fields will have $(d+1)$ components. And the ones that can lie within $\Delta^d$ will be the ones with $(1,\dots,1)$ projected out. However, it will be convenient algebraically to think about the vectors before projecting out $(1,\dots,1)$. So, we will introduce the notation:

\begin{equation}
v \sim w \quad \text{if } v = w + a (1,\dots,1), \text{ for some } a \in \R
\end{equation}

We'll give most of the details of the fields' definitions inside each $d$-simplex in the main text. But, we'll relegate some other details to Appendix \ref{windingsInHigherDims}, like how to glue the vector fields at neighboring simplices and how to compute their windings with respect to each other. 

\subsubsection{The shared framing}
Let's discuss first what is the shared framing that we referred to above and see how it relates to the vector fields we constructed to thicken and intersect the cells. Then, we'll talk about some of this framing that will be necessary for us.

Recall that Eq(\ref{higherCupSolnMatrix}) represented a set of vector fields we could use to connect to the higher cup product. And for $\beta,\beta'$ being $(d-1)$-cochains, the corresponding vectors that we thicken along inside the $d$-simplex are:

$$v^\text{shared}_i \sim \Vec{b}_i = (1,\frac{1}{2^i},\frac{1}{3^i}, \dots, \frac{1}{(d+1)^i}) \text{ for } i = 1,\dots,d-2$$

We'll also include the extra vector $\Vec{b}_0 = (1,\dots,1)$ for convenience.

Inside each $d$-simplex, these vectors will represent the shared framing that are common to both the background and tangent framings. In particular, along each of the $(d+1)$ edges of the dual 1-skeleton, away from the boundary the $d$-simplex these vectors will be constant. But, we'll want to modify the vectors as the points approach the boundary of the $d$-simplex, onto the $(d-1)$-simplex $\ihat_\Delta$. This is so that it will be possible to extend the framing to nearby $d$-simplices. So, as we approach $\ihat_\Delta$, we'll project all of these vectors onto the subspace of $\R^{d+1}$ with the $i$th component being zero (except for $\Vec{b}_0$ which always remains $(1,\dots,1)$). We do this in anticipation that we'll compare the vector fields on different simplices.

Note that the vectors will remain linearly independent as we project out this component. This linear independence follows from the fact that every $k \times k$ minor of a Vandermonde matrix consists of linearly independent $k$-component vectors, and the $\Vec{b}_j$ with any component projected out can be thought of as a $(d-1) \times (d-2)$ submatrix of a Vandermonde matrix.

Another important property is that this frame of vectors is linearly independent from the vectors tangent to any of the dual 2-cells. Let $(ij)$ denote the 2-cell dual to the $(d-2)$-simplex $(0 \dots \ihat \dots \jhat \dots d)$. To see this, note that the vectors that span this dual 2-cell are $\{(c-f_i),(c-f_j)\}$ as defined in Section \ref{sec:Prelim}. Also note that:

$$(c-f_i) \sim \frac{1}{n} (0,\dots0,\underbrace{1}_{i\text{th component}},0\dots,0)$$

So, the shared frame are linearly independent from $(ij)$ for the same reason: projecting out these $i$th and $j$th components leaves the frame linearly independent since that's saying that a $(d-2) \times (d-2)$ minor of a Vandermonde matrix has nonzero determinant.

\subsubsection{The background framing}
Now, let's define the other two vectors that go into the background framing. The first additional vector that we'll add will simply be the vector

$$v^\text{bkgd}_{n-1} \sim \Vec{b}_{n-1} = (1,\frac{1}{2^{n-1}}, \dots, \frac{1}{(d+1)^{n-1}})$$

And again, as we approach the boundary at the $(d-1)$-simplex $\ihat$, we project out the $i$th component of the vector. 

The second vector, $v^\text{bkgd}_n$ will be analogous to the earlier vector that points parallel to the dual 1-skeleton, except near the center of the $d$-simplex. To define this vector, we need to be careful about the direction along the 1-skeleton points, either towards or away from the center. So, we need to be sure that the orientation defined by the entire frame is consistent throughout the $d$-simplex. Let's recall how we did this for the 2-simplex, as in Fig(\ref{vectorFieldOn2Simplices}), for which the prescription differed for `$+$' and `$-$' simplices. For a `$+$' simplex, this vector along the 1-skeleton pointed towards the center along the edges $\hat{1}$ and pointed away from the center for the edges $\hat{0},\hat{2}$, and oppositely for a `$-$' simplex. The reason for this is by considering the induced orientations on the $(d-1)$-simplices $\ihat_\Delta$: the branching structure gives opposite orientations on the simplices labeled by $i$ even versus $i$ odd. 

So, away from the center, we'll have that for a `$+$' simplex,

$$v^\text{bkgd}_n = (-1)^i(c - f_i) \sim \frac{1}{n}(0,\dots,0,\underbrace{(-1)^i 1}_{i\text{th component}},0,\dots,0) \quad \text{ along } \ihat$$

with opposite signs for a `$-$' simplex. Now, while we can make these definitions along $\ihat$ away from the center of a $d$-simplex, we have to be careful when approaching their centers and making sure that we can make a continuous vector field in some neighborhood of the center. The solution to this is that we should first pick some neighborhood of the center whose shape is a $d$-simplex with vertices on each edge of the dual 1-skeleton at some same, small coordinate distance from the center. As the curve goes from edges $\ihat \to \jhat$, then $v^\text{bkgd}_n$ will look like $t (-1)^i(c - f_i) + (1-t)(-1)^j(c - f_j)$ where $t$ is some appropriate parameter of the curve between $\ihat \to \jhat$. 

The important point is that as we approach the center, it will be possible to arrange that:

$$v^\text{bkgd}_n \sim \Vec{b}_n = (1,\frac{1}{2^n}, \dots, \frac{1}{(d+1)^n}) \quad \text{ at the center}$$

We alluded to this previously in Fig(\ref{vectorFieldOn2Simplices}), where we demonstrated visually that a natural continuation of the vector field points in the same direction as the direction of the Morse flow at the center. 

Of course, we need to make sure these constructions make sense and indeed define a nondegenerate framing everywhere in a neighborhood of the 1-skeleton.  We'll verify this and put the constructions on more solid footing in Appendix \ref{windingsInHigherDims}.

\subsubsection{The tangent framing}
We can define the tangent framing in a similar way. Let's consider a path $\ihat \to \jhat$. Then the vector tangent to the curve, which we'll call $v^\text{tang}_n$, will start out as 

$$v^\text{tang}_n = (c - f_i) \sim \frac{1}{n}(0,\dots,0,\underbrace{1}_{i\text{th component}},0,\dots,0) \quad \text{ along } \ihat$$

and end as 

$$v^\text{tang}_n = (f_j - c) \sim \frac{1}{n}(0,\dots,0,\underbrace{-1}_{j\text{th component}},0,\dots,0) \quad \text{ along } \jhat$$

And in between, we'll have that

$$v^\text{tang}_n = t (c-f_i) + (1-t)(c - f_j) \quad \text{ in between } \ihat, \jhat$$

Now, away from the boundary, we can choose the other vector, $v^\text{tang}_{n-1}$, to be 

$$v^\text{tang}_{n-1} \sim \pm \frac{1}{n}(0,\dots,0,\underbrace{-1}_{j\text{th component}},0,\dots,0) \quad \text{ along } \ihat$$

$$v^\text{tang}_{n-1} \sim \pm \frac{1}{n}(0,\dots,0,\underbrace{-1}_{i\text{th component}},0,\dots,0) \quad \text{ along } \jhat$$

$$v^\text{tang}_{n-1} \sim \pm \frac{1}{n}(0,\dots,0,\underbrace{-(1-t)}_{i\text{th component}},0,\dots,0,\underbrace{-t}_{j\text{th component}},0,\dots,0) \quad \text{ in between } \ihat, \jhat$$

The choice of $\pm$ here will depend on a couple factors: whether the simplex is a `$+$' or `$-$' simplex, whether the `orientation vectors' agree or disagree, and the values of $i,j$. The details of this will be given in Appendix \ref{windingsInHigherDims}.

\subsection{Verifying the formal properties: windings, trivalent resolution} \label{verifyingTheFormalPropertiesHigherDimensions}
Given the background and tangent framing, we can ask what are their windings with respect to each other? In other words, the two framings determine a relative element of $SO(d)$ with each other, and we want to know how to determine this relative winding's path in $\pi_1(SO(d)) = \Z_2$. \footnote{Strictly speaking, these framings give a relative framing in $GL^+(d)$ and a loop determines an element of $\pi_1(GL^+(d)) = \Z_2$. But, first note that we can freely choose an appropriate inner product that makes the background framing orthonormal. Then, we can orthogonalize the tangent framing with respect to this inner product. This will give the same element of $\Z_2$, since $GL^+(d)$ deformation retracts onto $SO(d)$ via the orthogonalization procedure.}

Throughout this subsection, we'll again relegate to Appendix \ref{windingsInHigherDims} technical details. In fact, we won't need to explicitly state what the winding matrices are to describe the formal properties for now. But for reference, the windings and the trivalent resolutions are given in Fig(\ref{fig:trivalentResAndWinding}).

\begin{figure}[h!]
  \centering
  \includegraphics[width=\linewidth]{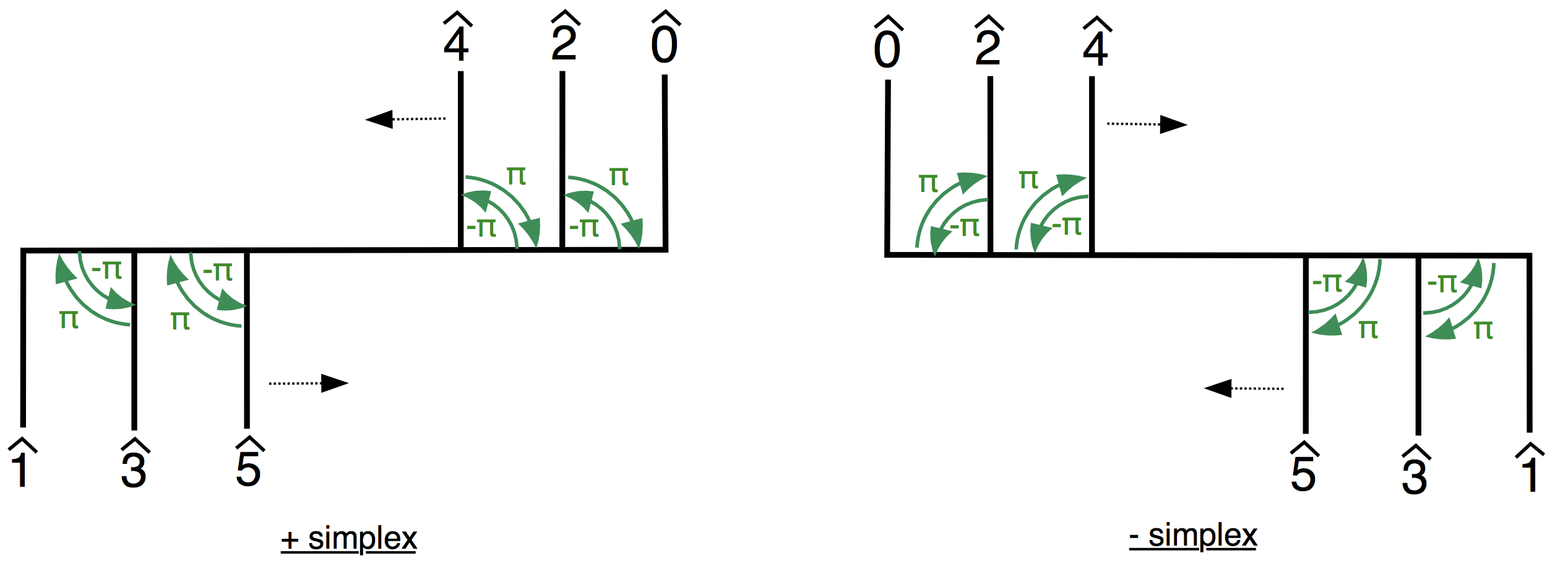}
  \caption{The trivalent resolution used, and the windings on both orientations. The $\pm \pi$ can refer to the winding angle in the orientable case, but more generally refers to the winding matrices of Section \ref{2DWindingMatrices}. Note that winding between $\ihat$ and $\jhat$ occur iff $i \equiv j \text{ (mod 2)}$ and depend on whether or not $i < j$.}
  \label{fig:trivalentResAndWinding}
\end{figure}

The first thing we should can do is verify that for some elementary plaquette loop $C$ bounding the dual 2-cell $P$, that $(-1)^{(w_2 + w_1^2)(P)} = \sigma(M,\alpha_C)$. The reason for this is precisely the same as it was before. For the orientable case, note that the tangent framing will always have two vectors spanning the plaquette's tangent bundle along the boundary. This means that the winding (mod 2) of the background frame with respect to the tangent frame determines the number of singularities of a generic extension that must occur, which shows for the orientable case that $(-1)^{w_2(P)} = \sigma(M,\alpha_C)$. The same argument for the nonorientable case also applies in the same manner as earlier, and gives us the additional $w_1^2$ part.

Next, we should verify the quadratic refinement part. This is the place where viewing the higher cup product as a thickening with respect to a vector field flow will be helpful in showing quadratic refinement. Again, the argument reduces to showing that $\sigma(\beta)\sigma(\beta') = \sigma(\beta + \beta')(-1)^{\int \beta \cup_{d-2} \beta'}$ for when $\beta$ and $\beta'$ are both dual to a single loop on the dual 1-skeleton. In 2D, we were able to verify this by considering the edges shared between the loops of $\beta,\beta'$. In particular, these shared segments defined several possible crossings, which we called Type I, II, III, IV, which corresponded to whether each curve starts and ends on the `same side' or `opposite side' of their shared edges and whether the loops were parallel to each other on the section we were resolving. The different crossings changed the number of loops by either $0$ or $\pm 1$, and they differed in how many times (mod 2) the curves intersected each other with respect the background field along them. In all the cases, this allowed us to identify the change in $(-1)^\text{\# of loop} i^\text{winding}$ after resolving the intersection with the intersection number of the curves, as perturbed by the background vector field. 

There are two issues in trying to extend this logic to higher dimensions is that in higher dimensions, and consequently two cases we need to deal with to verify the quadratic refinement property. The first is related to the issue of why we need to introduce a trivalent resolution in the first place. If there are two curves that meet at a single point at the center of a $d$-simplex, then we need a trivalent resolution of that dual 1-skeleton to unambiguously say whether the curves split up and join each other or whether they stay in tact. So, these kinds of intersections are the first case. The fact that the quadratic refinement holds for this case of intersections is handled in Appendix \ref{windingsInHigherDims}.

The second issue deals with when the shared edges along the curves' intersection are the original edges of the dual 1-skeleton itself. In 2D, it made sense to distinguish the types of intersections based on which `side' of the curves' shared edges the curves start and end. But in higher dimensions, this notion doesn't make sense by itself. However, we can give this notion a meaning via \textit{thickening} the curves along the shared framing. This is because thickening along these $(d-2)$ shared vectors $\{v^\text{shared}_1,\dots,v^\text{shared}_{d-2}\}$ will locally give near the curve a codimension-1 set of points for which it's possible to ask which `side' of this set a curve is on. And, the `background' vector field $v^\text{bkgd}_{n-1}$ will act as the `perturbing' vector to separate the curve from its thickening. 

The fact that we chose the shared frame vectors and the perturbing background vector to be the same ones used to interpret the higher cup product will allow us to interpret quadratic refinement in the same way in higher dimensions as we did in lower dimensions. As depicted in Fig(\ref{higherDimIntersectionsAndWinding}), the `side' of the thickening that the curves enter and exit the shared region correspond exactly to how many times (mod 2) the curve intersects the thickened region after being perturbed by some other vector. We depicted the cases of Type I and II crossings in Fig(\ref{higherDimIntersectionsAndWinding}), but Types III and IV crossings can be drawn similarly. Note that based on the vectors we chose, these intersections points on each shared segment exactly give the contribution of the segment to $\int \beta \cup_{d-2} \beta'$! We can also consider projecting all the vectors the direction of thickening which would flatten the whole image to 2D. Then, resolving this intersection after flattening shows that the relationship after each resolution between the intersection number and the change of winding is follows exactly the same pattern as in 2D. Except the intersection $\int \beta \cup \beta'$ of 2D gets replaced with $\int \beta \cup_{d-2} \beta'$ in higher dimensions.

\begin{figure}[h!]
  \centering
  \begin{minipage}{0.81\textwidth}
    \centering
    \includegraphics[width=\linewidth]{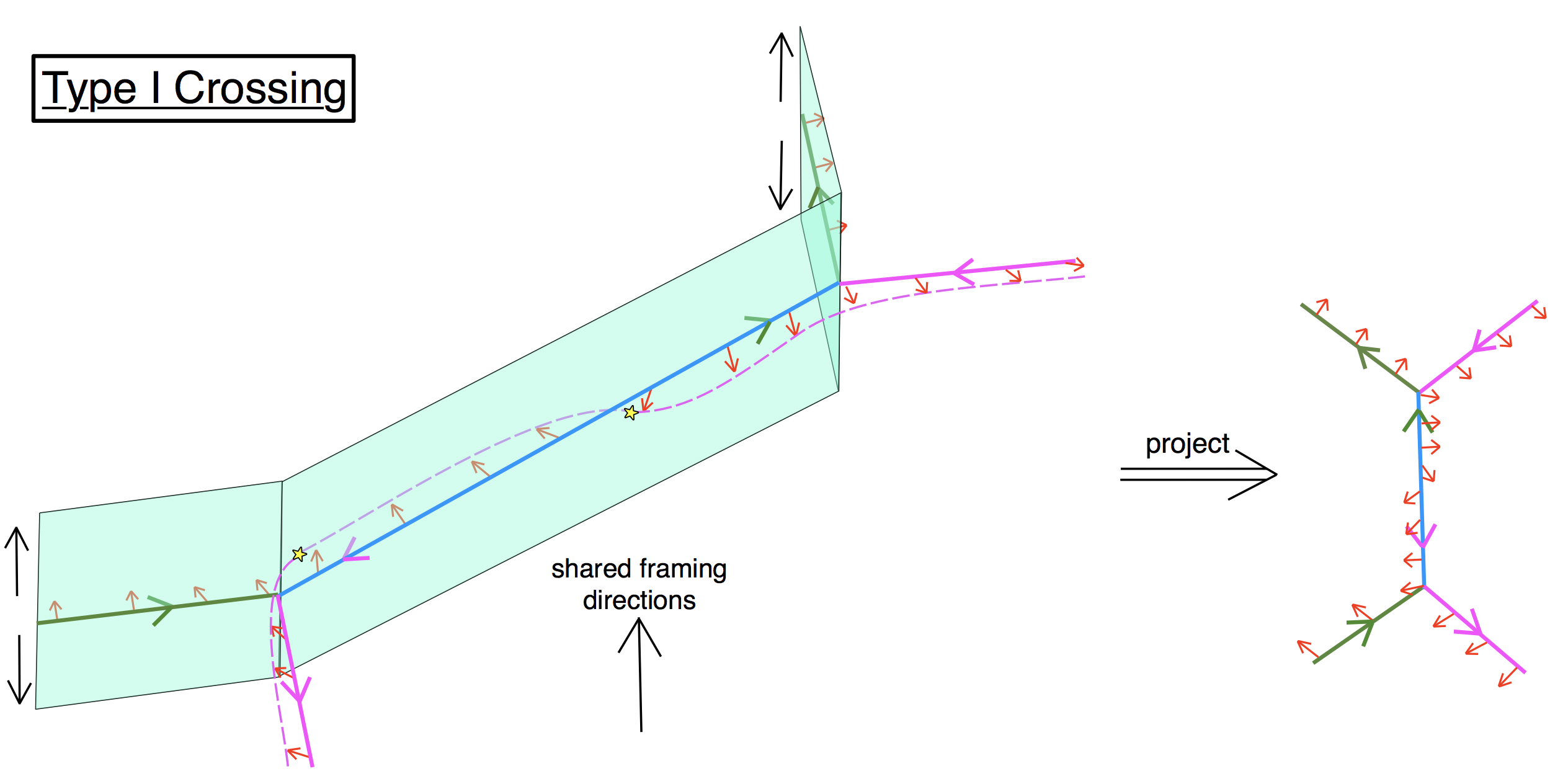}
  \end{minipage} \quad \quad \quad
  \begin{minipage}{0.81\textwidth}
    \centering
    \includegraphics[width=\linewidth]{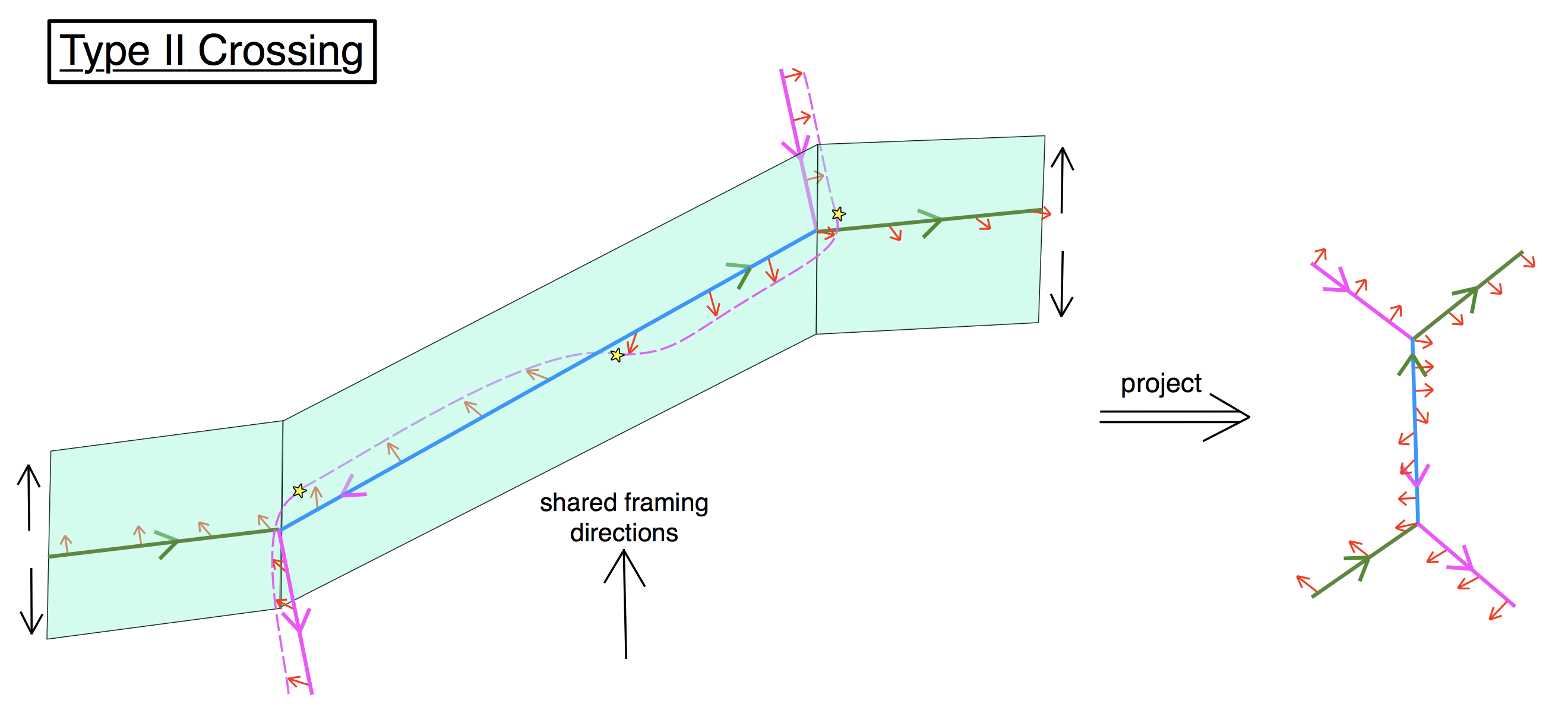}
  \end{minipage}
  \caption{In higher dimensions, we can consider thickening one curve in the direction of the shared framing and perturbing it in the direction of $(n-1)$th vector of the background frame (or equivalently perturbing the other curve in \textit{opposite} direction, as shown here). Given our definitions of the shared framing and the $(n-1)$th vector of the background frame, the number of intersections between this thickening and the other perturbed curve is a contribution of $\int \beta \cup_{d-2} \beta'$, where $\beta$ is the green curve, $\beta'$ is the pink curve, and their shared region is in blue as in Fig(\ref{intersectionsAndWindings}). By projecting out all components of the vectors in the direction of the shared framing, we can reduce the comparison of the intersection numbers and windings to the same analysis we gave for the 2D case.}
  \label{higherDimIntersectionsAndWinding}
\end{figure}

\subsection{Encoding a $Spin/Pin^-$ structure}
Now, we can describe how to use this construction to encode $Spin/Pin^-$ structures given the background framing and these triangulated manifolds, following discussion in \cite{mathOverflowThread} (see also \cite{BudneyCombinatorialSpin}). Recall that a $Pin^-$ structure can be thought of as a cochain $\eta \in C^1(M^\vee,\Z_2) = C_{n-1}(M,\Z_2)$ such that $\delta \eta = w_2 + w_1^2 \in C^2(M^\vee,\Z_2) = C_{n-2}(M,\Z_2)$. We want to see how this choice of $\eta$ can be thought of geometrically. Note that $\delta \eta = w_2$ means as chains that $\eta$ will be represented by some collection of $(d-1)$-simplices whose boundary is given by $w_2 + w_1^2$. Now, we should ask how this relates to our winding picture of a $Spin$ structure.

Remember that a $Spin$ or $Pin^-$ structure can be viewed as a trivialization of $TM$ or $TM \oplus \det(TM)$ on the 1-skeleton that extends to even-index singularities on the 2-skeleton. But, the framings we constructed to talk about $\sigma(M,\alpha)$ often extend to odd-index singularities on dual 2-cells where $w_2 + w_1^2$ doesn't vanish. So, to `fix' this, we'll choose some collection $\eta$ of edges on the dual 1-skeleton and `twist' the two unshared background vectors by $360^\circ$ going around the circle. We want to choose the collection of edges so that every dual 2-cell with $w_2+w_1^2 = 0$ will have an even number of edges in $\eta$ and those with $w_2+w_1^2$ will have an odd number of edges in $\eta$, and then twist the two unshared background vectors by $360^\circ$ along each edge, like in Fig(\ref{fig:spinStructTwist}). We can also think of $\eta$ as the collection of $(d-1)$-simplices dual to the edges in the dual 1-skeleton, and the boundary of this collection sum up to be a representative of $w_2 + w_1^2$. Note that this is only possible if $w_2 + w_1^2$ vanishes in cohomology.

This has the effect that for a curve traversing in a loop, its winding gets an additional full-twist (or equivalently multiplied by $-\mathbbm{1}$) per edge it contains in $\eta$ (or equivalently for every $(d-1)$-simplex in $\eta$ it crosses). For a cochain $\alpha$, we can write this phase as $(-1)^{\int \eta (\alpha)}$. This means that elementary plaquette loop $C$ will have $\sigma(\alpha_C)(-1)^{\int \eta (\alpha_C)} = 1$. These extra twists ensure that the twisted framing will extend to even-index singularities on each elementary plaquette, which is equivalent to saying that it defines a $Pin^-$ structure.

Two such $\eta,\eta'$ will give equivalent $Spin$/$Pin^-$ structures if $\eta + \eta'$ is represented by some homologically trivial sum of $(d-1)$-simplices. But, they'll instead give inequivalent $Spin$/$Pin^-$ structures if $\eta + \eta'$ is homologically nontrivial in $H^1(M,\Z_2)$. In other words, given $\lambda \in H^1(M,\Z_2)$ which is represnted by some closed codimension-1 collection of simplices, $\lambda.\eta$ is the $Spin$ structure we get by adding a $360^\circ$ twist relative to $\eta$ every time we cross the $\lambda$ surface.

Simlar reasoning can be used to combinatorially encode a $Pin^+$ structure, as in Appendix \ref{combDefinePinPlus}.

\begin{figure}[h!]
  \centering
  \includegraphics[width=0.5\linewidth]{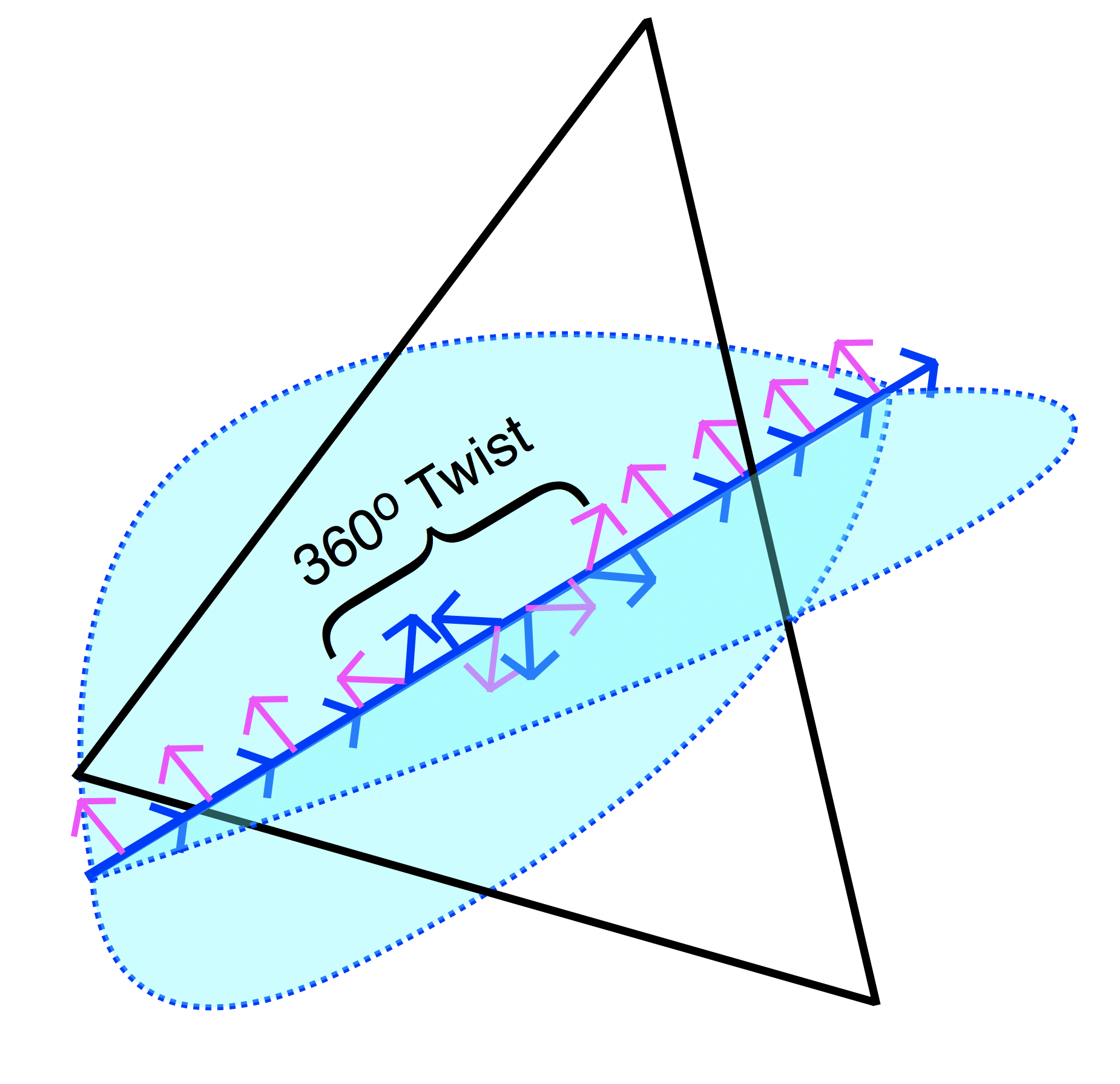}
  \caption{Twisting the two unshared vectors of the background framing with respect to each other along an edge by $360^\circ$ gives a $(d-1)$-simplex that's a part of the $Pin^-$ structure. Each such twist adds a minus sign to $(-1)^{\int \eta (\alpha_C)}$ for each crossing of $C$ with a $(d-1)$-simplex of $\eta$.}
  \label{fig:spinStructTwist}
\end{figure}

\section{Discussion}
We've constructed a set of vector field flows inside the standard simplex $\Delta^n$ that allows us to interpret the higher cup products as a generalized intersection between a cochain and a thickened, shifted version of the dual of another cochain. This allows us generalize the cup product, whose geometric interpretation was an intersection with respect to a shifting, but without any thickening. In particular, the Steenrod operations can then be interpreted as generalized `self-intersection' classes, with respect to how the submanifolds dual to the cochains intersect themselves upon being thickened and perturbed. This is a rephrasing to the formula of Thom \cite{ThomPlongee}, that when representing $\alpha$ by some submanifold $M'$ with normal bundle $\nu(M')$ and embedding map $f$, that

$$Sq^i(\alpha) = f_*(w_i(\nu(M')))$$

So, this interpretation of the $\cup_i$ products can be thought of as extending this understanding to the cochain level and as a binary operation. And, we found that the same vector fields were useful in defining combinatorial $Pin$ structures, defining the GWGK Grassmann Integral for $Spin$ and $Pin^-$ manifolds, and elucidating the geometric content of the `quadartic refinement' property of the GWGK Integral. 

We conclude with some questions and possible further directions about our constructions and how the applicability of these vector fields may be extended.

\begin{enumerate}
    \item Can similar methods define higher cup products on other cell decompositions of manifolds?
    \item Can we extend this understanding to more general `mod p' Steenrod operations?
    \item In solving the intersection equations, we notice that there are often `lower-dimensional' cells that arise but don't contribute to the $\cup_i$ formulas. Do these have any geometric or cohomological significance?
    \item Can similar methods using Vandermonde and Schur determinants be used to elucidate cochain-level proofs of the Cartan relations and Ádem relations (as done recently in \cite{CochainLevelCartan,CochainLevelAdem}), or perhaps the Wu formula? 
    \item Does the GWGK Integral $\sigma(\alpha)$ have a natural definition with respect to windings in the $Pin^+$ case?
    \item If we `smooth out' our vector fields near the boundaries of the simplices, can we use them to reproduce the cochain-level descriptions of $w_i$ from \cite{GoldsteinTurner}, similarly to the description of $w_2$? For example, our construction of $w_2$ is closely related to a formula for $w_2$ derived in \cite{YuAnBosonization}. It would be interesting if similar formulas existed for other $w_i$.
    \item Our definition of $\sigma(\alpha)$ via a loop decomposition depended on a trivalent resolution of the dual 1-skeleton. Our choice was \textit{ad hoc} and chosen to reproduce the quadratic refinement formula. Is there a more principled reason for this choice? 
\end{enumerate}

\section*{Acknowledgements}
\addcontentsline{toc}{section}{Acknowledgements}
We thank Maissam Barkeshli and Danny Bulmash for related discussions. And we acknowledge TA appointments and the Condensed Matter Theory Center at UMD for support.

\appendix
\section{Proof of solutions of Eq(\ref{mainProofFinalForm})} \label{proofOfVandermonde}
Now, let's show that the expressions (\ref{mainSolZ},\ref{mainSolA}) solve Eq(\ref{mainProofFinalForm}). It'll be possible to verify once we multiply each side of the equation by the denominators of the solutions,

\begin{equation}\label{mainProofDenominatorDeterminant}
\det
\begin{pmatrix}
1      & b_{1 \lambda_0} & b_{2 \lambda_0} & \cdots & b_{m \lambda_0} \\
1      & b_{1 \lambda_1} & b_{2 \lambda_1} & \cdots & b_{m \lambda_1} \\
\vdots &    \vdots       &   \vdots          &        &   \vdots          \\
1      & b_{1 \lambda_m} & b_{2 \lambda_1} & \cdots & b_{m \lambda_m} \\
\end{pmatrix}.
\end{equation}

Let's first examine the RHS of Eq(\ref{mainProofFinalForm}) after being multiplied by this determinant (\ref{mainProofDenominatorDeterminant}), which we'll call $RH$. For each $k$, it will read:

\begin{equation}
\begin{split}
RH = n(n+1)\Bigg(&
\begin{vmatrix}
1      & b_{1 \lambda_0} & \cdots & b_{m+1,\lambda_0}     \\
\vdots &  \vdots         &        & \vdots                \\
1      & b_{1 \lambda_m} & \cdots & b_{m+1,\lambda_m}     \\
1      & b_{1 k}         & \cdots & b_{m+1,k} \\
\end{vmatrix} \\
&- \sum_{\ell = 1}^m (-1)^{m-\ell+1} b_{\ell k}
\begin{vmatrix}
1      & b_{\Tilde{\imath}_1 \lambda_0} & \cdots & b_{\Tilde{\imath}_{m+1},\lambda_0}     \\
\vdots & \vdots                         &        & \vdots                \\
1      & b_{\Tilde{\imath}_1 \lambda_m} & \cdots & b_{\Tilde{\imath}_{m+1},\lambda_m}     \\
\end{vmatrix}
- b_{m+1,k}
\begin{vmatrix}
1      & b_{1 \lambda_0} & \cdots & b_{m \lambda_0}     \\
\vdots &  \vdots         &        & \vdots                \\
1      & b_{1 \lambda_m} & \cdots & b_{m \lambda_m}     \\
\end{vmatrix} \Bigg)
\end{split}
\end{equation}

where for each $\ell$,  $\{\Tilde{\imath}_1,\dots,\Tilde{\imath}_m\} = \{1,\dots,m+1\} \textbackslash \{ \ell \}$. Note that the first term is from $Z_k \delta_{k \in \{\hat{\lambda}\}}$. Note this term is automatically zero if $k \notin \{\hat{\lambda}\}$, i.e. if $k \in \{\lambda_0,\dots\,\lambda_n\}$, then the matrix would have two rows that are equal, so the $\delta_{k \in \{\hat{\lambda}\}}$ part is satisfied. And, note that the second row of the above equation is the cofactor expansion of a matrix, so that we get

\begin{equation}
\begin{split}
RH &= n(n+1)\Bigg(
\begin{vmatrix}
1      & b_{1 \lambda_0} & \cdots & b_{m+1,\lambda_0}     \\
\vdots &  \vdots         &        & \vdots                \\
1      & b_{1 \lambda_m} & \cdots & b_{m+1,\lambda_m}     \\
1      & b_{1 k}         & \cdots & b_{m+1,k} \\
\end{vmatrix}
+
\begin{vmatrix}
1      &  b_{1 \lambda_0} & \cdots & b_{m+1,\lambda_0}     \\
\vdots &   \vdots         &        & \vdots                \\
1      &  b_{1 \lambda_m} & \cdots & b_{m+1,\lambda_m}     \\
0      & -b_{1 k}         & \cdots & -b_{m+1,k} \\
\end{vmatrix}
\Bigg) \\
& = n(n+1) 
\begin{vmatrix}
1      & b_{1 \lambda_0} & \cdots & b_{m+1,\lambda_0}     \\
\vdots &  \vdots         &        & \vdots                \\
1      & b_{1 \lambda_m} & \cdots & b_{m+1,\lambda_m}     \\
1      & 0               & \cdots & 0 \\
\end{vmatrix} \\
&= n(n+1)
\begin{vmatrix}
b_{1 \lambda_0} & \cdots & b_{m+1,\lambda_0}     \\
\vdots          &        & \vdots                \\
b_{1 \lambda_m} & \cdots & b_{m+1,\lambda_m}     \\
\end{vmatrix}
\end{split}
\end{equation}

Now we should verify that the LHS of Eq(\ref{mainProofFinalForm}) is indeed the same quantity. Let's give the name $LH$ to the determinant (\ref{mainProofDenominatorDeterminant}) times the LHS of Eq(\ref{mainProofFinalForm}). We'll have

\begin{equation}
\begin{split}
LH &= n \Bigg( \sum_{\hat{\lambda}}
\begin{vmatrix}
1      & b_{1 \lambda_0}     & \cdots & b_{m+1,\lambda_0}     \\
\vdots &  \vdots             &        & \vdots                \\
1      & b_{1 \lambda_m}     & \cdots & b_{m+1,\lambda_m}     \\
1      & b_{1 \hat{\lambda}} & \cdots & b_{m+1,\hat{\lambda}} \\
\end{vmatrix} \\
&- \sum_{\ell = 1}^m (-1)^{m-\ell+1} B_\ell
\begin{vmatrix}
1      & b_{\Tilde{\imath}_1 \lambda_0} & \cdots & b_{\Tilde{\imath}_{m+1},\lambda_0}     \\
\vdots & \vdots                         &        & \vdots                \\
1      & b_{\Tilde{\imath}_1 \lambda_m} & \cdots & b_{\Tilde{\imath}_{m+1},\lambda_m}     \\
\end{vmatrix}
- B_{m+1}
\begin{vmatrix}
1      & b_{1 \lambda_0} & \cdots & b_{m \lambda_0}     \\
\vdots &  \vdots         &        & \vdots                \\
1      & b_{1 \lambda_m} & \cdots & b_{m \lambda_m}     \\
\end{vmatrix} \Bigg)
\end{split}
\end{equation}

We can do two things. First, we can again see that the second row of the above expression is the cofactor expansion of a matrix. And, we can also replace the sum $\sum_{\hat{\lambda}}$ with a sum $\sum_{\ell=0}^n$. This is because if $\ell \in \{\lambda_0,\dots,\lambda_m\}$, then the corresponding determinant will have two rows that are equal, so it will be zero. So we can freely add them without changing the sum. So, we'll have

\begin{equation}
\begin{split}
LH &= n \Bigg( \sum_{\ell=0}^n 
\begin{vmatrix}
1      & b_{1 \lambda_0} & \cdots & b_{m+1,\lambda_0}     \\
\vdots &  \vdots         &        & \vdots                \\
1      & b_{1 \lambda_m} & \cdots & b_{m+1,\lambda_m}     \\
1      & b_{1 \ell}      & \cdots & b_{m+1,\ell} \\
\end{vmatrix}
+
\begin{vmatrix}
1      &  b_{1 \lambda_0} & \cdots & b_{m+1,\lambda_0}     \\
\vdots &   \vdots         &        & \vdots                \\
1      &  b_{1 \lambda_m} & \cdots & b_{m+1,\lambda_m}     \\
0      & -B_{1}         & \cdots & -B_{m+1} \\
\end{vmatrix} \Bigg) \\
&= 
n \Bigg( 
\begin{vmatrix}
1      & b_{1 \lambda_0} & \cdots & b_{m+1,\lambda_0}     \\
\vdots &  \vdots         &        & \vdots                \\
1      & b_{1 \lambda_m} & \cdots & b_{m+1,\lambda_m}     \\
n+1    & B_1             & \cdots & B_{m+1} \\
\end{vmatrix} 
+
\begin{vmatrix}
1      &  b_{1 \lambda_0} & \cdots & b_{m+1,\lambda_0}     \\
\vdots &   \vdots         &        & \vdots                \\
1      &  b_{1 \lambda_m} & \cdots & b_{m+1,\lambda_m}     \\
0      & -B_{1}         & \cdots & -B_{m+1} \\
\end{vmatrix}  \Bigg) \\
&= n 
\begin{vmatrix}
1      & b_{1 \lambda_0} & \cdots & b_{m+1,\lambda_0}     \\
\vdots &  \vdots         &        & \vdots                \\
1      & b_{1 \lambda_m} & \cdots & b_{m+1,\lambda_m}     \\
n+1    & 0               & \cdots & 0 \\
\end{vmatrix} \\
&= n(n+1)
\begin{vmatrix}
b_{1 \lambda_0} & \cdots & b_{m+1,\lambda_0}     \\
\vdots          &        & \vdots                \\
b_{1 \lambda_m} & \cdots & b_{m+1,\lambda_m}     \\
\end{vmatrix} \\
&= RH
\end{split}
\end{equation}

So, we've just shown that LHS and RHS of Eq(\ref{mainProofFinalForm}) are equal substituting the expressions (\ref{mainSolZ},\ref{mainSolA}). And given generic $\Vec{b}$ such that any $n$ of the $\{\Vec{b}_1,\dots,\Vec{b}_{m+1}, c-f_0,\dots,c-f_{n}\}$ are linearly independent, this is the only solution.

\section{The $Spin$ and $Pin$ groups} \label{spinPinAppendix}
Let's review some basic properties and definitions of the $Spin$ and $Pin^\pm$ groups. 

\subsection{$Spin(n)$}
Recall that $Spin(n)$ is the connected double-cover of the group $SO(n)$ given by an exact sequence $\Z_2 \xhookrightarrow{} Spin(n) \twoheadrightarrow SO(n)$, and that $Spin(n)$ and $SO(n)$ share the same Lie Algebra.

The Lie Algebra of $SO(n)$ is generated by the matrices $\{ M^{i j} \}_{1 \le i < j \le n}$ with matrix elements

\begin{equation}
(M^{i j})_{\mu \nu} := \delta^i_\mu \delta^j_\nu - \delta^i_\nu \delta^j_\mu
\end{equation} 

satisfying the commutation relations

\begin{equation}
[M^{i j},M^{\alpha \beta}] = \delta^{j \alpha}M^{i \beta} - \delta^{i \alpha}M^{j \beta} - \delta^{j \beta}M^{i \alpha} + \delta^{i \beta}M^{j \alpha}
\end{equation}

Before constructing $Spin(n)$, we need a fundamental representation of the Clifford algebra $\gamma^i \gamma^j + \gamma^i \gamma^j = 2 \delta_{i j} \mathbbm{1}$. This representation will mean that all the $\gamma^i$ are $2^{\floor{d/2}} \times 2^{\floor{d/2}}$ matrices and have equal numbers of $\pm 1$ eigenvalues. Given this, one can show that the matrices $-\frac{1}{4}[\gamma^i,\gamma^j]$ define the same Lie Algebra commutation relations as the $M^{i j}$. From here, the group $Spin(n)$ can be defined as the set:

\begin{equation}
Spin(n) = \Big\{\exp(-\frac{1}{4} \sum_{1 \le i < j \le n} a_{ij} [\gamma^i,\gamma^j]) \Big| a_{ij} \in \R \Big\}
\end{equation}

This will be a double-cover of $SO(n)$, since we can check that $\exp(2 \pi M^{i j}) = \mathbbm{1}$, whereas  $\exp(-\frac{2\pi}{4} [\gamma^i,\gamma^j]) = -\mathbbm{1}$.

\subsection{$Pin^{\pm}(n)$}
Now let's describe the construction of the $Pin^\pm$ groups. We'll follow the presentation of \cite{TachikawaYonekura}, who note that $Pin^\pm$ fit into a natural length-4 periodic structure of $\Z_2$ extensions of $Spin(n)$. The main idea is that $Pin^{\pm}(n)$ are double-covers of $Spin(n)$ that include $Spin(n)$ as a subgroup and include a disconnected part that's topologically the same as $Spin(n)$ and can be reached from $Spin(n)$ by some reflection $R$. $Pin^-(n)$ satisfies $R^2 = -\mathbbm{1}$ and $Pin^+(n)$ has $R^2 = +\mathbbm{1}$.

\subsubsection{$Pin^-(n)$}
Let's describe first how to construct $Pin^-(n)$. We can view $Pin^-(n)$ as the subgroup of $Spin(n+1)$ such that projecting to $SO(n+1)$ acts on $\R^{n+1} = \R^n \oplus \R^1$ as $\pm 1$ on the last $\R^1$ factor. So, given the sequence $1 \to \Z_2 \xrightarrow{i} Spin(n+1) \xrightarrow{p} SO(n+1) \to 1$, $Pin^-(n) \subset Spin(n+1)$ can be written as a preimage of a subgroup of $SO(n)$ as:

\begin{equation}
Pin^-(n) = p^{-1} 
\Big(\big\{ \left(\begin{array}{@{}c|c}
W & 0 \\ \hline
0 & \pm 1
\end{array}\right) \big| W \in O(n) \big\}\Big) 
\end{equation}

Note that in this preimage, the matrix acting on the first $\R^n$ factor will have determinant $-1$ if the matrix acts as $-1$ on the $\R^1$ factor, which is why we say $Pin^-$ acts as a reflection. So, our goal will be to express $Pin^-$ in terms of the $\gamma_1,\dots,\gamma_{n+1}$ used to generate $Spin(n+1)$. Note that the matrix $\exp(\frac{\pi}{2}\gamma_i \gamma_j) = \gamma_i \gamma_j $ projected onto $SO(n)$ acts as $180^\circ$ rotation in the $(i,j)$ plane, which sends $x_k \to x_k$ if $k \neq i,j$ and $x_k \to -x_k$ if $k = i,j$. This means that the matrix $\gamma_n \gamma_{n+1}$ will act on the last $\R^1$ factor as $-1$. Also note that $(\gamma_n \gamma_{n+1})^2 = -1$, which means we can identify a reflection $R$ with $\gamma_n \gamma_{n+1}$ with $R^2 = -1$. So, we can write $Pin^-(n)$ in terms of $Spin(n)$ and allowing multiplication by this additional matrix.

\begin{equation}
Pin^-(n) = \Big\{(\gamma_n \gamma_{n+1})^{b} \exp(-\frac{1}{4} \sum_{1 \le i < j \le n} a_{ij} [\gamma^i,\gamma^j]) \Big| a_{ij} \in \R, b \in \{0,1\} \Big\}
\end{equation}

This view of $Pin^-(n)$ is intuition for why a $Pin^-$ structure on a manifold $M$ can be viewed as a $Spin$ structure on $TM \oplus \det(TM)$. 

\subsubsection{$Pin^+(n)$}
We can similarly express the $Pin^+$ groups. But, $Pin^+(n)$ will be viewed instead as a suitable subgroup of $Spin(n+3)$. We'll have $Pin^+(n)$ is the the preimage that projects to $SO(n+3)$ acting on $\R^{n+3} = \R^n \oplus \R^3$ acting as $\pm \mathbbm{1}_{3 \times 3}$ on the last $\R^3$ factor. So given the sequence $1 \to \Z_2 \xrightarrow{i} Spin(n+3) \xrightarrow{p} SO(n+3) \to 1$, we'll have:

\begin{equation}
Pin^+(n) = p^{-1} 
\Big(\big\{ \left(\begin{array}{@{}c|c}
W & 0 \\ \hline
0 & \pm \mathbbm{1}_{3 \times 3}
\end{array}\right) \big| W \in O(n) \big\}\Big) 
\end{equation}

And, we can similarly express a reflection matrix $R$ as doing a $180^\circ$ rotation in each of the $(n,n+1),(n,n+2),(n,n+3)$ planes, which would give $R = (\gamma_n \gamma_{n+1})(\gamma_n \gamma_{n+2})(\gamma_n \gamma_{n+3}) = -\gamma_n \gamma_{n+1} \gamma_{n+2} \gamma_{n+3}$. And, we can check that this $R$ satisfies $R^2 = +\mathbbm{1}$. So we can write $Pin^+(n)$ as

\begin{equation}
Pin^+(n) = \Big\{(\gamma_n \gamma_{n+1} \gamma_{n+2} \gamma_{n+3})^{b} \exp(-\frac{1}{4} \sum_{1 \le i < j \le n} a_{ij} [\gamma^i,\gamma^j]) \Big| a_{ij} \in \R, b \in \{0,1\} \Big\}
\end{equation}

Similarly, this view of $Pin^+(n)$ is intuition for why a $Pin^+$ structure on a manifold $M$ can be viewed as a $Spin$ structure on $TM \oplus 3 \det(TM)$. 

\section{Framings and Windings on Higher-Dimensional Simplices} \label{windingsInHigherDims}
Now, we should verify and clarify some aspects of our constructions from the Section \ref{theDifferentFramings}. We'll start by reiterating the constructions we gave there. Then, we'll describe the computation of the windings along each segment. And finally, we'll describe how the choice of trivalent resolution as shown in Fig(\ref{fig:trivalentResAndWinding}) allows us to verify that quadratic refinement holds. 

\subsection{The Framings}
First, let's describe how the framings look within each $d$-simplex. Then we'll describe how to look at the vectors as we go across a $(d-1)$-simplex into an adjacent $d$-simplex.

\subsubsection{Within a $d$-simplex}
Let's say we're given some path within a $d$-simplex which enters at the dual edge $\ihat$ and exits at $\jhat$. Then, we'll define four points, $(1),(2),(3),(4)$, along the path as depicted in Fig(\ref{fig:iHatTojHatPath}). $(1)$ is some point near the end of the edge $\ihat_{ed}$ near the $(d-1)$-simplex $\ihat_\Delta$. $(2)$ is some point on $\ihat$ near the center of the $d$-simplex. $(3)$ is some point on $\jhat$ that's near the center. And $(4)$ is on the end $\jhat_{ed}$ near $\jhat_\Delta$.

\begin{figure}[h!]
  \centering
  \includegraphics[width=0.6\linewidth]{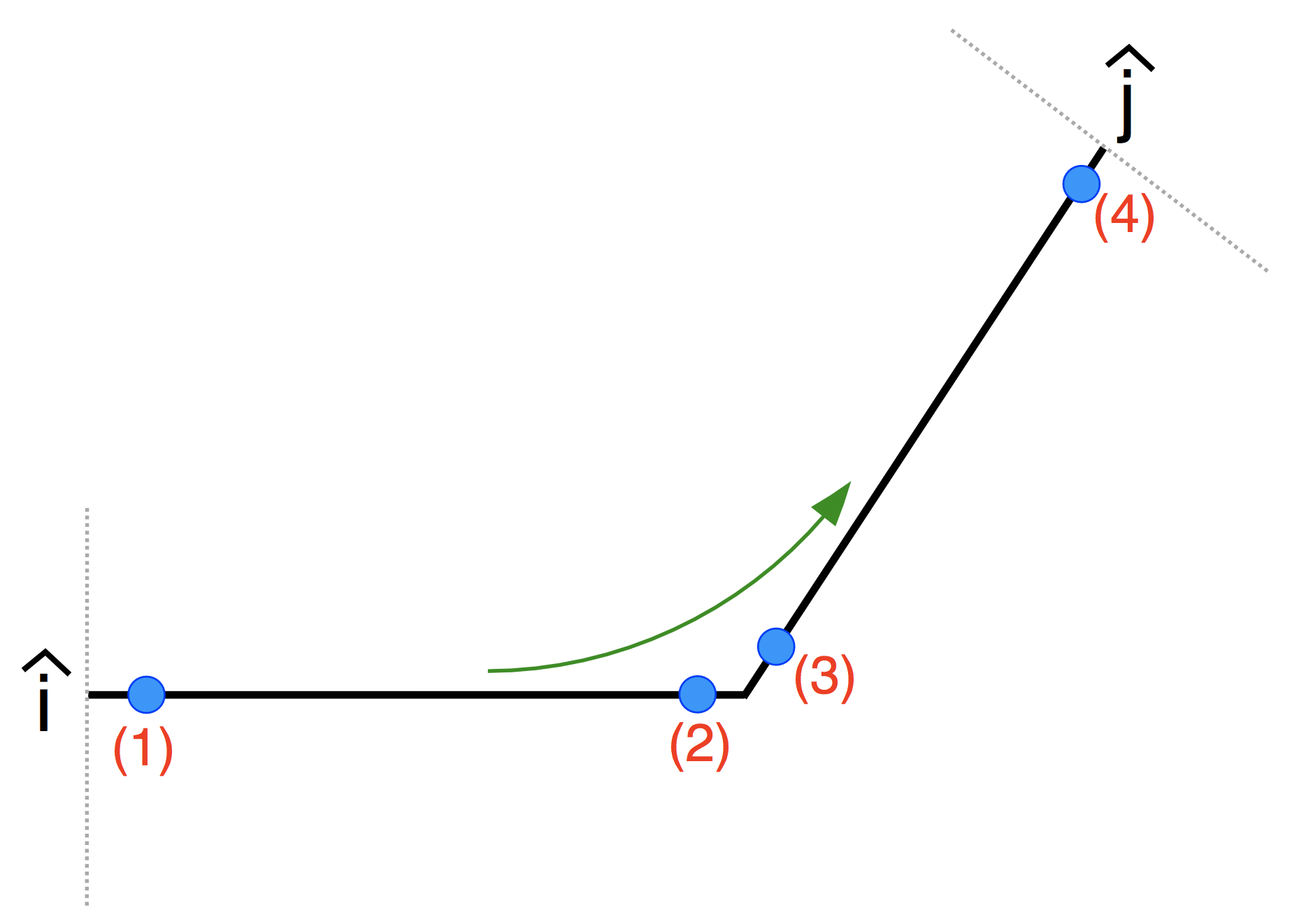}
  \caption{Different points along the path $\ihat \to \jhat$ within a $d$-simplex. The path traverses as $\ihat_\Delta \to (1) \to (2) \to (3) \to (4) \to \jhat_\Delta$. Some of the vectors in each of the relevant framings in the text change as we traverse some of these parts of the path.}
  \label{fig:iHatTojHatPath}
\end{figure}

We want to specify how each of the different framings look like along each part of the path $\ihat \to \jhat$. The general features are as follows. The parts $(1) \to \ihat_\Delta$ and $(4) \to \jhat_\Delta$ project out the $i, j$ respectively components of all the vectors, \textit{except} the one that runs along the 1-skeleton. Along the rest of the path, the shared framing's vectors will remain constant and the only vectors that can change are some of the extra two vectors of the background and tangent framings. The part $(2) \to (3)$ is meant to change the direction of the vector going along the 1-skeleton for both framings and the other vector for the tangent framing.

Let's start with describing this procedure for the background framing. Let's call the matrix corresponding to the framing $F^\text{bkgd}$. Let's illustrate the matrices here for a `$+$' simplex. We'll have along the path $\ihat \to \jhat$ that:

\begin{equation}
F^\text{bkgd} =
  \left(\begin{array}{@{}ccccccc}
    1      & \cdots & 1                 & 1                                  & 1                     & \cdots & 1                     \\
    1      & \cdots & \frac{1}{i}       & t \frac{1}{i+1}                    & \frac{1}{i+2}         & \cdots & \frac{1}{d+1}         \\
    \vdots &        & \vdots            &                                    & \vdots                &        & \vdots                \\
    1      & \cdots & \frac{1}{i^{d-2}} & t \frac{1}{(i+1)^{d-2}}            & \frac{1}{(i+2)^{d-2}} & \cdots & \frac{1}{(d+1)^{d-2}} \\ \hline
    1      & \cdots & \frac{1}{i^{d-1}} & t \frac{1}{(i+1)^{d-1}}            & \frac{1}{(i+2)^{d-1}} & \cdots & \frac{1}{(d+1)^{d-1}} \\ 
    0      & \cdots & 0                 & \frac{(-1)^{i+d+{d \choose 2}}}{n} & 0                     & \cdots & 0                     \\
  \end{array}\right) \text{ for } \ihat_\Delta \to (1)
\end{equation}

\begin{equation}
F^\text{bkgd} =
  \left(\begin{array}{@{}ccccccc}
    1      & \cdots & 1                 & 1                                  & 1                     & \cdots & 1                     \\
    1      & \cdots & \frac{1}{i}       & \frac{1}{i+1}                      & \frac{1}{i+2}         & \cdots & \frac{1}{d+1}         \\
    \vdots &        & \vdots            &                                    & \vdots                &        & \vdots                \\
    1      & \cdots & \frac{1}{i^{d-2}} & \frac{1}{(i+1)^{d-2}}              & \frac{1}{(i+2)^{d-2}} & \cdots & \frac{1}{(d+1)^{d-2}} \\ \hline
    1      & \cdots & \frac{1}{i^{d-1}} & \frac{1}{(i+1)^{d-1}}              & \frac{1}{(i+2)^{d-1}} & \cdots & \frac{1}{(d+1)^{d-1}} \\ 
    0      & \cdots & 0                 & \frac{(-1)^{i+d+{d \choose 2}}}{n} & 0                     & \cdots & 0                     \\
  \end{array}\right) \text{ for } (1) \to (2)
\end{equation}

\begin{equation}
F^\text{bkgd} =
  \left(\begin{array}{@{}ccccccccccc}
    1      & \cdots & 1                 & 1                                        & 1                     & \cdots & 1                 & 1                                    & 1                     & \cdots & 1                 \\
    
    1      & \cdots & \frac{1}{i}       & \frac{1}{i+1}                            & \frac{1}{i+2}         & \cdots & \frac{1}{j}       & \frac{1}{j+1}                        & \frac{1}{j+2}         & \cdots & \frac{1}{d+1}     \\
    
    \vdots &        & \vdots            &                                          & \vdots                &        & \vdots            &                                      & \vdots                &        & \vdots            \\
    
    1      & \cdots & \frac{1}{i^{d-2}} & \frac{1}{(i+1)^{d-2}}                    & \frac{1}{(i+2)^{d-2}} & \cdots & \frac{1}{j^{d-2}} & \frac{1}{(j+1)^{d-2}}                & \frac{1}{(j+2)^{d-2}} & \cdots & \frac{1}{(d+1)^{d-2}}  \\ \hline
    
    1      & \cdots & \frac{1}{i^{d-1}} & \frac{1}{(i+1)^{d-1}}                    & \frac{1}{(i+2)^{d-1}} & \cdots & \frac{1}{j^{d-1}} & \frac{1}{(j+1)^{d-1}}                & \frac{1}{(j+2)^{d-1}} & \cdots & \frac{1}{(d+1)^{d-1}}  \\
    
    0      & \cdots & 0                 & (1-t) \frac{(-1)^{i+d+{d \choose 2}}}{n} & 0                     & \cdots & 0                 & t \frac{(-1)^{j+d+{d \choose 2}}}{n} & 0                     & \cdots & 0                      \\
    
  \end{array}\right)
  \text{ for } (2) \to (3)
\end{equation}

\begin{equation}
F^\text{bkgd} =
  \left(\begin{array}{@{}ccccccc}
    1      & \cdots & 1                 & 1                                  & 1                     & \cdots & 1                     \\
    1      & \cdots & \frac{1}{j}       & \frac{1}{j+1}                      & \frac{1}{j+2}         & \cdots & \frac{1}{d+1}         \\
    \vdots &        & \vdots            &                                    & \vdots                &        & \vdots                \\
    1      & \cdots & \frac{1}{j^{d-2}} & \frac{1}{(i+1)^{d-2}}              & \frac{1}{(i+2)^{d-2}} & \cdots & \frac{1}{(d+1)^{d-2}} \\ \hline
    1      & \cdots & \frac{1}{j^{d-1}} & \frac{1}{(j+1)^{d-1}}              & \frac{1}{(j+2)^{d-1}} & \cdots & \frac{1}{(d+1)^{d-1}} \\ 
    0      & \cdots & 0                 & \frac{(-1)^{j+d+{d \choose 2}}}{n} & 0                     & \cdots & 0                     \\
  \end{array}\right) \text{ for } (3) \to (4)
\end{equation}

\begin{equation}
F^\text{bkgd} =
  \left(\begin{array}{@{}ccccccc}
    1      & \cdots & 1                 & 1                                  & 1                     & \cdots & 1                     \\
    1      & \cdots & \frac{1}{j}       & (1-t) \frac{1}{j+1}                & \frac{1}{j+2}         & \cdots & \frac{1}{d+1}         \\
    \vdots &        & \vdots            &                                    & \vdots                &        & \vdots                \\
    1      & \cdots & \frac{1}{j^{d-2}} & (1-t) \frac{1}{(j+1)^{d-2}}        & \frac{1}{(j+2)^{d-2}} & \cdots & \frac{1}{(d+1)^{d-2}} \\ \hline
    1      & \cdots & \frac{1}{j^{d-1}} & (1-t) \frac{1}{(j+1)^{d-1}}        & \frac{1}{(j+2)^{d-1}} & \cdots & \frac{1}{(d+1)^{d-1}} \\ 
    0      & \cdots & 0                 & \frac{(-1)^{j+d+{d \choose 2}}}{n} & 0                     & \cdots & 0                     \\
  \end{array}\right) \text{ for } (4) \to \jhat_\Delta
\end{equation}

For a `$-$' simplex, the first $(d-1)$ the rows stay the same and the last row is multplied by $-1$. And, note that while we rendered these matrices looking like $i < j$, we could also have $i > j$. The factors $(-1)^{i +d + {d \choose 2}}$,$(-1)^{j+d+{d \choose 2}}$ factors are there in the last row to make sure that the determinant is positive. 

In all the matrices above, the top $(d-2)$ rows above the bar represent the shared framing's vectors, and the bottom 2 rows represent the vectors unique to the background framing. And technically, the vectors listed above are not parallel to $\Delta^d \in \R^{d+1}$. We choose to write the vectors without projecting out the $(1,\dots,1)$ component for convenience, and because including or projecting away that part won't affect the windings. And in some of the matrices, there was a parameter `$t$' ranging in $[0,1]$ we included, which represents some parameterization of the segment. For example, in the segments $\ihat \to (1)$ and $(4) \to \jhat$, $t$ serves to project out the $i$ or $j$ component as we move closer to $\ihat_\Delta, \jhat_\Delta$. A useful fact is that all square submatrices of the same size of a Vandermonde matrix have the same sign determinant if the entries are in ascending or descending order and are the same sign. \footnote{This is because the a submatrix of a Vandermonde matrix has a determinant which is a Schur polynomial of the corresponding columns' elements times the Vandermonde determinant of the entries corresponding to the columns of the submatrix. A Schur polynomial of any positive arguments is positive. And if the entries of the matrix are in ascending or descending order, then each submatrix has a corresponding Vandermonde determinant that's the same sign.} From this, it's a simple exercise to show that this family of matrices always has a nonzero determinant along the whole path. 

Now, we'll write the analogous matrices for the tangent framing, which we'll denote $F^\text{tang}$. Here, we'll give the matrices for a `$+$' simplex and where the orientation vectors are in agreement between the background and tangent frames. The cases of a `$-$' simplex or when the orientation vectors disagree can be obtained by multiplying the $(d-1)$-th vector by a $-1$ factor for \textit{each} of those changes. 

\begin{equation}
F^\text{tang} = 
  \left(\begin{array}{@{}ccccccccccc}
    1      & \cdots & 1                 & 1                       & 1                     & \cdots & 1                 & 1                                                       & 1                     & \cdots & 1                      \\
    
    1      & \cdots & \frac{1}{i}       & t \frac{1}{i+1}         & \frac{1}{i+2}         & \cdots & \frac{1}{j}       & \frac{1}{j+1}                                           & \frac{1}{j+2}         & \cdots & \frac{1}{d+1}          \\
    
    \vdots &        & \vdots            &                         & \vdots                &        & \vdots            &                                                         & \vdots                &        & \vdots                 \\
    
    1      & \cdots & \frac{1}{i^{d-2}} & t \frac{1}{(i+1)^{d-2}} & \frac{1}{(i+2)^{d-2}} & \cdots & \frac{1}{j^{d-2}} & \frac{1}{(j+1)^{d-2}}                                   & \frac{1}{(j+2)^{d-2}} & \cdots & \frac{1}{(d+1)^{d-2}}  \\ \hline
    
    0      & \cdots & 0                 & 0                       & 0                     & \cdots & 0                 & \frac{(-1)^{i+j + {d-1 \choose 2} + \delta_{i > j}}}{n} & 0                     & \cdots & 0                       \\
    
    0      & \cdots & 0                 & \frac{1}{n}             & 0                     & \cdots & 0                 & 0                                                       & 0                     & \cdots & 0                       \\
    
  \end{array}\right) \text{ for } \ihat_\Delta \to (1)
\end{equation} 

\begin{equation}
F^\text{tang} = 
  \left(\begin{array}{@{}ccccccccccc}
    1      & \cdots & 1                 & 1                      & 1                     & \cdots & 1                 & 1                                                       & 1                     & \cdots & 1                      \\
    
    1      & \cdots & \frac{1}{i}       & \frac{1}{i+1}          & \frac{1}{i+2}         & \cdots & \frac{1}{j}       & \frac{1}{j+1}                                           & \frac{1}{j+2}         & \cdots & \frac{1}{d+1}          \\
    
    \vdots &        & \vdots            &                        & \vdots                &        & \vdots            &                                                         & \vdots                &        & \vdots                 \\
    
    1      & \cdots & \frac{1}{i^{d-2}} & \frac{1}{(i+1)^{d-2}}  & \frac{1}{(i+2)^{d-2}} & \cdots & \frac{1}{j^{d-2}} & \frac{1}{(j+1)^{d-2}}                                   & \frac{1}{(j+2)^{d-2}} & \cdots & \frac{1}{(d+1)^{d-2}}  \\ \hline
    
    0      & \cdots & 0                 & 0                      & 0                     & \cdots & 0                 & \frac{(-1)^{i+j + {d-1 \choose 2} + \delta_{i > j}}}{n} & 0                     & \cdots & 0                       \\
    
    0      & \cdots & 0                 & \frac{1}{n}            & 0                     & \cdots & 0                 & 0                                                       & 0                     & \cdots & 0                       \\
    
  \end{array}\right) \text{ for } (1) \to (2)
\end{equation} 

\begin{equation}
F^\text{tang} = 
  \left(\begin{array}{@{}ccccccccccc}
    1      & \cdots & 1                 & 1                                                         & 1                     & \cdots & 1                 & 1                                                             & 1                     & \cdots & 1                      \\
    
    1      & \cdots & \frac{1}{i}       & \frac{1}{i+1}                                             & \frac{1}{i+2}         & \cdots & \frac{1}{j}       & \frac{1}{j+1}                                                 & \frac{1}{j+2}         & \cdots & \frac{1}{d+1}          \\
    
    \vdots &        & \vdots            &                                                           & \vdots                &        & \vdots            &                                                               & \vdots                 &        & \vdots                 \\
    
    1      & \cdots & \frac{1}{i^{d-2}} & \frac{1}{(i+1)^{d-2}}                                     & \frac{1}{(i+2)^{d-2}} & \cdots & \frac{1}{j^{d-2}} & \frac{1}{(j+1)^{d-2}}                                         & \frac{1}{(j+2)^{d-2}} & \cdots & \frac{1}{(d+1)^{d-2}}  \\ \hline
    
    0      & \cdots & 0                 & t \frac{(-1)^{i+j + {d-1 \choose 2} + \delta_{i > j}}}{n} & 0                     & \cdots & 0                 & (1-t) \frac{(-1)^{i+j + {d-1 \choose 2} + \delta_{i > j}}}{n} & 0                     & \cdots & 0                       \\
    
    0      & \cdots & 0                 & (1-t) \frac{1}{n}                                         & 0                     & \cdots & 0                 & -t \frac{1}{n}                                                & 0                     & \cdots & 0                       \\
    
  \end{array}\right)
\end{equation} 
\quad\quad\quad\quad\quad for $(2) \to (3)$
 
\begin{equation}
F^\text{tang} = 
  \left(\begin{array}{@{}ccccccccccc}
    1      & \cdots & 1                 & 1                                                       & 1                     & \cdots & 1                 & 1                     & 1                     & \cdots & 1                      \\
    
    1      & \cdots & \frac{1}{i}       & \frac{1}{i+1}                                           & \frac{1}{i+2}         & \cdots & \frac{1}{j}       & \frac{1}{j+1}         & \frac{1}{j+2}         & \cdots & \frac{1}{d+1}          \\
    
    \vdots &        & \vdots            &                                                         & \vdots                &        & \vdots            &                       & \vdots                &        & \vdots                 \\
    
    1      & \cdots & \frac{1}{i^{d-2}} & \frac{1}{(i+1)^{d-2}}                                   & \frac{1}{(i+2)^{d-2}} & \cdots & \frac{1}{j^{d-2}} & \frac{1}{(j+1)^{d-2}} & \frac{1}{(j+2)^{d-2}} & \cdots & \frac{1}{(d+1)^{d-2}}  \\ \hline
    
    0      & \cdots & 0                 & \frac{(-1)^{i+j + {d-1 \choose 2} + \delta_{i > j}}}{n} & 0                     & \cdots & 0                 & 0                     & 0                     & \cdots & 0                      \\
    
    0      & \cdots & 0                 & 0                                                       & 0                     & \cdots & 0                 & -\frac{1}{n}          & 0                     & \cdots & 0                      \\
    
  \end{array}\right) \text{ for } (3) \to (4)
\end{equation} 

\begin{equation}
F^\text{tang} = 
  \left(\begin{array}{@{}ccccccccccc}
    1      & \cdots & 1                 & 1                                                       & 1                     & \cdots & 1                 & 1                           & 1                     & \cdots & 1                      \\
    
    1      & \cdots & \frac{1}{i}       & \frac{1}{i+1}                                           & \frac{1}{i+2}         & \cdots & \frac{1}{j}       & (1-t) \frac{1}{j+1}         & \frac{1}{j+2}         & \cdots & \frac{1}{d+1}          \\
    
    \vdots &        & \vdots            &                                                         & \vdots                &        & \vdots            &                             & \vdots                &        & \vdots                 \\
    
    1      & \cdots & \frac{1}{i^{d-2}} & \frac{1}{(i+1)^{d-2}}                                   & \frac{1}{(i+2)^{d-2}} & \cdots & \frac{1}{j^{d-2}} & (1-t) \frac{1}{(j+1)^{d-2}} & \frac{1}{(j+2)^{d-2}} & \cdots & \frac{1}{(d+1)^{d-2}}  \\ \hline
    
    0      & \cdots & 0                 & \frac{(-1)^{i+j + {d-1 \choose 2} + \delta_{i > j}}}{n} & 0                     & \cdots & 0                 & 0                           & 0                     & \cdots & 0                      \\
    
    0      & \cdots & 0                 & 0                                                       & 0                     & \cdots & 0                 & -\frac{1}{n}                & 0                     & \cdots & 0                      \\
    
  \end{array}\right) \text{ for } (4) \to \jhat_\Delta
\end{equation}

Here, $\delta_{i>j}$ is 1 if $i>j$ and 0 if $i<j$. One can check that this factor of $(-1)^{i+j + {d-1 \choose 2} + \delta_{i > j}}$ guarantees that the sign of the determinant of $F^\text{tang}$ matches that of $F^\text{bkgd}$. But for the case that the orientation vectors disagree, we want there to be a sign difference in their determinants, which is why there's an extra factor of $-1$ on the $(d-1)$th vector if the orientation vectors disagree. 

\subsubsection{Going between different $d$-simplices}
Let's now handle the case of going between different $d$-simplices. We'll only describe here the case when we don't pass across a representative of $w_1$, since the constructions in Section \ref{sec:TMdetTMFraming} apply straightforwardly to that case. 

First, we should note a few things. Say we are going from one $d$-simplex to another, $\Delta_1 \to \Delta_2$ along the boundary $(d-1)$-simplex $\Delta_\partial$. The labeling of the coordinates $\{0,\dots,d\}$ will be different on $\Delta_2$ versus the $\Delta_1$. And in general, the coordinates of $\Delta_2$ will be given by some permutation $s \in Perm(\{0,\dots,d\})$. Let's say we enter the first simplex on the edge $\ihat$ and exit on $\jhat$ (i.e. traversing $\ihat \to \jhat$), where `$i,j$' refers to the coordinate on the first simplex. Then, we can note that the branching structure on $\Delta_\partial$ induces the same partial ordering of $\{0,\dots,d\} \textbackslash \{j\}$ (as labeled by $\Delta_1$) on $\Delta_2$. This means the permutation $s$ will satisfy $s(\ell_1) < s(\ell_2)$ for each $\{\ell_1 < \ell_2\} \subset \{0,\dots,d\} \textbackslash \{j\}$.

With this in mind, we'll express the transition of the framing across the $\Delta_\partial$ in the coordinates of $\Delta_1$ while using the permutation $s$. In these coordinates, the last vector of all the framings, which is the one that runs `along' the 1-skeleton, will remain unchanged across the transition. Let's start with the background framing, which is simpler because it doesn't depend on which path is taken. We can start and denote $g_t[m] := (1-t) m + t s(m)$ for $m \in \{0,\dots,d\}$.

\begin{equation}
F^\text{bkgd} = 
  \left(\begin{array}{@{}ccccccc}
    1      & \cdots & 1                      & 1                 & 1                        & \cdots & 1                        \\
    1      & \cdots & \frac{1}{g_t[j]}       & 0                 & \frac{1}{g_t[j+2]}       & \cdots & \frac{1}{g_t[d+1]}       \\
    \vdots &        & \vdots                 & 0                 & \vdots                   &        & \vdots                   \\
    1      & \cdots & \frac{1}{g_t[j]^{d-2}} & 0                 & \frac{1}{g_t[j+2]^{d-2}} & \cdots & \frac{1}{g_t[d+1]^{d-2}} \\ \hline
    1      & \cdots & \frac{1}{g_t[j]^{d-1}} & 0                 & \frac{1}{g_t[j+2]^{d-1}} & \cdots & \frac{1}{g_t[d+1]^{d-1}} \\ 
    0      & \cdots & 0                      & \frac{(-1)^j}{n}  & 0                        & \cdots & 0                        \\
  \end{array}\right)
\end{equation}

Here, $t: 0 \to 1$ is again a parameter that represents moving across $\Delta_\partial$. And it's simple to see that this determinant is always nonzero, since $s(\ell_1) < s(\ell_2)$ for $\{\ell_1 < \ell_2\} \subset \{0,\dots,d\} \textbackslash \{j\}$.

For the tangent framing, there's only one vector we have to worry about, since it'll share the shared framing listed in $F^\text{bkgd}$, and the last vector will be unchanged. The only vector that we need to worry about is the $(n-1)$th vector $v^\text{tang}_{n-1}$. Say that on $\Delta_2$ that the curve traverses from $\widehat{s(j)} \to \widehat{s(k)}$. If $i\neq k$, we'll say it changes across $\Delta_\partial$ (using the coordinates on $\Delta_1$) as:

\begin{equation} \label{otherVectAcrossSimplices}
v^\text{tang}_{n-1} = (0, \cdots, 0, \underbrace{(1-t) \frac{(-1)^{i+j + {d-1 \choose 2} + \delta_{i > j}}}{n}}_{i\text{th component}}, 0, \cdots, \underbrace{0}_{j\text{th component}}, \cdots, 0, \underbrace{t \frac{(-1)^{k+j + {d-1 \choose 2} + \delta_{k > j}}}{n}}_{k\text{th component}}, 0, \cdots, 0)
\end{equation}

And if $i=k$, then this vector's components don't change. And, it's again simple to check that this vector keeps the matrix $F^\text{tang}$ nondegenerate the whole way. 

\subsection{Windings}
Now, let's compute the windings and verify that the total winding can be computed by breaking up the windings along each segment to match what we said in Fig(\ref{fig:trivalentResAndWinding}). We're explicitly consider the case of orientable manifolds where the extra `orientation vectors' from $\det(TM)$ will match for both framings, and at the end say how we can compute things for the other cases. And we'll first describe everything within in a `$+$' simplex.

First, let's think about what the winding means. The matrix $(F^\text{bkgd})^{-1}F^\text{tang}$ going around a loop will give us an element of $\pi_1(GL^+(d)) = \pi_1(SO(d)) = \Z_2$. So, we want to figure out a way to compute this element given our knowledge of the order of edges we pass through. Note that since the first $(d-2)$ vectors of $F^\text{bkgd}$ and $F^\text{tang}$ agree, all the winding will come from the winding in $\pi_1(SO(2)) = \Z$ of the bottom $2 \times 2$ block of $(F^\text{bkgd})^{-1}F^\text{tang}$.

Strictly speaking, the framings $F^\text{bkgd}$ and $F^\text{tang}$ themselves will not themselves wind with the windings in that Figure. The windings of $0$ or $\pm \pi$ imply that the two frames will rotate by a relative $\pm \mathbbm{1}$ with respect to each other on each leg of the journey, and this doesn't literally hold for $F^\text{bkgd}$ and $F^\text{tang}$. But, since we only care about the \textit{homotopy class} of the relative framing, we can deform $F^\text{bkgd}$ to some homotopic framing and measure the windings with respect to this deformed framing. Our deformation will involve changing the vector $v^\text{bkgd}_{n-1}$. This shouldn't be surprising, since for most of the path away from the centers of the $d$-simplices, the last vectors of the framings $v^\text{bkgd}_n, v^\text{tang}_n$ will be along the 1-skeleton. So we'll find it easier to deform $v^\text{bkgd}_{n-1}$ to some other vector that will make it straightforward to compute the relative winding.

Consider some other vector $\Tilde{v}^\text{bkgd}_{n-1}$ for which replacing $v^\text{bkgd}_{n-1}$ in $F^\text{bkgd}$, giving some matrix $\Tilde{F}^\text{bkgd}$ whose determinant always has the same sign as $F^\text{bkgd}$. Then, we can see that \textit{only} deforming $v^\text{bkgd}_{n-1}$ into $\Tilde{v}^\text{bkgd}_{n-1}$ as $(1-t) v^\text{bkgd}_{n-1} + t \Tilde{v}^\text{bkgd}_{n-1}$ around the loop will give us a homotopy between $F^\text{bkgd}$ and $\Tilde{F}^\text{bkgd}$. First, let's describe $\Tilde{v}^\text{bkgd}_{n-1}$ within a `$+$' $d$-simplex. We'll set:

\begin{equation}
\begin{split}
\Tilde{v}^\text{bkgd}_{n-1} &= (0, \cdots, 0, \underbrace{\frac{(-1)^{j+1+\delta_{i>j}}}{n}}_{j\text{th component}}, 0, \cdots, 0) \quad\quad\quad\quad\quad\quad\quad\quad\quad\quad\quad\quad\quad\quad\quad \text{for } \ihat_\Delta \to (1) \to (2) \\
                            &= (0, \cdots, 0, \underbrace{t \frac{(-1)^{i+1+\delta_{i>j}}}{n}}_{i\text{th component}}, 0, \cdots, 0, \underbrace{(1-t) \frac{(-1)^{j+1+\delta_{i>j}}}{n}}_{j\text{th component}}, 0, \cdots, 0) \quad \text{ for } (2) \to (3) \\
                            &= (0, \cdots, 0, \underbrace{\frac{(-1)^{i+1+\delta_{i>j}}}{n}}_{i\text{th component}}, 0, \cdots, 0) \quad\quad\quad\quad\quad\quad\quad\quad\quad\quad\quad\quad\quad\quad\quad \text{for } (3) \to (4) \to \jhat_\Delta 
\end{split}
\end{equation}

And, to move between different two different $d$-simplices, the procedure we used for $v^\text{tang}_{n-1}$ Eq(\ref{otherVectAcrossSimplices}) applies for $\Tilde{v}^\text{bkgd}_{n-1}$. And, it's straightforward to see that this vector will cause the determinant of the deformed framing $\Tilde{F}^\text{bkgd}$ to have the same sign as $F^\text{tang}$.

Now, we can read off the windings by comparing the last two vectors traversing through the `$+$' $d$-simplex. Note that the winding between $i$ and $j$ only depends on whether $(i-j)$ is even or odd and whether $i>j$. In particular, (up to a minus sign depending on our definition of clockwise), we'll have that between $i \to j$, we'll have have that there's nonzero winding iff $i \equiv j$ (mod 2) and the winding is $+\pi$ if $i>j$ and it's $-\pi$ if $i<j$. Similarly, we'll multiply these windings by a minus sign on a `$-$' simplex. 

These windings were for the orientable case. In general, we'll have on a `$+$' simplex that the winding matrix $W_k$ will be:

$$W_k = 
\begin{pmatrix} 
-i X & 0 \\ 0 & i X
\end{pmatrix} \quad \text{if} \quad i \equiv j \text{ (mod 2) and } i>j$$

$$W_k = 
\begin{pmatrix} 
i X & 0 \\ 0 & -i X
\end{pmatrix} \quad \text{if} \quad i \equiv j \text{ (mod 2) and } i<j$$

$$W_k = 
\begin{pmatrix} 
\mathbbm{1}& 0 \\ 0 & \mathbbm{1}
\end{pmatrix} \quad \text{otherwise}$$

And on a `$-$' simplex:

$$W_k = 
\begin{pmatrix} 
-i X & 0 \\ 0 & i X
\end{pmatrix} \quad  \text{if} \quad i \equiv j \text{ (mod 2) and } i<j$$

$$W_k = 
\begin{pmatrix} 
i X & 0 \\ 0 & -i X
\end{pmatrix} \quad \text{if} \quad i \equiv j \text{ (mod 2) and } i>j$$

$$W_k = 
\begin{pmatrix} 
\mathbbm{1}& 0 \\ 0 & \mathbbm{1}
\end{pmatrix} \quad \text{otherwise}$$

And, going across $w_1$ will again give us:

$$W_k = 
\pm \begin{pmatrix} 
0 & iY \\ -iY & 0
\end{pmatrix} \quad \text{crossing } w_1$$

Here, the matrix is the opposite sign in opposite directions. But the sign in a given direction is our choice, and these choices determine the chain representative of $w_1^2$ in the same way as they did for the 2D case.

\subsection{Trivalent Resolution and quadratic refinement for $\sigma$}
Now, we'll show that the trivalent resolution given in Fig(\ref{fig:trivalentResAndWinding}) is consistent with the higher-dimensional quadratic refinement. Remember that we were able to give a geometric argument for quadratic refinement for pairs of curves that share edges on the dual 1-skeleton. However, that analysis doesn't say anything about the case where the pair of curves only meet at a point in the middle of a $d$-simplex. So, we need to introduce a trivalent resolution so we can unambiguously split a given cochain's Poincaré dual into distinct loops. 

\begin{figure}[h!]
  \centering
  \includegraphics[width=\linewidth]{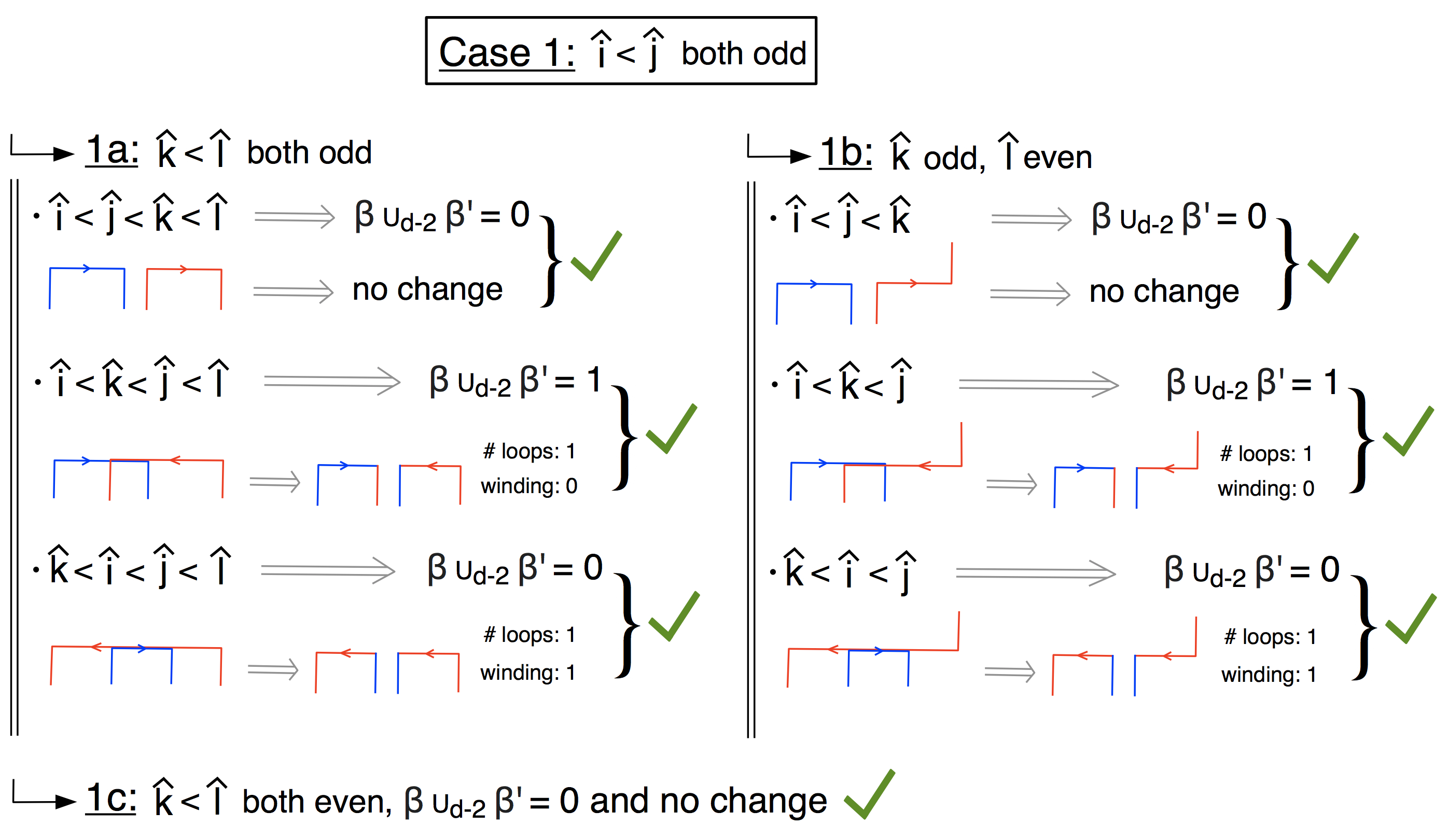}
  \caption{Proof that the quadratic refinement property of $\sigma$ holds for the trivalent resolution given in Fig(\ref{fig:trivalentResAndWinding}) for the case of $i,j$ both being odd. $\beta$ consists of the edges $\{\ihat,\jhat\}$ and $\beta'$ consists of the edges $\{\hat{k} \to \hat{\ell}\}$. Here, the edges pointing downward represent the edges $\hat{m}$ for $m$ odd and the edges pointing up represent $m$ even, like how we drew the `$+$' simplex in Fig(\ref{fig:trivalentResAndWinding}).  We also arrange the loops to be going in opposite directions to each other, which correspond to the Type I and II intersections of Fig(\ref{intersectionsAndWindings}). The other kinds of intersections can be handled using similar logic as in 2D. \\
  Here, ``winding" refers to the change in the winding angle as a multiple of $2\pi$ as given in Fig(\ref{fig:trivalentResAndWinding}), or equivalently using the winding matrices of the previous subsection. And ``loops" refers to whether the number of loops changes by 1. Using Eq(\ref{cupDminus2Formula}), $\beta \cup_{d-2} \beta'$ can be computed using the information of whether $i,j,k,l$ are even or odd and their order smallest-to-largest. If we compare this to how $(-1)^\text{\# of loops}(-1)^\text{winding}$ changes, we find that it agrees in all cases. (The other cases of $i<j$ both odd are equivalent to some of the cases listed here).}
  \label{fig:trivalentResProof}
\end{figure}

First, we will state an economical expression for $\beta \cup_{d-2} \beta'$ on a $d$-simplex $(0 \dots d)$. 

\begin{equation} \label{cupDminus2Formula}
(\beta \cup_{d-2} \beta')(0 \dots d) = \sum_{\substack{i < j \text{ both odd, OR} \\ i > j \text{ both even}}} \beta(\ihat) \beta'(\jhat)
\end{equation}

Where $\beta(\ihat)$ refers to $\beta(0,\dots,\ihat,\dots,d)$. Now, we want to verify the trivalent resolution and the windings will conspire to change $(-1)^\text{\# of loops} (-1)^\text{winding}$ change by $(\beta \cup_{d-2} \beta')(0 \dots d)$ on a $d$-simplex. Note that like in the 2D case, we'll only need to resolve one intersection at a time, so we can assume that $\beta$ consists of the edges $\{\ihat \to \jhat\}$ and $\beta'$ consists of the edges $\{\hat{k} \to \hat{\ell}\}$ with all of $i,j,k,l$ distinct. There will be many different cases that we'll need to check, corresponding to whether each of $i,j,k,l$ are even or odd and what their order is smallest-to-largest.

We'll explicitly state the cases of when $i<j$ and they're both odd, as in Fig(\ref{fig:trivalentResProof}), and leave the rest to the reader. There, we give the case of when the loops are running in opposite directions to each other along their intersections, which correspond to the Type I and II crossings of Fig(\ref{intersectionsAndWindings}). The case of Type III and IV windings can be dealt with using similar logic in this case as in the 2D case. Fig(\ref{intersectionsAndWindings}).

The same analysis as Fig(\ref{fig:trivalentResProof}) can be used to verify quadratic refinement in all the other possible cases.

\section{Combinatorially encoding a $Pin^+$ structure} \label{combDefinePinPlus}
We will now discuss how to combinatorially encode a $Pin^+$ structure on a manifold. (See Appendix \ref{spinPinAppendix} for details on the $Pin^+$ groups.) We will not attempt to give an analog of the function $\sigma(M,\alpha)$ for a $Pin^+$ structure, since we don't know whether such a notion would make sense. However, we will note that we can use entirely analogous ideas to our description of a $Pin^-$ to encode a $Pin^+$ structure. For $Pin^-$, recall that we explicitly trivialized $TM \oplus \det(TM)$ on the 1-skeleton given a choice of representative of $w_1$ and some choices of how to extend the trivialization across $w_1$. We were able to see explicitly how the obstruction to extending this framing to the dual 2-skeleton was homologous to the Poincaré dual of $w_2 + w_1^2$. Similarly, here we'll explicitly trivialize $TM \oplus 3 \det(TM)$ on the 1-skeleton and show that the obstruction to extending this framing to the 2-skeleton is homologous to the dual of $w_2$. 

Throughout the description of the winding matrices on a $Pin^-$ manifold, we used the coordinates `$x,y,z$', where the `$x$' direction typically referred to the $\det(TM)$ direction and the `$y,z$' direction typically referred to two directions on $TM$. For the $Pin^+$ case where we have three copies of $\det(TM)$, we can refer to the three directions typically pointing in the $3 \det(TM)$ direction as `$x_1,x_2,x_3$'. So in this case, the background and tangent framings will have three `orientation vectors'. And these vectors will remain undisturbed except when we cross a representative of $w_1$, at which point the three vectors all get reversed. So together, we'll have 5 directions $x_1,x_2,x_3,y,z$ that we care about.

Now, we'll rephrase the winding matrices in terms of the $\gamma$ matrices. There will still be the structure of organizing the winding in terms of a length-2 tuple based on whether the orientation vectors agree or disagree. Before, we were able to express our winding matrices in terms of the Pauli matrices, as $iX,iY,iZ$. But we could have just as well talked about them in terms of the $\gamma$ matrices. Recall that $iX$ represented a $180^\circ$ rotation of the $y,z$ axes. So we just as well could have replaced $iX \leftrightarrow \gamma_y\gamma_z$, and cyclic permutations thereof. So, the winding matrices for traversing within a $d$-simplex betweeen edges $\ihat \to \jhat$ can be rewritten on a `$\pm$' simplex as:

$$W_k = \pm
\begin{pmatrix} 
-\gamma_y \gamma_z & 0 \\ 0 & \gamma_y \gamma_z
\end{pmatrix} \quad \text{if} \quad i \equiv j \text{ (mod 2) and } i>j$$

$$W_k = \pm
\begin{pmatrix} 
\gamma_y \gamma_z & 0 \\ 0 & -\gamma_y \gamma_z
\end{pmatrix} \quad \text{if} \quad i \equiv j \text{ (mod 2) and } i<j$$

$$W_k = 
\begin{pmatrix} 
\mathbbm{1}& 0 \\ 0 & \mathbbm{1}
\end{pmatrix} \quad \text{otherwise}$$

Now, we should see what happens when crossing a representative of $w_1$. Before, the winding got multiplied by $\pm i Y \leftrightarrow \pm \gamma_z \gamma_x$, with the sign depending on if the orientation vectors start out as agreeing or disagreeing. Similarly, we'll have for this case that we multiply by $\pm (\gamma_z \gamma_{x_1})(\gamma_z \gamma_{x_2})(\gamma_z \gamma_{x_3}) = \mp \gamma_z\gamma_{x_1}\gamma_{x_2}\gamma_{x_3}$ with the sign depending on whether the orientation vectors start out agreeing or disagreeing, which we can organize as:

$$W_k = 
\pm \begin{pmatrix} 
0 & -\gamma_z\gamma_{x_1}\gamma_{x_2}\gamma_{x_3} \\ \gamma_z\gamma_{x_1}\gamma_{x_2}\gamma_{x_3} & 0
\end{pmatrix} \quad \text{crossing }w_1$$

Here, the choice of $\pm$ out front depends on what directions we choose the orientation vectors to point as they cross $w_1$, like we did in the $Pin^-$ case. Note that the square of this matrix is $-\mathbbm{1}$ as opposed to squaring to $+\mathbbm{1}$ in the $Pin^-$ case. But, this matrix still commutes with all other windings, as it did in the $Pin^-$. 

A crucial difference between the $Pin^+$ case and $Pin^-$ is that the matrix will be the same going in both directions, whereas the different directions differed by a sign in the $Pin^-$ case. This can be traced back to the fact that reflections square to different signs of the identity for the $Pin^\pm$ groups. This is the reason why we don't know if a definition of $\sigma(M,\alpha)$ makes sense for such structures - because traversing a loop going in one direction will give an opposite winding from going the other direction.

So since the choice of $\pm$ sign in front is the same in both directions, we'll say for simplicity that 

$$W_k = 
\begin{pmatrix} 
0 & -\gamma_z\gamma_{x_1}\gamma_{x_2}\gamma_{x_3} \\ \gamma_z\gamma_{x_1}\gamma_{x_2}\gamma_{x_3} & 0
\end{pmatrix} \quad \text{crossing }w_1$$

Now, we'll imitate the argument in Section \ref{formalPropsPinMinusSigma} to show that the obstruction to extending this framing to the 2-skeleton is homologous to a Poincaré dual of $w_2(M)$. Away from the $w_1$ wall, the argument is again the same as the orientable case. And again, the issue will just be looking at what happens along the representative of $w_1$. Say that going around the plaquette $P$ gives a total winding with $\mathcal{O}^\text{final} = W_k \cdots W_1 \mathcal{O}^\text{initial}$, where $i_0 < j_0$ are the segments crossing $w_1$. Then as before, we'll have that the sign of 

$$-(-1)^{j_0 - i_0 - 1} W_1 \cdots W_{i_0-1} W_{i_0 + 1} \cdots W_{j_0 - 1} W_{j_0 + 1} \cdots W_k$$

gives $(-1)^{\int_P w_2}$, where the extra minus sign out front can be thought of as analogous to the extra minus sign we attached to each loop in $\sigma(M,\alpha)$. 

So, we should examine what $(-1)^{j_0 - i_0 -1} W_{i_0} W_{j_0}$ does. Based on our choices of keeping all the windings across $w_1$ the same, we'll have that $W_{i_0}W_{j_0} = -\mathbbm{1}$, so that $(-1)^{j_0 - i_0 -1} W_{i_0} W_{j_0} = (-1)^{j_0 - i_0}$. Now we won't have in general that this sign is always $1$, so we should argue that the set of all $(d-2)$-simplices where the sign is $-1$ sum up to a homologically trivial set. To do this, we claim that $(-1)^{j_0 - i_0}$ gives $-1$ whenever the $(d-2)$-simplex corresponding to $P$ is part of the representative $w_1(PD(w_1(M)))$. Here, `$PD(w_1(M))$' stands for the the Poincaré dual of $w_1(M)$, so we are claiming that $(-1)^{j_0 - i_0} = -1$ iff the local orientations on the representative of $w_1(M)$ are opposite at $i_0,j_0$. It is known that $w_1(PD(w_1(M)))$ is trivial (i.e. that $PD(w_1(M))$ is an orientable manifold), so once we show the claim, we've demonstrated the $w_2$ obstruction for these structures.

We know that $(-1)^{j_0 - i_0}$ gives 1 if the tangent framing and background framing go through an odd number of $180^\circ$ turns while the curve is on one side of $w_1$, and gives $-1$ if there's an even number of such turns. Note that since the curve enters and exits the $w_1$ wall on opposite sides, the induced orientation of the tangent framing will be opposite at $i_0$ and $j_0$. This means that if the tangetn and background framings go through an odd number of $180^\circ$ turns, then the induced orientation from the background framing will induce equivalent orientations at $i_0$ and $j_0$. And an even number of such turns means that the induced orientations are the opposite. This tells us that $(-1)^{j_0 - i_0} = -1$ iff the background framing induces opposite orientations at $i_0,j_0$, which is the same as saying that the $(d-2)$-simplex corresponding to $P$ is part of the representative of $w_1(PD(w_1(M)))$.

This argument does not depend on the fact that our choice of winding matrices across $w_1$ were all the same. If we switched the sign of one of them, then the set of all plaquettes for which the winding going around is $-1$ will change by a coboundary, i.e. the set of all plaquettes adjacent to that edge would have the winding change by $-1$. 

We note that a version of quadratic refinement holds for these windings as they did for windings on $TM \oplus \det(TM)$. In particular, the same argument of quadratic refinement for resolving Type I and Type II crossings applies for $TM \oplus 3 \det(TM)$ just as it did for $TM \oplus \det(TM)$ in Section \ref{formalPropsPinMinusSigma}. First, Type I and II crossings that avoid $w_1$ will satisfy quadratic refinement automatically. Next, the same reason that crossings intersecting $w_1$ satisfy quadratic refinement carries over from Section \ref{formalPropsPinMinusSigma}, since the product of the winding matrices in opposite directions across $w_1$ is:

$$\begin{pmatrix} 
0 & -\gamma_z\gamma_{x_1}\gamma_{x_2}\gamma_{x_3} \\ \gamma_z\gamma_{x_1}\gamma_{x_2}\gamma_{x_3} & 0
\end{pmatrix} 
\begin{pmatrix} 
0 & -\gamma_z\gamma_{x_1}\gamma_{x_2}\gamma_{x_3} \\ \gamma_z\gamma_{x_1}\gamma_{x_2}\gamma_{x_3} & 0
\end{pmatrix} = -\mathbbm{1}$$

But, the fact that the windings are opposite in opposite directions means that we can't define $\sigma(M,\alpha)$ in the same way, and similar reasoning wouldn't apply to Type III or IV crossings.

And given this obstruction, we can encode a $Pin^+$ structure $\eta$ in the same way we did for $Spin$ and $Pin^-$ structures.

\bibliographystyle{unsrt} 


\end{document}